\documentclass[11pt,reqno]{article}
\usepackage[top=2cm, bottom=2cm, left=2.4cm, right=2.4cm]{geometry}
\usepackage{dsfont,amssymb,amsmath,amscd,latexsym,amsthm,amsxtra,amsfonts}
\allowdisplaybreaks[4]
\usepackage{lineno}
\usepackage[all]{xy}
\usepackage{enumitem}
\usepackage{tikz}
\usepackage[round]{natbib}
\usepackage{mathrsfs}
\usepackage{graphicx}
\usepackage{placeins}
\usepackage{subcaption}
\usepackage{comment}
\usepackage{mathtools}
\usepackage{cases}
\usepackage{tcolorbox}
\tcbuselibrary{most}
\usepackage{xcolor}
\usepackage{bbm}

\usetikzlibrary{calc,arrows}
\usepackage{verbatim}
\usepackage{epstopdf}
\usepackage{bm}
\usepackage[title]{appendix}

\newtheorem{theorem}{Theorem}[section]
\newtheorem{assumption}[theorem]{Assumption}

\newtheorem{corollary}[theorem]{Corollary}
\newtheorem{definition}[theorem]{Definition}

\newtheorem{lemma}[theorem]{Lemma}

\newtheorem{proposition}[theorem]{Proposition}
\newtheorem{remark}[theorem]{Remark}
\numberwithin{equation}{section}

\DeclareMathOperator{\Var}{Var}
\DeclareMathOperator{\Cov}{Cov}

\newcommand{\md}{\mathrm{d}}

\newcommand{\mR}{\mathbb{R}}

\newcommand{\mE}{\mathbb{E}}

\renewcommand{\epsilon}{\varepsilon}

\newcommand{\F}{\mathcal{F}}

  \usepackage[pdfstartview=FitH, bookmarksnumbered=true,bookmarksopen=true, colorlinks=true, pdfborder={0 0 1}, citecolor=blue, linkcolor=blue,urlcolor=blue]{hyperref}
\usepackage{graphics}
\graphicspath{{figures/}}

\usepackage[utf8]{inputenc}
\usepackage{newunicodechar}
\newunicodechar{，}{,}

\title{Mean-field game of mean-variance portfolio optimization  with peer-based risk aversion}

\author{Weilun Cheng\thanks{Department of Mathematical Sciences, Tsinghua University, Beijing, China. Email:\url{chengwl25@mails.tsinghua.edu.cn}} 
\and Zongxia Liang\thanks{Department of Mathematical Sciences, Tsinghua University, Beijing, China. Email:\url{liangzongxia@tsinghua.edu.cn}}
\and Sheng Wang\thanks{Department of Statistics and Data Science, The Chinese University of Hong Kong,
		Shatin, N.T., Hong Kong SAR.
        Email:\url{sheng-wa15@tsinghua.org.cn}}
\and Xiang Yu\thanks{Department of Applied Mathematics,  The Hong Kong Polytechnic University, Kowloon, Hong Kong. Email:\url{xiang.yu@polyu.edu.hk}}
}
\date{\vspace{-0.4in}}

\begin{document}

\maketitle

\begin{abstract}
This paper investigates a class of mean-field game (MFG) for mean-variance (MV) portfolio optimization, highlighting a new type of relative performance encoded by the peer-based risk aversion. Specifically, the risk aversion is formulated as a piecewise form that depends on whether the individual's wealth is above or below the population average, leading to a time-inconsistent MFG. Our goal is to seek a mean-field equilibrium, characterized by a forward-backward stochastic differential equation (FBSDE) system and a mean-field consistency condition. 
The new challenge stems from the discontinuous coefficients induced by the piecewise risk aversion. In response, 
we first introduce a smooth regularization technique to establish the existence of a solution to the discontinuous multidimensional FBSDE; this solution then yields the existence of an intra-personal equilibrium for the representative agent.
Finally, we conclude the existence of the mean-field equilibrium in the time-inconsistent MFG by invoking fixed-point arguments and convergence analysis as the smoothing regularization vanishes.

\vspace{0.1in}
\end{abstract}

\section{Introduction}\label{intro}

Hedge fund agents often evaluate their performance relative to their peers. The ubiquitous competition in the hedge fund industry suggests that an agent's optimal decision-making is dynamically shaped by the agent's relative standing in the financial market. An expanding strand of literature has formalized these relative performance concerns, primarily within the expected utility (EU) framework. \cite{espinosa2015optimal} considered a finite-player model, and \cite{lacker2019mean} further investigated the problem in a mean-field game setting with a continuum of agents. Subsequently, numerous studies have explored this direction under the EU framework. To name a few, \cite{Fu2023} studied mean-field portfolio games with random parameters via the FBSDE approach; \cite*{BWY24a} extended the MFG of portfolio management in the model with jump risk; \cite{Tangpi24} examined a class of graphon MFG with relative performance concerns; \cite{Zarip24} analyzed the master equation for MFG with general utility and unbounded controlled common noise; \cite*{BWY24b} addressed the external consumption habit formation as a relative performance MFG; \cite*{BHY25} studied a MFG of portfolio management by considering the relative performance incorporated in the risk measure of expected largest shortfall; \cite*{LLY25} studied a MFG with two heterogeneous populations with Poissonian common noise.

The shift of focus on relative performance from the EU framework to the MV criterion has occurred recently. \cite*{guan2022time} pioneered the integration of an MV objective with relative performance in both $N$-agent and mean-field games of investment and reinsurance. Subsequently, \cite*{guan2025n} investigated competition among multiple fund managers with relative performance measured by the excess logarithmic return. Later, \cite*{huang2026partial} studied a MV portfolio selection game under relative performance criteria, addressing the interplay of partial information. 
More recent studies on $N$-player games and MFGs with relative performance in the MV framework can be found in \cite*{shao2025competitive}, \cite*{huang2025mean}, and \cite*{guan2026robust}.

In the aforementioned problem formulations based on EU and MV, relative performance is typically modeled by incorporating the relative wealth, such as $X_i(T)-\bar{X}(T)$, into the agent's objective functional, where $X_i(T)$ denotes the agent's individual wealth and $\bar{X}(T)$ represents the average wealth of peers. By contrast, we adopt a distinct, behaviorally motivated approach: relative performance  influences the agent's risk aversion attitude. 
This new formulation is motivated by recent behavioral economics literature demonstrating that the individual risk attitude can be dynamically shaped by social comparisons such as her relative status or position. Specifically, \cite*{muller2019decisions} discussed that an individual's risk aversion is attenuated when lagging behind peers, whereas it is significantly amplified under favorable income inequality. Complementing this, \cite*{schwerter2024social} showed through laboratory experiments that decision-makers actively reduce risk-taking when ahead of their peers to secure their favorable relative earnings and avoid falling behind. Conversely, when placed in a disadvantageous position, individuals exhibit much riskier behavior in an effort to ``catch up" with their peers.
Motivated by these insights, our model assumes that risk aversion takes a piecewise form that depends on the comparison between the agent's own wealth level and the population average performance. Specifically, the objective functional evaluated at time 
$t$ is defined by
\begin{align}\label{eq:obj_func}
J(t,\pi;m):=\frac{1}{2}\textup{Var}_{t}[X^{\pi}(T)]-\gamma(t,X^{\pi}(t);m)\mathbb{E}_{t}[X^{\pi}(T)],
\end{align}
where $X^{\pi}$ is the wealth process under the investment strategy $\pi$, and $m$ represents a population aggregation as the mean-field interaction, which is treated as an exogenously given deterministic benchmark in the representative agent's control problem. In this MV formulation, the risk aversion coefficient is represented by $1/\gamma$. To reflect the asymmetric risk aversion attitudes by referring to the population average $m(t)$ as a psychological benchmark, our peer-based risk tolerance\footnote{In the objective functional \eqref{eq:obj_func}, $\gamma$ acts as the risk tolerance. Conversely, its reciprocal $1/\gamma$ represents the effective risk aversion coefficient. To avoid ambiguity, we distinguish these two terms throughout this paper, although both conceptually model the agent's risk attitude.} is formulated by
\begin{align*}
\gamma(t, x; m) &:= \gamma_{1}\mathbbm{1}_{x > m(t)} + \gamma_{2}\mathbbm{1}_{x \le m(t)},
\end{align*}
where we assume $\gamma_2 > \gamma_1 > 0$. That is, an investor exhibits higher risk aversion (i.e., more conservative behavior) when outperforming the benchmark ($x > m(t)$) and lower risk aversion (i.e., relatively greater risk tolerance) when falling below or equaling it ($x \le m(t)$).  This stepwise formulation captures how peer effects shape risk preferences in a direct manner and is supported by some recent theoretical and empirical studies in the literature. For instance, \cite*{ahern2014peer} empirically showed that individuals dynamically adjust their risk aversion to conform to their peers' average level, justifying our use of the population average \(m(t)\). Moreover, \cite*{cuoco2011equilibrium} observed that institutional contracts induce managers to react asymmetrically to these benchmarks. Falling behind the benchmark tends to drive managers to act as variance-minimizers with ``lower effective" risk aversion (\(1/\gamma_2\)), whereas securing a lead induces conservative hedging with ``higher effective" risk aversion (\(1/\gamma_1\)).
Motivated by these studies, we adopt the above discontinuous adjustment of the risk tolerance $\gamma$ in response to the relative standing in the population.

This formulation introduces two essential features: time-inconsistency and mean-field interactions. First, for a fixed $m$, it is well known that the variance term in the objective functional \eqref{eq:obj_func} violates the iterated expectation property, which makes the dynamic optimization problem time-inconsistent. In addition, the dependence of $\gamma$ on the initial state further contributes to this time-inconsistency. To address the issue of time-inconsistency, we adopt the consistent planning (intra-personal equilibrium) approach, originally proposed by \cite{Strotz1955}, which interprets the problem as a sequential intra-personal game among the decision-maker's temporal selves. This idea was first formulated in continuous time by \cite{ekeland2010golden} (a previous version of this paper: \cite{ekeland2006being}). In the context of the MV problem, \cite{basak2010dynamic} was the first to characterize the corresponding intra-personal equilibrium. Subsequent extensions include the study of state-dependent risk aversion in \cite*{bjork2014mean}, as well as time-inconsistent linear--quadratic (LQ) control problems in \cite*{hu2012time,Hu2017}. 
{In other general continuous-time problems exhibiting time-inconsistency, }
\cite*{Bjork2017}, \cite*{yan2019time}, \cite*{he2021equilibrium}, and \cite*{CLWY25} further investigated the existence and structural properties of equilibrium controls.

Second, the dependency of the risk aversion on $m$ necessitates analyzing the equilibrium strategy in the face of competition within a large population, leading to a new type of MFG when the mean-field interaction occurs in the risk aversion coefficient. This theory was independently initiated by \cite*{lasry2007mean} and \cite*{huang2006large}, with general mathematical and probabilistic frameworks expanded by \cite*{bensoussan2013mean} and \cite*{Carmona2018}. However, time-inconsistent MFG problems remain relatively under-explored. Existing studies often focus on specific LQ or MV models where the risk aversion parameter is typically constant or a linear function of wealth, which facilitates the derivation of explicit equilibria, as demonstrated by \cite*{ni2018, jun2021, Wang2023, liang2024time}. For more general time-inconsistent MFG problems, the existence of equilibria in various settings has been investigated in recent studies; see, e.g., 
\cite*{WZ2024,BW2025, zhou2026existence, BWYZ26}.
\vskip 2pt
Contrary to previous studies in time-inconsistent control problems and MFG problems, our formulation poses a new technical challenge due to the piecewise risk aversion parameter. To the best of our knowledge, this work appears to be the first study to handle the time-inconsistent MFG with discontinuous coefficients. In both the step of characterizing the intra-personal equilibrium for a fixed population aggregation and the step of seeking the  fixed point for the mean-field consistency $m$ (i.e., $m_*(t) = \mathbb{E}[X^{\bar{\pi}_*}(t)]$ in Definition \ref{def:MFE}), we have to develop some  tailor-made techniques in handling the discontinuous coefficients. To navigate these hurdles, we first establish a verification theorem that characterizes the intra-personal equilibrium via an associated FBSDE \eqref{eq:fbsde_sufficiency}. Building upon this characterization, our methodology proceeds in two main steps:
\begin{itemize}
    \item \textbf{Step 1: Intra-personal Equilibrium for the Representative Agent.}  
    The first step is to solve the FBSDE associated with the intra-personal equilibrium for a given population aggregation. Different methods have been developed to address the FBSDE in the literature. For instance, \cite*{ma1994solving} and \cite*{2002On} established partial differential equation (PDE)-based methods such as the well-known four-step scheme, and \cite*{ma2015well} developed a unified decoupling field approach. However, these methods often require  Lipschitz continuity conditions. Alternatively, \cite*{hu1995solution} introduced the method of continuation to relax some regularity requirements, which however imposes a strict monotonicity condition that our formulation fails to satisfy. Indeed, due to the piecewise risk tolerance, the main theoretical challenge in our setting arises from the discontinuous coefficients of the FBSDE.
    To surmount this hurdle, we first employ a mollification technique to construct regularized quasi-linear parabolic PDEs and obtain their smooth solutions. We establish the compactness of these regularized solutions and subsequently extract a strong solution to the original PDE with discontinuous coefficients by deriving energy estimates. Furthermore, we leverage this strong solution with Krylov's estimates to obtain the existence of a solution to the FBSDE.
    
    \item \textbf{Step 2: Mean-Field Consistency and Equilibrium.} We then seek a fixed point for $m$ such that the expected value of the individual wealth under the intra-personal equilibrium matches the population aggregation, thereby confirming the existence of a mean-field equilibrium. While most existing studies crucially rely on continuous model coefficients (e.g., \cite*{carmona2013probabilistic, Lacker2014Mean}), the inherent discontinuity in our setting renders these standard arguments inapplicable. To overcome the technical obstacle, we first establish the consistency for the regularized system by applying Schauder's fixed-point theorem. The core of our theoretical contribution then lies in a rigorous convergence analysis to successfully pass the consistency condition to the limit model with discontinuous coefficients. Specifically, we establish the tightness of the regularized auxiliary state processes and utilize the Gy\"ongy-Krylov lemma to guarantee their  convergence in probability. 
    Passing to the limit as the smoothing regularization vanishes, we recover the mean-field consistency condition in the original MFG problem.
\end{itemize}

In summary, there are three main contributions of the present work, resolving the intertwined challenges of time-inconsistency, discontinuous coefficients, and mean-field consistency. First, from a modeling perspective, we incorporate a social comparison-based risk preference to capture peer-based risk aversion, leading to a new class of time-inconsistent MFGs.
Second, from a methodological standpoint, we develop a smooth regularization approach to solve the associated FBSDE, where discontinuities arise in the drift coefficients of both the forward and backward equations, as well as in the diffusion coefficient of the forward equation.
Third, we establish the mean-field equilibrium in the original system by first proving the existence of equilibrium for the regularized MFG and subsequently deploying a delicate probabilistic convergence analysis to verify that the limit indeed fulfills the mean-field consistency condition when the smoothing regularization vanishes.  

\vskip 3pt
The remainder of this paper is organized as follows. Section \ref{Problem Formulation} formulates the model, introduces the peer-based risk tolerance to depict the relative performance, and defines the  mean-field equilibrium. Section \ref{Equilibrium Strategy for a Given Mean Field} tackles the existence of intra-personal equilibrium for a representative agent by analyzing the associated FBSDE with discontinuous coefficients using the smoothing regularization techniques. Section \ref{Mean-field Consistency and Equilibrium} further establishes the mean-field consistency by utilizing the convergence analysis as the smoothing effect vanishes, thereby confirming the existence of mean-field equilibrium in the original time-inconsistent MFG. Some proofs of auxiliary results are collected in Appendices \ref{sec:App}, \ref{sec:AppB}, \ref{sec:AppC}, respectively.

\ \\
\textbf{Notations}: We first present some notations that will be used frequently throughout the paper. Let $T>0$ be a finite time horizon and  $\left(\Omega,\mathcal{F},\mathbb{F},\mathbb{P}\right)$  a complete filtered probability space satisfying the usual conditions. It is assumed that the space supports a standard $d$-dimensional Brownian motion $W=\{W(t):= \bigl(W_{1}(t),\cdots,W_{d}(t)\bigr)^{\top}, 0\leq t\leq T\}$, and an essentially bounded, atomless\footnote{A $\mathbb{R}$-valued random variable $\xi$ is called atomless if, for any $x\in\mR$, $\mathbb{P}(\xi=x)=0$.} and $\mathbb{R}$-valued random variable $\xi$, which is independent of $W$.
Let
$\mathbb{F}:=\left\{\mathcal{F}_{t}\right\}_{0\leq t\leq T}$ be the filtration generated by $W$ and $\xi$, augmented by all null sets. Moreover, $\F=\F_T$.

We denote by $\mathbb{R}^n$ the $n$-dimensional Euclidean space equipped with the standard inner product $\langle \cdot, \cdot \rangle$ and the Euclidean norm $|\cdot|$. For a matrix $A$, $A^{\top}$ denotes its transpose. Let $Q_T = [0,T] \times \mathbb{R}$ be the spatiotemporal domain,
and let  $\mE_t[\cdot] := \mathbb{E}[\cdot \mid \mathcal{F}_t]$ be the conditional expectation given the $\mathcal{F}_t$. Moreover, we also adopt the following notations:

\noindent$\bullet$  $L_{\mathcal{F}_{t}}^{p}(\Omega; \mathbb{R}^n)$ ($1\leq p\leq \infty$) denotes the set of all $\mathbb{R}^n$-valued, $L^p$ integrable and $\mathcal{F}_t$-measurable random variables;  $L_{\mathbb{F}}^{2}(0, T; \mathbb{R}^n)$ denotes the set of all $\mathbb{R}^n$-valued, $\mathbb{F}$-progressively measurable processes such that $\mathbb{E}[ \int_0^T |f(t)|^2 dt ] < \infty$; $L_{\mathbb{F}}^{2}(\Omega; C([0, T]; \mathbb{R}^n))$ denotes  the set of all continuous, $\mathbb{F}$-adapted processes such that $\mathbb{E}[ \sup_{t \in [0,T]} |f(t)|^2 ] < \infty$.

\noindent$\bullet$ 
$L^p$ space and Sobolev spaces: For any $N\in\mathbb{N}$, any Lebesgue measurable set $S\subset \mathbb{R}^N$, and any $1\leq p\leq \infty$,
$L^p(S)$  denotes the set of all $\mathbb{R}$-valued, $L^p$ integrable and measurable   functions on $S$; $L_{\mathrm{loc}}^p(S)$ refers to locally $p$-integrable functions, i.e., $u \in L^p(K)$ for any compact subset $K \subset S$;  $W_{p}^{2}(\mathbb{R})$  is the space of functions $u \in L^p(\mathbb{R})$ with  weak  derivatives $\frac{\partial u}{\partial x}, \frac{\partial^2 u}{\partial x^2} \in L^p(\mathbb{R})$; $W_{p}^{1,2}([0,T] \times \mathbb{R})$ ($1\leq p \le \infty$) consists of functions $u \in L^p([0,T] \times \mathbb{R})$ whose weak derivatives $\frac{\partial u}{\partial t}, \frac{\partial u}{\partial x}, \frac{\partial^2 u}{\partial x^2}$ belong to $L^p([0,T] \times \mathbb{R})$; the local spaces
$W^2_{p,\mathrm{loc}}$ and
$W_{p,\mathrm{loc}}^{1,2}([0,T] \times \mathbb{R})$ are defined analogously by requiring the corresponding Sobolev regularity on every compact subset, in the same way as $L_{\mathrm{loc}}^p(S)$.

\noindent$\bullet$  $C([0,T];\mathbb{R}^n)$ denotes continuous functions with supremum norm $\|f\|_\infty := \sup\limits_{t \in [0,T]} |f(t)|$; for $n=1$, we write $C([0,T])$. $C_c([0,T] \times \mathbb{R})$ denotes continuous functions with compact support.  $C^{1,2}([0,T] \times \mathbb{R})$ consists of functions continuously differentiable in $t$ and twice in $x$, while $C^\infty([0,T] \times \mathbb{R})$ denotes infinitely differentiable functions.

\noindent$\bullet$  H\"older space: For $\alpha\in(0,1)$,
 $C^\alpha(\mathbb{R})$ denotes bounded, globally $\alpha$-H\"older continuous functions;  $C^{2+\alpha}(\mathbb{R})$ consists of  $C^2$ functions whose derivatives up to second order are bounded and in $C^\alpha(\mathbb{R})$; $C_{\mathrm{loc}}^{\alpha/2, \alpha}([0,T] \times \mathbb{R})$ denotes functions locally $\alpha/2$-H\"older continuous in $t$ and locally $\alpha$-H\"older continuous in $x$.

\noindent$\bullet$ Banach-space-valued functions: For a Banach space $X$, $B([0,T]; X)$ denotes the space of measurable and bounded functions equipped with $\|f\|_{B} := \sup\limits_{t \in [0,T]} \|f(t)\|_X$. $L^p(0,T; X)$ ($1\leq p \le \infty$) denotes the space of measurable functions $f: [0,T] \to X$ such that $\|f(\cdot)\|_X \in L^p([0,T])$.

\section{Market Model and Problem Formulation}\label{Problem Formulation}

We consider a large population of agents indexed by $k\in\mathbb{N}$. On a product extension of the probability space, let $\{(\xi^k,W^k)\}_{k\geq 1}$ be independent copies of $(\xi,W)$,  then the agents' $d$-dimensional Brownian motions are mutually independent and there is no common noise. It is assumed that each agent has access to an individual copy of the investment opportunity set, consisting of one risk-free asset (bank account) and $d$ risky assets (stocks), with identical deterministic market coefficients across agents. The price of the risk-free bond $S_{0}(\cdot)$ satisfies $\md S_{0}(s) = r(s)S_{0}(s) \md s$ with $S_{0}(0) = s_{0} > 0$, where $r(\cdot) > 0$ is the deterministic risk-free return rate. It is assumed that $r$ is continuous on $[0,T]$ and bounded by $r_{\max}$. The price process $S_{i}^k$, $i=1,\cdots,d$, available to agent $k$, follows the dynamics
\begin{align*}
		 \md S_{i}^k(t)=S_{i}^k(t)\left[\mu_{i}(t)\md t+\sigma_{i}(t) \md W^k(t)\right],\quad t\in[0,T],\ \ i=1,\cdots,d.
\end{align*}
where the market coefficients $\mu:[0,T]\rightarrow\mathbb{R}^{d}$ and $\sigma:[0,T]\rightarrow\mathbb{R}^{d\times d}$ are continuous and deterministic;  $\sigma_i$ denotes the $i$-th row of $\sigma$. Moreover, there are two positive constants $\sigma_{\min}$ and $\sigma_{\max}$ such that 
\begin{equation}\label{c}
	\sigma_{\min}\lVert\alpha\rVert^{2}\leq\lVert\sigma^{\top}(t)\alpha\rVert^{2}\leq \sigma_{\max}\lVert\alpha\rVert^{2}\quad \text{for all }\alpha\in\mathbb{R}^{d}\,\,\text{and}\,\,t\in[0,T].
\end{equation}
Let $\theta(\cdot):=\mu(\cdot)-r(\cdot)\mathbf{1}_d$ be the \emph{risk premium} and  $\lambda(\cdot) := (\sigma(\cdot))^{-1}\theta(\cdot)$ be the \emph{market price of risk},  where $\mathbf{1}_d=(1,\ldots,1)^\top\in\mathbb{R}^d$ is the $d$-dimensional vector of ones.

We work directly in the mean-field model. Let $(\xi,W)$ denote the initial wealth and idiosyncratic Brownian motion of a representative agent, and let $\pi(t)$ be the dollar amount that this agent invests in the risky assets at time $t$.
The wealth process $X^{\pi}(\cdot)$ of the representative agent evolves as
\begin{equation}\label{wealthdynamic}
	\left\{
	\begin{aligned}
		&dX^{\pi}(s)=r(s)X^{\pi}(s)\md s+\pi^{\top}(s)\theta(s)\md s+\pi^{\top}(s)\sigma(s)\md W(s),\\	
		&X^{\pi}(0)=\xi.
	\end{aligned}
	\right.
\end{equation}
It is straightforward to verify that for any $\pi(\cdot) \in L_{\mathbb{F}}^{2}(0, T; \mathbb{R}^d)$, \eqref{wealthdynamic} admits a unique strong solution satisfying $X^{\pi}(\cdot) \in L_{\mathbb{F}}^{2}(\Omega; C([0, T];\mathbb{R}))$.

\begin{remark}
In the mean field model, assuming the initial state as an atomless random variable $\xi$ naturally captures the heterogeneity of wealth distribution across this population. Furthermore, the assumption that $\xi$ is essentially bounded is by no means restrictive, which reflects the real-life situation that the initial wealth of agents must be finite.
\end{remark}

For a given deterministic continuous function $m \in C([0,T];\mR)$, let us define the relative state-dependent risk tolerance function as
\begin{align*}
\gamma(t, x; m) &:= \gamma_{1}\mathbbm{1}_{x > m(t)} + \gamma_{2}\mathbbm{1}_{x\leq m(t)},
\end{align*}
where $\gamma_2 > \gamma_1 > 0$ are constants. The objective functional of the representative agent evaluated at time $t$ is defined through the MV criterion with peer-based  risk aversion as
\begin{align*}
J(t,\pi;m):=\frac{1}{2}\textup{Var}_{t}[X^{\pi}(T)]-\gamma(t,X^{\pi}(t);m)\mathbb{E}_{t}[X^{\pi}(T)].
\end{align*}
Note that the effective risk aversion is given by $1/\gamma(t,x;m)$. The asymmetric specification captures the behavioral phenomenon: each agent exhibits lower effective risk aversion ($1/\gamma_2$) when underperforming their peers ($x \le m(t)$), while becoming more risk-averse ($1/\gamma_1$) when securing a leading position ($x > m(t)$) in the population. 

In contrast to existing relative-performance MFG, our mean-field interaction enters the state-dependent risk tolerance term. That is, instead of altering the expected return or variance, peer-based relative performance directly shapes each agent's risk aversion attitude. Interestingly, due to the discontinuity of $\gamma$ and the inherent time-inconsistency in MV criterion, we encounter a class of time-inconsistent MFG problems with discontinuous coefficients, new to the literature. Let us first give the following definition of the open-loop mean-field equilibrium.

\begin{definition}[Mean-field Equilibrium]\label{def:MFE}
A pair $(\bar{\pi}_{*},m_{*})$, where $\bar{\pi}_{*}=\bar{\pi}(\cdot;m_{*})\in L_{\mathbb{F}}^{2}(0,T;\mathbb{R}^{d})$ and $m_* \in C([0,T])$, is called an open-loop mean-field equilibrium if the following conditions are satisfied:

\begin{enumerate}
    \item[(a)] (Equilibrium strategy in the intra-personal game by the representative agent)
    For any $t \in [0, T)$ and $\eta \in L_{\mathcal{F}_{t}}^{2}(\Omega; \mathbb{R}^d)$, it holds that
    \begin{align}\label{eq:def:intra}
        \liminf_{\epsilon \downarrow 0} \frac{J(t,\pi^{t,\epsilon,\eta}(\cdot;m_{*});m_{*})-J(t,\bar{\pi}(\cdot;m_{*});m_{*})}{\epsilon} &\ge 0, \quad \mathbb{P}\text{-a.s.},
    \end{align}
    where
    \begin{align*}
        \pi^{t,\epsilon,\eta}(s;m_{*}):=\bar{\pi}(s;m_{*})+\eta \mathbbm{1}_{[t,t+\epsilon)}(s), \quad s \in [0, T].
    \end{align*}

    \item[(b)] (Consistency condition on population aggregation) $m_{*}(t)=\mathbb{E}[X^{\bar{\pi}_{*}}(t)]$ for all $ t \in [0, T]$.
\end{enumerate}
\end{definition}

Accordingly, our method to look for the mean-field equilibrium in Definition~\ref{def:MFE} is split into two main steps:

\medskip

\noindent
\textbf{Step 1 (Intra-personal equilibrium  for the representative agent).}
For a given deterministic function $m \in C([0,T])$ modeling the population aggregation, 
we focus on a representative agent and shall call  
$\bar{\pi}(\cdot; m) \in L_{\mathbb{F}}^{2}(0, T; \mathbb{R}^d)$ described by \eqref{eq:def:intra} in condition (a) of Definition~\ref{def:MFE} \textit{an intra-personal equilibrium}. The goal of this step is to determine such intra-personal equilibrium together with the associated state process $X^{\bar{\pi}(\cdot; m)}$.

\medskip

\noindent
\textbf{Step 2 (Mean-field consistency).}
We then tackle the MFG problem and seek a function $m_* \in C([0,T])$ as the fixed point such that the mean-field consistency condition (b) in Definition~\ref{def:MFE} holds:
\begin{align*}
    m_*(t)=\mathbb{E}\left[X^{\bar{\pi}(\cdot; m_*)}(t)\right]\quad \text{for all } t\in[0,T].
\end{align*}

\section{Characterization of Intra-personal Equilibrium}\label{Equilibrium Strategy for a Given Mean Field}
In this section, we fix a deterministic population aggregation $m \in C([0,T])$ and first establish the existence and characterization of an intra-personal equilibrium $\bar{\pi}(\cdot; m) \in L_{\mathbb{F}}^{2}(0, T; \mathbb{R}^d)$  for the representative agent under game-theoretic thinking with the agent's future selves. 

\subsection{FBSDE characterization of the Intra-personal Equilibrium}
This subsection first provides an equivalent FBSDE characterization of an intra-personal equilibrium for the representative agent. For notational simplicity, we suppress the dependence on $m$ whenever no confusion arises; for instance, we write $\gamma(t, X(t); m)$ as $\gamma(t, X(t))$, and $\overline{\pi}(\cdot; m)$ as $\overline{\pi}(\cdot)$. Furthermore, in the sequel, we denote $X^{\bar{\pi}}$ simply by $X$.

The following theorem provides a necessary condition for a strategy to be an intra-personal equilibrium. Its proof  is delegated to Appendix \ref{Proof of Necessary}.

\begin{theorem}\label{necessary}
Suppose that $\bar{\pi}\in L_{\mathbb{F}}^{2}(0, T; \mathbb{R}^d)$ is an intra-personal  equilibrium. Let the adjoint process $(P,Q)\in L_{\mathbb{F}}^{2}(\Omega; C([0, T];\mathbb{R})) \times L_{\mathbb{F}}^{2}(0, T; \mathbb{R}^d)$ satisfy the following backward stochastic differential equation (BSDE):
\begin{align}\label{eq:bsde}
\begin{cases}
\md P(t) = -\left( \gamma(t, X(t)) |\lambda(t)|^2 - \lambda^{\top}(t)Q(t) \right) \md t + Q^{\top}(t) \md W(t), \\
P(T) = 0.
\end{cases}
\end{align}
Then we have the characterization that, $\md t \otimes \md \mathbb{P}$-a.e. on $[0, T] \times \Omega$,
\begin{align}\label{eq:barU}
\bar{\pi}(t) = -\frac{1}{P_0(t)}(\sigma^\top(t))^{-1}\left[ Q(t) - \gamma(t, X(t))\lambda(t) \right],
\end{align}
where $P_0(t):=\exp\{\int_t^Tr(s)\md s\}$.
\end{theorem}
\vskip 2pt \noindent
According to Theorem \ref{necessary}, if $\bar{\pi} \in L_{\mathbb{F}}^{2}(0, T; \mathbb{R}^d)$ is an intra-personal equilibrium strategy,  there exists an adjoint process $(P,Q) \in L_{\mathbb{F}}^{2}(\Omega; C([0, T];\mathbb{R}))\times L_{\mathbb{F}}^{2}(0, T; \mathbb{R}^d)$ such that $(X, P, Q)$ satisfies the FBSDE system:
\begin{align}\label{eq:fbsde_sufficiency}
\begin{cases}
\md X(t) = \left( r(t)X(t) - \frac{1}{P_0(t)} \left[ Q(t) - \gamma(t, X(t))\lambda(t) \right]^{\top} \lambda(t) \right) \md t \\
\qquad \quad\quad\quad - \frac{1}{P_0(t)} \left[ Q(t) - \gamma(t, X(t))\lambda(t) \right]^{\top} \md W(t), \\
\md P(t) = -\left( \gamma(t, X(t)) |\lambda(t)|^2 - \lambda^{\top}(t)Q(t) \right) \md t + Q^{\top}(t) \md W(t), \\
X(0) = \xi, \quad P(T) = 0.
\end{cases}
\end{align}
The next result claims that if $X$ is atomless\footnote{The stochastic process $X$ is called atomless if, for any $x\in\mR$ and $t\in[0,T)$, $\mathbb{P}(X(t)=x)=0$.}, the converse also holds, and its  proof is given in Appendix \ref{Proof of Sufficiency}.

\begin{theorem}\label{sufficiency}
Suppose that FBSDE \eqref{eq:fbsde_sufficiency} admits a solution $(X, P, Q) \in L_{\mathbb{F}}^{2}(\Omega; C([0, T];\mathbb{R}))\times L_{\mathbb{F}}^{2}(\Omega; C([0, T];\mathbb{R}))\times L_{\mathbb{F}}^{2}(0, T; \mathbb{R}^d)$. 
{Assume further that $X$ is atomless.}
 Let $\bar{\pi}(\cdot)$ be defined by \eqref{eq:barU}.
Then $\bar{\pi}(\cdot)$ is an intra-personal equilibrium.
\end{theorem}

\vskip 3pt \noindent
Theorems \ref{necessary} and \ref{sufficiency} state that, provided $X$ is atomless, $\bar{\pi}$ is an intra-personal equilibrium if and only if it satisfies \eqref{eq:barU}, where $(X, P, Q)$ is a solution to the FBSDE \eqref{eq:fbsde_sufficiency}. In the next subsection, we investigate the existence of solutions to the FBSDE \eqref{eq:fbsde_sufficiency} and verify that the associated state process is indeed atomless.

\subsection{Solving the FBSDE  \eqref{eq:fbsde_sufficiency}}
In this subsection, to cope with the FBSDE \eqref{eq:fbsde_sufficiency},  we adopt a PDE approach and transform the problem into the study of the following quasi-linear parabolic PDE:
\begin{align}\label{eq:pde}
\begin{cases} 
\dfrac{\partial u}{\partial t}(t, x) + \dfrac{1}{2}|\lambda(t)|^2 \gamma \left( t, \dfrac{u(t, x)}{P_0(t)}\right)^2 \dfrac{\partial^2 u}{\partial x^2}(t, x) - |\lambda(t)|^2 \gamma \left( t, \dfrac{u(t, x)}{P_0(t)}\right) \dfrac{\partial u}{\partial x}(t, x) = 0, \\ 
u(T, x) = x.
\end{cases}
\end{align}
That is,  if the PDE \eqref{eq:pde} admits a solution $u$ with sufficient regularity and monotonicity, the existence of a solution to the  FBSDE \eqref{eq:fbsde_sufficiency} can be guaranteed. Indeed, define  the \emph{auxiliary process} $\{x(t): t\leq T\}$ satisfying $u(0, x(0)) = P_0(0)\xi$ and 
\begin{align}\label{eq:auxiliary process}
\md x(t) = \gamma\big(t, P_0^{-1}(t)u(t, x(t))\big) \lambda(t)^{\top} \md W(t)
\end{align}
as well as the processes $ X$, $P$ and $Q$:
  $X(t) := P_0^{-1}(t)u(t, x(t))$, $P(t) := x(t) - u(t, x(t))$ and $Q(t) := \gamma\big(t, X(t)\big) \big(1 - \frac{\partial u}{\partial x}(t,x(t))\big)\lambda(t)$, $ \ t\leq T$. A direct application of It\^o's formula shows that the triplet process  $(X, P, Q)$ solves \eqref{eq:fbsde_sufficiency}.

\subsubsection{Existence of Solution to the  PDE \eqref{eq:pde}}
In this subsection, we focus on the existence of solutions to the PDE \eqref{eq:pde}. Due to the discontinuity of $\gamma$, the solutions are generally expected to be in the Sobolev space $W_{2,\mathrm{loc}}^{1,2}(Q_T)$. To proceed, we impose the following assumption on $\lambda$.

\begin{assumption}\label{assumption-1}
There exist two positive constants $\lambda_{\min}$ and $\lambda_{\max}$ such that $\lambda_{\min}\leq |\lambda(t)|\leq \lambda_{\max}$ for all $t\in[0,T]$.
\end{assumption}\noindent
Let us define $ R(\cdot):=P_0(\cdot)m(\cdot)$ and
\begin{align*}
 &D(t, y) := \frac{1}{2}|\lambda(t)|^2 \gamma\left(t, \frac{y}{P_0(t)}\right)^2,\quad (t,y)\in Q_T,\\
&V(t, y) := |\lambda(t)|^2 \gamma\left(t, \frac{y}{P_0(t)}\right),\quad  (t,y)\in Q_T.
\end{align*}
Under Assumption \ref{assumption-1}, let  $\kappa := \frac{1}{2}\lambda_{\min}^2\gamma_1^2$, $\Lambda := \frac{1}{2}\lambda_{\max}^2\gamma_2^2$, and $M := \lambda_{\max}^2\gamma_2$.
Then 
\begin{align*}
 \kappa \le D(t, y) \le \Lambda \, \,\,\,\text{and}\,\,\,\, |V(t, y)| \le M \quad \text{for all } (t,y)\in Q_T.
\end{align*}
Moreover, let  $P_{\min}$ and $P_{\max}$  be two constants such that $0 < P_{\min} \le P_0(t) \le P_{\max} < \infty$ for all $t\in[0,T]$. The existence of a solution to the PDE \eqref{eq:pde} is established in the next result.
\begin{theorem}\label{thm: solving pde}
Under Assumption \ref{assumption-1}, the PDE \eqref{eq:pde} has a strong solution $u \in W_{2,\mathrm{loc}}^{1,2}(Q_{T}) \cap C^{\alpha/2, \alpha}_\mathrm{loc}(Q_{T})$ for any $\alpha < 1/2$. Furthermore, this solution admits the following properties:
\begin{enumerate}[label=(\arabic*), font=\upshape]
    \item  $|u(t, x) - x| \le M(T - t)$ holds for all $(t, x) \in Q_T$.
    \item There exist two positive constants $\delta_1$ and $M_1$ independent of  $m$ such that $\delta_1 \le \frac{\partial u}{\partial x} \le M_1$  a.e. on $Q_T$. Moreover, for each fixed $t \in [0, T]$, the mapping $x \mapsto u(t, x)$ is a globally bi-Lipschitz continuous bijection on $\mathbb{R}$ (i.e., $\delta_1 |x_1 - x_2| \le |u(t, x_1) - u(t, x_2)| \le M_1 |x_1 - x_2|$ for all $x_1, x_2 \in \mathbb{R}$).
    \item $u(t, x) = P_0(t)m(t)$ uniquely determines a globally continuous \emph{curve of separation boundary}  $x = b^*(t)$ with $b^*(\cdot) \in C([0, T])$.
\end{enumerate}
\end{theorem}

\begin{remark}
In view of \eqref{eq:auxiliary process}, the diffusion coefficient of $x(t)$, defined by
\begin{align}\label{eq:hat:sigma}
\hat{\sigma}(t,x) := \gamma(t, P_0^{-1}(t)u(t,x))\lambda(t),\quad (t,x)\in Q_T,
\end{align}
takes the form of $\lambda(t)$ multiplied by a function that takes only two values. The curve $x = b^*(t)$ obtained in Theorem~\ref{thm: solving pde}(3), naturally divides the domain $Q_T$ into two regions, where $\gamma$ takes $\gamma_1$ and $\gamma_2$, respectively.
\end{remark}
\vskip 5pt
The proof of Theorem~\ref{thm: solving pde} consists of two steps: In {\bf Step 1}, to overcome the issue of the discontinuous coefficients in the original equation, we introduce a regularized PDE that admits a global classical solution $\{u^\epsilon\}$ with bounded first derivatives and establish its strict spatial monotonicity based on  a smoothing regularization technique. In {\bf Step 2}, we obtain the strong solution of the PDE \eqref{eq:pde}  and its continuous curve of separation boundary based on the approximating sequence $\{u^\epsilon,\epsilon>0\}$ of solutions. 

\vskip 5pt
\textbf{Step 1: Smoothing Regularization and Strict Spatial Monotonicity.}
To overcome the issue of discontinuity, let $\epsilon > 0$ and consider a standard mollifier $\omega_\epsilon(\cdot) := \frac{1}{\epsilon}\omega\left(\frac{\cdot}{\epsilon}\right)$, where $\omega(\cdot)$ is a smooth, non-negative, compactly supported  function with $\int_{\mathbb{R}} \omega(y) \md y = 1$. Define the regularized function $\gamma^\epsilon$ by convolution:
\begin{align*}
\gamma^\epsilon(t, y) := \int_{\mathbb{R}} \gamma(t, z) \omega_\epsilon(y - z) \md z,\quad (t,y)\in Q_T,
\end{align*}
and consider the regularized coefficients that
\begin{align*}
D^\epsilon(t, y) &:= \frac{1}{2}|\lambda(t)|^2 \gamma^\epsilon\left(t, \frac{y}{P_0(t)}\right)^2,\quad (t,y)\in Q_T, \\
V^\epsilon(t, y) &:= |\lambda(t)|^2 \gamma^\epsilon\left(t, \frac{y}{P_0(t)}\right),\quad (t,y)\in Q_T.
\end{align*}
Then $D^{\epsilon}\in C(Q_T)$ and $V^\epsilon\in C(Q_T)$ are smooth in the second argument and satisfy
\begin{align}\label{eq:epsilon:bound}
\kappa \le D^\epsilon(t, y) \le \Lambda \quad\text{and}\quad
|V^\epsilon(t, y)| \le M \quad \text{for all }(t,y)\in Q_T.
\end{align}
Furthermore, $D^{\epsilon}(t, \cdot)$ and $V^{\epsilon}(t, \cdot)$ are globally Lipschitz continuous with respect to the second argument, with Lipschitz constants given by
\begin{align*}
L_{D, \epsilon} := \frac{\lambda_{\max}^2 \gamma_2}{P_{\min}} L_{\gamma, \epsilon}, \quad L_{V, \epsilon} := \frac{\lambda_{\max}^2}{P_{\min}} L_{\gamma, \epsilon},
\end{align*}
where $L_{\gamma, \epsilon} := \frac{\gamma_2}{\epsilon} \int_{\mathbb{R}} |\omega'(z)| \md z$ is the Lipschitz constant of $\gamma^{\epsilon}(t, \cdot)$. These constants are independent of $m$.
Consider the regularized PDE:
\begin{align}\label{eq:regularized}
\left\{
\begin{aligned}
&\frac{\partial u^\epsilon}{\partial t} + D^\epsilon(t, u^\epsilon) \frac{\partial^2 u^\epsilon}{\partial x^2} - V^\epsilon(t, u^\epsilon) \frac{\partial u^\epsilon}{\partial x} = 0,\\
&u^\epsilon(T,x)=x.
\end{aligned}
\right.
\end{align}
The existence and uniqueness of the classical solution to  the PDE \eqref{eq:regularized} follow from the next result.
\begin{theorem}\label{epsilon-existence}
For any $\epsilon>0$, there exists a  unique classical solution $u^\epsilon \in C^{1,2}([0,T] \times \mathbb{R})$ to the PDE    \eqref{eq:regularized}. Moreover, the following properties hold:
\begin{enumerate}
    \item [(i)] The derivatives $\frac{\partial u^\epsilon}{\partial t}$, $\frac{\partial u^\epsilon}{\partial x}$ and $\frac{\partial^2 u^\epsilon}{\partial x^2}$ are bounded on $Q_T$, with $\left\| \frac{\partial u^\varepsilon}{\partial x} \right\|_\infty \le M_1^{\varepsilon}$ and $\left\| \frac{\partial^2 u^\varepsilon}{\partial x^2} \right\|_\infty \le M_2^{\varepsilon}$ independent of $m$. Moreover,  $|u^\epsilon(t, x) - x| \le M(T - t)$ for all $(t,x) \in Q_T$.
    \item [(ii)] For any $t\in [0,T]$, $\frac{\partial^2 u^\varepsilon(t,\cdot)}{\partial x^2}$ is continuously differentiable with respect to $x$, and $\frac{\partial^3 u^\epsilon}{\partial x^3}$ is  bounded on $Q_T$.
\end{enumerate}

\end{theorem}
\begin{proof}
Define $v^\varepsilon(t, x) := u^\varepsilon(t, x) - x$. A direct calculation yields that $v^\epsilon$ satisfies
\begin{align}\label{eq:v:epsilon}
\left\{
\begin{aligned}
&\frac{\partial v^\varepsilon}{\partial t} + D^\varepsilon(t, x + v^\varepsilon) \frac{\partial^2 v^\varepsilon}{\partial x^2} - V^\varepsilon(t, x + v^\varepsilon) \frac{\partial v^\varepsilon}{\partial x} - V^\varepsilon(t, x + v^\varepsilon) = 0,\\
&v^\epsilon(T,x)=0.
\end{aligned}
\right.
\end{align}
First, we establish the existence and uniqueness of the classical solution. We embed \eqref{eq:v:epsilon} into the framework of \citet*{2002On} by setting $f(t,x,y,z) := -V^\varepsilon(t, x+y)$,  $g(t,x,y,z) := -V^\varepsilon(t, x+y)$,  $\sigma(t,x,y) := \sqrt{2D^\varepsilon(t, x+y)}$, and $h(x) := 0$.

We now verify that $f$, $g$, $h$ and $\sigma$ satisfy  Assumptions (A1), (A2), and (B.A2) in \citet*{2002On}.
Because $D^{\varepsilon}\in C(Q_T)$ and $V^\varepsilon\in C(Q_T)$ are smooth in the second argument, it follows that $f$, $g$, $h$ and $\sigma$ satisfy (A{1.0}), (A{1.4}), (A{2.4}), and (B.A{2.1}). 
Moreover, by \eqref{eq:epsilon:bound}, they also satisfy (A{1.3}), (A{2.2}) and (A{2.3}).
In addition, all partial derivatives of $D^\varepsilon$ and $V^\varepsilon$ with respect to the second argument are bounded. Consequently, $f$, $g$, $h$, and $\sigma$ satisfy (A{1.1}), (A{1.2}), (A{2.1}), and (B.A{2.2}). Therefore, $f$, $g$, $h$ and $\sigma$ satisfy Assumptions (A.1), (A.2) and (B.A2) in \citet*{2002On}. The PDE \eqref{eq:v:epsilon} admits a unique bounded solution $v^\varepsilon \in C^{1,2}([0,T] \times \mathbb{R})$ by Corollary B.7 therein.

Second, we  prove property (i).  The proof of Corollary B.7 in \citet*{2002On} gives a bound for $\frac{\partial v^\varepsilon}{\partial x}$ depending only on $T$, the bounds of the coefficients ($D^\epsilon$ and $V^\epsilon$), and the Lipschitz constants of the coefficients with respect to the second argument; in particular, this bound is independent of $m$.
Furthermore, both $\frac{\partial v^\varepsilon}{\partial x}$ and $\frac{\partial^2 v^\varepsilon}{\partial x^2}$ are Lipschitz continuous in $x$, and their Lipschitz constants depend only on the Lipschitz constants of the coefficients and their derivatives (up to second order) with respect to the second argument by Corollary B.4 therein.  It then follows from \eqref{eq:v:epsilon} that $\frac{\partial v^\varepsilon}{\partial t}$ is also bounded.
Moreover, combining \eqref{eq:v:epsilon} with the bound $|V^\varepsilon|\le M$ and the parabolic maximum principle yields $|v^\varepsilon(t, x)| \le M(T - t)$ for all $(t,x) \in Q_T$. Finally, the desired result for $u^\varepsilon$ follows immediately from $v^\varepsilon(t,x)=u^\varepsilon(t,x)-x$. In particular, we denote $\left\| \frac{\partial u^\varepsilon}{\partial x} \right\|_\infty \le M_1^{\varepsilon}$ and $\left\| \frac{\partial^2 u^\varepsilon}{\partial x^2} \right\|_\infty \le M_2^{\varepsilon}$, where both $M_1^{\varepsilon}$ and $M_2^{\varepsilon}$ are independent of $m$.

Third, we  prove property (ii). Consider spatial difference quotients of $u^\epsilon$.
Specifically, for any $(t,x)\in Q_T$ and $h\neq 0$, define
\begin{align*}
    &\tilde{D}^\epsilon(t,x) := D^\epsilon(t, u^\epsilon(t,x)),\quad \tilde{V}^\epsilon(t,x) := V^\epsilon(t, u^\epsilon(t,x)),\\
    &w^\epsilon_h(t,x) := \frac{u^\epsilon(T-t, x+h) - u^\epsilon(T-t, x)}{h},\\
    &f^\epsilon_h(t,x) := \frac{\tilde{D}^\epsilon(T-t, x+h) - \tilde{D}^\epsilon(T-t, x)}{h} \frac{\partial^2 u^\epsilon}{\partial x^2}(T-t, x)\\
&\quad\quad\quad\quad\quad -\frac{\tilde{V}^\epsilon(T-t, x+h) - \tilde{V}^\epsilon(T-t, x)}{h} \frac{\partial u^\epsilon}{\partial x}(T-t, x).
\end{align*}
Then, by \eqref{eq:regularized}, $w^\epsilon_h$ satisfies the linear parabolic PDE:
\begin{align}\label{eq:w^h}
\begin{cases}
\frac{\partial w^\epsilon_h}{\partial t}(t,x) - \tilde{D}^\epsilon(T-t, x+h) \frac{\partial^2 w^\epsilon_h}{\partial x^2}(t,x) + \tilde{V}^\epsilon(T-t, x+h) \frac{\partial w^\epsilon_h}{\partial x}(t,x) = f^\epsilon_h(t,x), \\
w^\epsilon_h(0,x) = 1.
\end{cases}
\end{align}
To obtain some necessary uniform estimates for $w^\epsilon_h$, we will verify that all structural assumptions (H1)-(H3) of Theorem 4.2 in \citet*{lorenzi2000optimal} are satisfied by the PDE \eqref{eq:w^h} as follows.
\vskip 4pt
{\bf Verifying  Assumption (H1)}: Both  $\tilde{D}^\epsilon(T-t, x+h)$ and $-\tilde{V}^\epsilon(T-t, x+h)$ are bounded and continuous.  Moreover, their spatial derivatives are globally bounded, independent of $h$ and $t$, which implies that they are globally $\alpha$-H\"{o}lder continuous in $x$. Hence, the coefficients belong to $L^\infty(Q_T) \cap B([0,T]; C^\alpha(\mathbb{R}))$,  and Assumption (H1) holds.
\vskip 4pt
{\bf Verifying  Assumption (H2)}: The coefficient $\tilde{D}^\epsilon(T-t, x+h)$ preserves the uniform ellipticity condition $\tilde{D}^\epsilon \ge \kappa > 0$, independent of $h$, thus satisfies Assumption (H2).
\vskip 4pt
{\bf Verifying  Assumption (H3)}: We examine the regularity of the data. By Theorem \ref{epsilon-existence} (i), $\frac{\partial u^\epsilon}{\partial x}$ and $\frac{\partial^2 u^\epsilon}{\partial x^2}$ are globally bounded and Lipschitz continuous. $\tilde{D}^\epsilon$ and $\tilde{V}^\epsilon$ admit globally bounded first and second spatial derivatives by the chain rule. Consequently, their spatial difference quotients are uniformly Lipschitz continuous, with their Lipschitz constants independent of $h$ and $t$. It follows that $f^\epsilon_h$ is bounded and Lipschitz continuous, uniformly in $h$ and $t$, and thus belongs to $ B([0,T]; C^\alpha(\mathbb{R}))$. Finally, the initial condition $w^\epsilon_h(0,x) = 1$ trivially belongs to $C^{2+\alpha}(\mathbb{R})$. Therefore, Assumption (H3) is satisfied.

Thus, applying Theorem 4.2 in \citet*{lorenzi2000optimal}, we  obtain the following estimate:
\begin{align*}
\|w^\epsilon_h\|_{B([0,T]; C^{2+\alpha}(\mathbb{R}))} \le c(\epsilon) \left( \|w^\epsilon_h(0, \cdot)\|_{C^{2+\alpha}(\mathbb{R})} + \|f^\epsilon_h\|_{B([0,T]; C^\alpha(\mathbb{R}))} \right)\leq C(\epsilon)
\end{align*}
for some constants $c(\epsilon)>0$ and $C(\epsilon)>0$, which only depend on $\epsilon$, $\kappa$, $\Lambda$, $M$, and $T$ but are  independent of $h$ and $t$. It follows that, for any fixed $t \in [0,T]$, the sequence $\left\{ \frac{\partial^2 w^\epsilon_h}{\partial x^2}(t, \cdot) \right\}_{h \ne 0}$ is uniformly bounded and equicontinuous. By the Arzelà–Ascoli theorem, there exists a subsequence $h_k \to 0$ such that $\frac{\partial^2 w^{\epsilon}_{h_k}}{\partial x^2}(t, \cdot)$ converges uniformly on compact sets to a continuous function $\tilde{w}^\epsilon_t(\cdot)$. For any $x_0$ and $x$, we have 
\begin{align*}
\int_{x_0}^x \frac{\partial^2 w^{\epsilon}_{h_k}}{\partial y^2}(t,y) \md y = \frac{1}{h_k} \int_x^{x+h_k} \frac{\partial^2 u^\epsilon}{\partial y^2}(T-t,y) \md y - \frac{1}{h_k} \int_{x_0}^{x_0+h_k} \frac{\partial^2 u^\epsilon}{\partial y^2}(T-t,y) \md y.
\end{align*}
Thus 
\begin{align*}
\int_{x_0}^x \tilde w^\epsilon_t(y) \md y = \frac{\partial^2 u^\epsilon}{\partial x^2}(T-t,x) - \frac{\partial^2 u^\epsilon}{\partial x^2}(T-t,x_0).
\end{align*}
As $\tilde w^\epsilon_t(x)$ is continuous, it follows that $\frac{\partial^2 u^\epsilon}{\partial x^2}(T-t,x)$ is continuously differentiable with respect to $x$. The uniform boundedness of $\left\{ \frac{\partial^2 w^\epsilon_h}{\partial x^2}(t, \cdot) \right\}_{h \ne 0}$ ensures that $\frac{\partial^3 u^\epsilon}{\partial x^3}$ is  bounded on $Q_T$, which completes the proof.
\end{proof}

As a direct consequence of Theorem \ref{epsilon-existence}, the next corollary gives additional regularity results of the regularized solution for the subsequent analysis.

\begin{corollary}\label{cor:uepsilon_properties}
For any fixed $x$, $\frac{\partial u^\epsilon}{\partial x}$ is absolutely continuous with respect to $t$, and  $\frac{\partial}{\partial t}(\frac{\partial u^\epsilon}{\partial x}) = \frac{\partial}{\partial x}(\frac{\partial u^\epsilon}{\partial t})$ a.e. on $Q_T$. In particular, $\frac{\partial u^\epsilon}{\partial x}\in W^{1,2}_{\infty}(Q_T)$.
\end{corollary}

\begin{proof}
The derivatives $\frac{\partial^i u^\epsilon}{\partial x^i}$ for $i=1,2,3$ exist and are bounded using  Theorem \ref{epsilon-existence}. Thus, based on  
 the PDE \eqref{eq:regularized}, $\frac{\partial}{\partial x}\left(\frac{\partial u^\epsilon}{\partial t}\right)$ exists and is bounded. Hence,
\begin{align*}
\frac{\partial u^\epsilon}{\partial x}(t,x) = \frac{\partial u^\epsilon}{\partial x}(0,x) + \int_0^t \frac{\partial}{\partial x}\left(\frac{\partial u^\epsilon}{\partial s}(s,x)\right) \md s,\quad (t,x)\in Q_T.
\end{align*}
This implies that, for any fixed $x$, the function $\frac{\partial u^\epsilon}{\partial x}$ is absolutely continuous with respect to $t$, and $\frac{\partial}{\partial t}\left(\frac{\partial u^\epsilon}{\partial x}\right) = \frac{\partial}{\partial x}\left(\frac{\partial u^\epsilon}{\partial t}\right)$ almost everywhere. In view of the boundedness of $\frac{\partial^2 u^\epsilon}{\partial x^2}$ and $\frac{\partial^3 u^\epsilon}{\partial x^3}$, it further holds that $\frac{\partial u^\epsilon}{\partial x} \in W^{1,2}_{\infty}(Q_T)$. 
\end{proof}

In order to investigate the continuous curve of separation boundary defined in Theorem \ref{thm: solving pde}, we first present the next result on the spatial monotonicity of the  regularized solution $ u^\epsilon$.

\begin{proposition}\label{lma:bound:'}
There exist two positive constants $\delta_1$ and $M_1$ depending only on $\kappa$, $\Lambda$, $M$, and $T$, and independent of $\epsilon$ and $m$, such that
$0<\delta_1 \leq \frac{\partial u^\epsilon}{\partial x}(t,x) \leq M_1 $ for all $(t,x) \in Q_T$ and $\epsilon>0$.
\end{proposition}
\begin{proof}

Define $w^\epsilon(t, x) := \frac{\partial u^\epsilon}{\partial x}(T-t, x)$. Based on Corollary \ref{cor:uepsilon_properties}, $\frac{\partial}{\partial x}\left(\frac{\partial u^\epsilon}{\partial t}\right)$ exists and $w^\epsilon$ is a strong solution to the following linear parabolic PDE in divergence form:
\begin{align*}
\begin{cases}
\frac{\partial w^\epsilon}{\partial t} = \frac{\partial}{\partial x} \left( \tilde{D}^\epsilon(T-t, x) \frac{\partial w^\epsilon}{\partial x} \right) - \frac{\partial}{\partial x} \left( \tilde{V}^\epsilon(T-t, x) w^\epsilon \right), \\
w^\epsilon(0, x) = 1.
\end{cases}
\end{align*}
Because $w^\epsilon(0,x)$, $\tilde D^\epsilon$ and $\tilde V^\epsilon$ are bounded, and $\tilde D^\epsilon$ satisfies the uniform ellipticity condition, \citet*{aronson1968non}[Theorem 10 (ii) and (vi)] ensures the existence of a globally defined weak fundamental solution $G^\epsilon(x, t; \zeta, 0)$ such that  $w^\epsilon(t, x) = \int_{\mathbb{R}} G^\epsilon(x, t; \zeta, 0) w^\epsilon(0, \zeta) \md \zeta$. Moreover, there exist two positive constants $c_1$ and $c_2$ depending only on $\kappa$, $\Lambda$, $M$, and $T$ such that
$G^\epsilon(x, t; \zeta, 0) \ge \frac{c_1}{\sqrt{t}} \exp\left( -c_2 \frac{(x-\zeta)^2}{t} \right)$.
Therefore,
\begin{align*}
w^\epsilon(t, x) &= \int_{\mathbb{R}} G^\epsilon(x, t; \zeta, 0) w^\epsilon(0, \zeta) \md \zeta = \int_{\mathbb{R}} G^\epsilon(x, t; \zeta, 0) \md \zeta\geq \int_{\mathbb{R}} \frac{c_1}{\sqrt{t}} \exp\left( -c_2 \frac{(x-\zeta)^2}{t} \right) \md \zeta= c_1 \sqrt{\frac{\pi}{c_2}}.
\end{align*}
Hence, taking $\delta_1 := c_1 \sqrt{\frac{\pi}{c_2}}$ yields the lower bound.  The upper bound $\frac{\partial u^\epsilon}{\partial x}(t,x) \le M_1$ can be proved similarly.
\end{proof}

\vskip 5pt
\textbf{Step 2: Strong Solution and Continuous Curve of Separation Boundary.}
Ideally, one would apply standard $L^p$ or Schauder estimates to obtain uniform regularity bounds for the approximating solutions $\{u^\epsilon\}$. However, the versions of these estimates relevant here require sufficiently regular coefficients, in particular Hölder continuity in time, whereas \(\gamma^\epsilon\) is not Hölder continuous with respect to \(t\) under our assumptions.       

To overcome this difficulty, we employ an energy method to derive uniform estimates for $\{u^\epsilon\}$ locally in $W^{1,2}_2$. These \emph{a priori} bounds are obtained directly from the equation and do not rely on the continuity of the coefficients. As a consequence, we obtain compactness of the sequence and the following convergence result.
\begin{theorem}\label{existence-u}
There exists a sequence $\{u^{\epsilon_k}\}$ and a function $$u\in W^{1,2}_{2,\mathrm{loc}}(Q_T)\cap C^{\alpha/2, \alpha}_{\mathrm{loc}}(Q_T)\quad \text{for all }\alpha\in\left(0,\frac{1}{2}\right)$$  such that, on every compact subset $K\subset Q_T$, $u^{\epsilon_k}$ converges to $u$ weakly in $W^{1,2}_2(K)$ and strongly in $C^{\alpha/2, \alpha}(K)$  for any $\alpha < \frac{1}{2}$. Furthermore, $u$ satisfies $|u(t, x) - x| \le M(T - t)$ on $Q_T$, and $\delta_1 \leq \frac{\partial u}{\partial x} \leq M_1$ a.e. in $Q_T$.

\end{theorem}

\begin{proof}
Define an exponentially decaying weight on $\mathbb{R}$ by
\begin{align*}
\phi(x) := e^{-  \sqrt{1+x^2}},\quad x\in\mR.
\end{align*}
It follows that $      (\phi'(x))^2 \le (\phi(x))^2\quad \text{for all }x\in\mR. $    
Multiplying both sides of the regularized PDE \eqref{eq:regularized} by the test function $\phi^2 \frac{\partial^2 u^\epsilon}{\partial x^2}$ and integrating over $\mathbb{R}$, we get
\begin{align}\label{eq:energy}
\int_{\mathbb{R}} \phi^2 \frac{\partial u^\epsilon}{\partial t} \frac{\partial^2 u^\epsilon}{\partial x^2} \md x + \int_{\mathbb{R}} \phi^2 D^\epsilon(t, u^\epsilon) \left(\frac{\partial^2 u^\epsilon}{\partial x^2}\right)^2 \md x = \int_{\mathbb{R}} \phi^2 V^\epsilon(t, u^\epsilon) \frac{\partial u^\epsilon}{\partial x} \frac{\partial^2 u^\epsilon}{\partial x^2} \md x.
\end{align}
Note that
\begin{align*}
\int_{\mathbb{R}} \phi^2 \frac{\partial u^\epsilon}{\partial t} \frac{\partial^2 u^\epsilon}{\partial x^2} \md x &= -\int_{\mathbb{R}} \frac{\partial}{\partial x}\left(\phi^2 \frac{\partial u^\epsilon}{\partial t}\right) \frac{\partial u^\epsilon}{\partial x} \md x= -\frac{1}{2} \frac{\md}{\md t} \int_{\mathbb{R}} \phi^2 \left(\frac{\partial u^\epsilon}{\partial x}\right)^2 \md x - 2\int_{\mathbb{R}} \phi \phi' \frac{\partial u^\epsilon}{\partial t} \frac{\partial u^\epsilon}{\partial x} \md x\\
&= -\frac{1}{2} \frac{\md}{\md t} \int_{\mathbb{R}} \phi^2 \left(\frac{\partial u^\epsilon}{\partial x}\right)^2 \md x+2\int_{\mathbb{R}} \phi \phi' D^\epsilon \frac{\partial^2 u^\epsilon}{\partial x^2} \frac{\partial u^\epsilon}{\partial x} \md x - 2\int_{\mathbb{R}} \phi \phi' V^\epsilon \left(\frac{\partial u^\epsilon}{\partial x}\right)^2 \md x,
\end{align*}
where the first equality follows from the exponential decay of $\phi$, the second equality from  $u^\epsilon \in C^{1,2}(Q_T)$, and the last equality results from \eqref{eq:regularized}.
Substituting this back into \eqref{eq:energy}, we obtain
\begin{align}\label{eq:energy:1}
&-\frac{1}{2} \frac{\md}{\md t} \int_{\mathbb{R}} \phi^2 \left(\frac{\partial u^\epsilon}{\partial x}\right)^2 \md x + \int_{\mathbb{R}} \phi^2 D^\epsilon \left(\frac{\partial^2 u^\epsilon}{\partial x^2}\right)^2 \md x \nonumber\\= &\int_{\mathbb{R}} \phi^2 V^\epsilon \frac{\partial u^\epsilon}{\partial x} \frac{\partial^2 u^\epsilon}{\partial x^2} \md x - 2\int_{\mathbb{R}} \phi \phi' D^\epsilon \frac{\partial^2 u^\epsilon}{\partial x^2} \frac{\partial u^\epsilon}{\partial x} \md x  + 2\int_{\mathbb{R}} \phi \phi' V^\epsilon \left(\frac{\partial u^\epsilon}{\partial x}\right)^2 \md x .
\end{align}
Using the uniform bounds $D^\epsilon \ge \kappa$ and $|V^\epsilon| \le M$ and  applying Young’s inequality $ab \le \frac{\beta}{2}a^2 + \frac{1}{2\beta}b^2$ with $\beta = \frac{\kappa}{2}$ to the first two terms on the right-hand side of \eqref{eq:energy:1},  we have 
\begin{align*}
&\left| \int_{\mathbb{R}} \phi^2 V^\epsilon \frac{\partial u^\epsilon}{\partial x} \frac{\partial^2 u^\epsilon}{\partial x^2} \md x \right| \le \frac{\kappa}{4} \int_{\mathbb{R}} \phi^2 \left(\frac{\partial^2 u^\epsilon}{\partial x^2}\right)^2 \md x + \frac{M^2}{\kappa} \int_{\mathbb{R}} \phi^2 \left(\frac{\partial u^\epsilon}{\partial x}\right)^2 \md x,\\
&\left| - 2\int_{\mathbb{R}} \phi \phi' D^\epsilon \frac{\partial^2 u^\epsilon}{\partial x^2} \frac{\partial u^\epsilon}{\partial x} \md x \right| \le \frac{\kappa}{4} \int_{\mathbb{R}} \phi^2 \left(\frac{\partial^2 u^\epsilon}{\partial x^2}\right)^2 \md x + \frac{4\Lambda^2}{\kappa} \int_{\mathbb{R}} (\phi')^2 \left(\frac{\partial u^\epsilon}{\partial x}\right)^2 \md x.
\end{align*}
For the third term, it holds by the Cauchy–Schwarz inequality that
\begin{align*}
\left| 2\int_{\mathbb{R}} \phi \phi' V^\epsilon \left(\frac{\partial u^\epsilon}{\partial x}\right)^2 \md x \right| \le M \int_{\mathbb{R}} \left(\phi^2 + (\phi')^2\right) \left(\frac{\partial u^\epsilon}{\partial x}\right)^2 \md x.
\end{align*}
Hence
\begin{align}\label{eq:energy:2}
&-\frac{1}{2} \frac{\md}{\md t} \int_{\mathbb{R}} \phi^2 \left(\frac{\partial u^\epsilon}{\partial x}\right)^2 \md x + \frac{\kappa}{2} \int_{\mathbb{R}} \phi^2 \left(\frac{\partial^2 u^\epsilon}{\partial x^2}\right)^2 \md x \nonumber\\
\le & \int_{\mathbb{R}} \left[ \frac{M^2}{\kappa}\phi^2 + \frac{4\Lambda^2}{\kappa}(\phi')^2 + M(\phi^2 + (\phi')^2) \right] \left(\frac{\partial u^\epsilon}{\partial x}\right)^2 \md x\leq C_1 \int_{\mathbb{R}} \phi^2 \left(\frac{\partial u^\epsilon}{\partial x}\right)^2 \md x
\end{align}
for some constant $C_1>0$ depending only on $\kappa, M$ and $\Lambda$,  where we have used $(\phi')^2 \le \phi^2$, $D^\epsilon \ge \kappa$ and substituted the derived estimates for the three integral terms back into \eqref{eq:energy:1}.
Define $Y^\epsilon(t):=\int_{\mathbb{R}} \phi^2 \left(\frac{\partial u^\epsilon}{\partial x}(t,x)\right)^2 \md x$. Corollary \ref{cor:uepsilon_properties} ensures that $Y^\epsilon(\cdot)$ is absolutely continuous. Then \eqref{eq:energy:2} implies that, for a.e. $t \in [0, T]$,
\begin{align*}
-\frac{1}{2} (Y^\epsilon)'(t) \le C_1 Y^\epsilon(t).
\end{align*}
Applying the backward Gronwall inequality on $[t, T]$, we obtain
\begin{align*}
Y^\epsilon(t) \le M_{\phi} e^{2C_1 (T-t)} \quad \text{for all } t\in[0,T],
\end{align*} 
where $M_{\phi}:=Y^\epsilon(T)=\int_{\mathbb{R}} e^{-2  \sqrt{1+x^2}} \md x$.
Integrating the inequality \eqref{eq:energy:2} from $0$ to $T$, we have
\begin{align*}
\frac{\kappa}{2} \int_0^T \int_{\mathbb{R}} \phi^2 \left(\frac{\partial^2 u^\epsilon}{\partial x^2}\right)^2 \md x \md s \le \frac{1}{2} Y^\epsilon(T) + C_1 \int_0^T Y^\epsilon(s) \md s \le \frac{1}{2} M_\phi + C_1 M_\phi e^{2C_1 T} T.
\end{align*}
On any compact subset $K \subset Q_T$, the continuous positive function $\phi(x)$ has a strict positive lower bound $c_K:=\inf\limits_{(t,x) \in K} \phi(x) > 0$. We deduce the uniform bound of the local Lebesgue norm:
\begin{align}\label{2-bound}
\iint_K \left(\frac{\partial^2 u^\epsilon}{\partial x^2}\right)^2 \md x \md t \le \frac{1}{c_K^2} \int_0^T \int_{\mathbb{R}} \phi^2 \left(\frac{\partial^2 u^\epsilon}{\partial x^2}\right)^2 \md x \md s \le \frac{2}{\kappa c_K^2}\left(\frac{1}{2} M_\phi + C_1 M_\phi e^{2C_1 T} T\right).
\end{align} 
We obtain that the local $L^2(K)$ norm of $\frac{\partial u^\epsilon}{\partial t}$ is uniformly bounded in $\epsilon$ thanks to $(a+b)^2 \le 2a^2 + 2b^2$. Hence, $u^\epsilon$ is uniformly bounded in the local parabolic Sobolev space $W^{1,2}_2(K)$.

By virtue of these uniform local energy estimates, the Banach–Alaoglu theorem and a diagonal argument can be invoked to yield a subsequence $\{u^{\epsilon_k}\}$ and a $u\in W^{1,2}_2(K)$ such that $u^{\epsilon_k} \rightharpoonup u$ in $W^{1,2}_2(K)$ for any compact $K \subset Q_T$ because $W^{1,2}_2(K)$ is reflexive. We further have $u^{\epsilon_k} \to u$ strongly in $C^{\alpha/2, \alpha}(K)$, implying that $u \in C^{\alpha/2, \alpha}(K)$ by the compact embedding $W^{1,2}_2(K) \hookrightarrow C^{\alpha/2, \alpha}(K)$ for any $\alpha < 1/2$.

Finally, the uniform convergence yields $|u(t, x) - x| \le M(T - t)$ on $Q_T$, and the weak convergence preserves $\delta_1 \leq \frac{\partial u}{\partial x} \leq M_1$ in view of $u^{\epsilon_k} \rightharpoonup u$.
\end{proof}

The next result ensures the absolute continuity of the strong solution $u$ on one-dimensional sections.

\begin{lemma}\label{lem:acl}
Let $u\in W_{2,\mathrm{loc}}^{1,2}(Q_{T})$. Then
\begin{enumerate}[label=(\arabic*), font=\upshape]
    \item[(i)] For a.e. $t \in [0,T]$, $u(t,\cdot) \in W_{2,\mathrm{loc}}^2(\mathbb{R})$, and $\frac{\partial u}{\partial x}(t,\cdot)$ is absolutely continuous on any finite interval.
    \item[(ii)] For a.e. $x \in \mathbb{R}$, $u(\cdot,x) \in W_2^1([0,T])$, and $u(\cdot,x)$ is absolutely continuous on $[0,T]$.
\end{enumerate}
\end{lemma}

\begin{proof}
See Appendix \ref{appendix:lma:smooth:method}.
\end{proof}

Building upon these regularity results, we can now confirm the existence of a continuous curve of separation boundary induced by the population aggregation.

\begin{proposition}\label{lem:dividing_curve}
There exists a continuous function $b^* \in C([0,T])$ such that 
$\{(t,x) \in Q_{T} : u(t,x) = R(t)\}=\{(t,b^*(t)):t\in[0,T]\}$. Moreover, $u$ is a strong solution of the PDE \eqref{eq:pde}.
\end{proposition}

\begin{proof}
Define the continuous function $F: Q_T \to \mathbb{R}$ by
\begin{align*}
F(t,x) = u(t,x) - R(t).
\end{align*}
There exists a set of full measure $\mathcal{T}_1 \subset [0,T]$ such that for each $t \in \mathcal{T}_1$, $u(t,\cdot) \in W_{2,\mathrm{loc}}^1(\mathbb{R})$ and its weak derivative coincides with $\frac{\partial u}{\partial x}(t,\cdot)$ by Lemma \ref{lem:acl}(i).

Moreover, because $\delta_1 \le \frac{\partial u}{\partial x}(t,x) \le M_1$ for a.e. $(t,x) \in Q_T$, Fubini's theorem yields another full-measure set $\mathcal{T}_2 \subset [0,T]$ such that for every $t \in \mathcal{T}_2$, $\delta_1 \le \frac{\partial u}{\partial x}(t,x) \le M_1$ for a.e. $x \in \mathbb{R}$. 
Let $\mathcal{T} := \mathcal{T}_1 \cap \mathcal{T}_2$. Then for each $t \in \mathcal{T}$, $u(t,\cdot) \in W_{2,\mathrm{loc}}^1(\mathbb{R})$ and $\delta_1 \le \frac{\partial u}{\partial x}(t,x) \le M_1$ for a.e. $x \in \mathbb{R}$.
Hence, for each $t\in\mathcal{T}$, $F(t,\cdot)$ is Lipschitz continuous and 
\begin{align*}
F(t,x_{2}) - F(t,x_{1}) = \int_{x_{1}}^{x_{2}} \frac{\partial u}{\partial x}(t,\zeta) \md\zeta \ge \delta_1(x_{2} - x_{1}) > 0\quad \text{for all } x_1<x_2.
\end{align*}
Both the Lipschitz continuity and the strict monotonicity extend to all $t \in [0,T]$ by the continuity of $F$ and the density of $\mathcal{T}$ in $[0,T]$. Consequently, according to Theorem 1H.3 in \cite{dontchev2009implicit}, $F(t,x)=0$ uniquely determines a root $x = b^*(t)$, which is continuous on $[0,T]$.

Note that the set $\{(t,b^*(t)) : t \in [0,T]\}$ has Lebesgue measure zero in $Q_T$, and $u^{\epsilon_k}$ converges pointwise to $u$, it follows that, for a.e. $(t,x)$,
\begin{align*}
D^{\epsilon_{k}}(t,u^{\epsilon_{k}}(t,x)) \to D(t,u(t,x)) \,\,\text{and}\,\, V^{\epsilon_{k}}(t,u^{\epsilon_{k}}(t,x)) \to V(t,u(t,x))  \quad \text{as }k\to\infty, 
\end{align*}
$D^{\epsilon_{k}}(t,u^{\epsilon_{k}})\psi \to D(t,u)\psi$ and $V^{\epsilon_{k}}(t,u^{\epsilon_{k}})\psi \to V(t,u)\psi$ in $L^{2}(Q_{T})$ by the dominated convergence theorem.
Multiplying \eqref{eq:regularized} by an arbitrary $\psi \in L^2(Q_T)$ with compact support and integrating over $Q_T$, we obtain
\begin{align*}
\iint_{Q_T} \frac{\partial u^{\epsilon_k}}{\partial t} \psi \md x \md t + \iint_{Q_T} D^{\epsilon_k}(t, u^{\epsilon_k}) \frac{\partial^2 u^{\epsilon_k}}{\partial x^2} \psi \md x \md t - \iint_{Q_T} V^{\epsilon_k}(t, u^{\epsilon_k}) \frac{\partial u^{\epsilon_k}}{\partial x} \psi \md x \md t = 0.
\end{align*}
 Combining this with the weak convergence of $\frac{\partial u^{\epsilon_k}}{\partial t}$, $\frac{\partial u^{\epsilon_k}}{\partial x}$, and $\frac{\partial^2 u^{\epsilon_k}}{\partial x^2}$ to $\frac{\partial u}{\partial t}$, $\frac{\partial u}{\partial x}$, and $\frac{\partial^2 u}{\partial x^2}$ locally in $L^2$ on the support of $\psi$, respectively, we obtain the limit equation:
\begin{align*}
\iint_{Q_{T}} \left( \frac{\partial u}{\partial t} + D(t,u) \frac{\partial^{2}u}{\partial x^{2}} - V(t,u) \frac{\partial u}{\partial x} \right) \psi \md x \md t = 0.
\end{align*}
Because such compactly supported $\psi \in L^{2}(Q_{T})$ is arbitrary, it follows that $u$ satisfies \eqref{eq:pde} a.e. in $Q_{T}$.
\end{proof}
\begin{proof}[{\bf Proof of Theorem \ref{thm: solving pde} }]
The proof follows from  Theorem~\ref{existence-u} and Proposition~\ref{lem:dividing_curve}. 
\end{proof}
With the globally continuous strong solution $u$ to the  PDE \eqref{eq:pde} at hand, we then construct the  solution to the original multi-dimensional FBSDE system \eqref{eq:fbsde_sufficiency} in the next subsection.

\subsubsection{Existence of Solution to  the FBSDE \eqref{eq:fbsde_sufficiency}}

Let $u \in W_{2,\mathrm{loc}}^{1,2}(Q_{T}) \cap C^{\alpha/2, \alpha}_{\mathrm{loc}}(Q_{T})$ for all $\alpha \in (0,\frac{1}{2})$ be the strong solution of the PDE \eqref{eq:pde}. Let $v\in C(Q_{T})$ be the unique function defined by $u(t,v(t,y))=y$. We introduce this spatial inverse mapping $v$ to transform the phase space, which is a key step for constructing the auxiliary process later. The following lemma rigorously establishes the Sobolev regularity of $v$, and gives explicit representations for its weak derivatives; the proof is deferred to Appendix \ref{appendix:lma:v:derivative}.

\begin{lemma}\label{lma:v:derivative}
$v \in W_{2,\mathrm{loc}}^{1,2}(Q_T) \cap C^{\alpha/2, \alpha}_{\mathrm{loc}}(Q_T)$ for all $\alpha \in (0,\frac{1}{2})$, and its weak derivatives are given by
\begin{align}\label{eq:derivative:1}
\frac{\partial v}{\partial y}(t, y) &= \left(\frac{\partial u}{\partial x}(t, v(t, y))\right)^{-1}, \quad
\frac{\partial^2 v}{\partial y^2}(t, y) = -\frac{\partial^2 u}{\partial x^2}(t, v(t, y))\left(\frac{\partial u}{\partial x}(t, v(t, y))\right)^{-3}, \\
\frac{\partial v}{\partial t}(t, y) &= -\frac{\partial u}{\partial t}(t, v(t, y)) \left(\frac{\partial u}{\partial x}(t, v(t, y))\right)^{-1}.
\label{eq:derivative:2}
\end{align}
\end{lemma}
\noindent
To proceed, we also impose the following assumption.
\begin{assumption}\label{assumption-2}
 $m(\cdot)$ is Lipschitz continuous on $[0,T]$.
\end{assumption}
\noindent
We are now ready to address the solution of the FBSDE system \eqref{eq:fbsde_sufficiency}.
\begin{proposition}\label{FBSDE-existence}
Under Assumptions \ref{assumption-1} and \ref{assumption-2}, the FBSDE system \eqref{eq:fbsde_sufficiency} admits a  solution $(X, P, Q) \in L_{\mathbb{F}}^{2}(\Omega; C([0, T];\mathbb{R}))\times L_{\mathbb{F}}^{2}(\Omega; C([0, T];\mathbb{R}))\times L_{\mathbb{F}}^{2}(0, T; \mathbb{R}^d)$.
\end{proposition}
\begin{proof}
Define $\tilde{\gamma}(Z) := \gamma_1 \mathbbm{1}_{\{Z > 0\}}  + \gamma_2 \mathbbm{1}_{\{Z \leq 0\}}$ and consider the process $Z(t)$ governed by
\begin{align}\label{eq:sde:z}
\left\{
\begin{aligned}
&\md Z(t) = \left[ |\lambda(t)|^2\tilde{\gamma}(Z(t))\frac{\partial u}{\partial x}(t, v(t, Z(t)+R(t))) - R'(t) \right] \md t \\
&\quad \quad\quad\quad+ \tilde{\gamma}(Z(t))\frac{\partial u}{\partial x}(t, v(t, Z(t)+R(t))) \lambda(t)^\top \md W(t),\\
&Z(0) = P_0(0)\xi - R(0).
\end{aligned}
\right.
\end{align}
We first establish the well-posedness of $Z(t)$. To this end, we verify that the coefficients of the SDE \eqref{eq:sde:z} satisfy the conditions 1-5 of \citet*[Theorem 2]{1980On}. Both the drift and diffusion coefficients are bounded and Borel measurable because $|\lambda(t)|$ is bounded by $\lambda_{\max}$, $\tilde{\gamma}$ is bounded by $\gamma_2$, and $\sigma^W(t, Z) := \frac{\partial u}{\partial x}(t, v(t, Z+R(t)))$ is  bounded by $M_1$. Furthermore, exploiting a smooth approximation argument by the same way  as in the proof of  Lemma \ref{lma:v:derivative}, we can apply the chain rule to obtain the weak derivative $\frac{\partial \sigma^W}{\partial Z} = \frac{\partial^2 u}{\partial x^2}(t, v) \left(\frac{\partial u}{\partial x}(t, v)\right)^{-1}$. Because $\frac{\partial u}{\partial x} \ge \delta_1 > 0$ and $\frac{\partial^2 u}{\partial x^2} \in L_{\mathrm{loc}}^2$, it follows that $\frac{\partial \sigma^W}{\partial Z}$ belongs to $L_{\mathrm{loc}}^2$. Consequently, $\sigma^W$ belongs to $W_{2, \mathrm{loc}}^{0,1}$ and Condition 1 holds with $\beta = 1$.  In addition, define $\sigma^V(z) := \tilde{\gamma}(z)$. Then, for any $z_1, z_2 \in \mathbb{R}$, the diffusion coefficient $\tilde{\Sigma}(t, z):=\lambda(t)\sigma^V(z)\sigma^W(t, z)$ satisfies
\begin{align*}
|\tilde{\Sigma}(t,z_1) - \tilde{\Sigma}(t,z_2)|^2 \le\rho_1^2\left( |\sigma^W(t,z_1) - \sigma^W(t,z_2)|\right) + \rho_2^2\left( |\sigma^V(z_1) - \sigma^V(z_2)|\right),
\end{align*}
where $\rho_1(x) := \sqrt{4 \lambda_{\max}^2 \gamma_2^2 M_1 x}$ and $\rho_2(x) := \sqrt{4 \lambda_{\max}^2 M_1^2 \gamma_2 x}$ satisfy Conditions 3--5 in \citet*[Theorem 2]{1980On}.
Moreover, Condition 2 is satisfied because $\sigma^V$ is of bounded variation.
Therefore, \eqref{eq:sde:z} admits a unique strong solution based on  Theorem 2 and Remark 4 of \citet*{1980On}.

With the well-posedness of $Z(t)$, we now translate back to the original coordinate system. Define $Y(t) := Z(t) + R(t)$ \footnote{A direct application to $Y(t)$ is infeasible because \citet*[Theorem 2]{1980On} requires $\sigma^V$ to be independent of time $t$. $Z(t)$ is introduced to fix the time-dependent discontinuity $y = R(t)$ at the origin to satisfy this condition.}, then $Y$ satisfies
\begin{align}\label{eq:sde:Y}
\left\{
\begin{aligned}
&\md Y(t) =  \tilde{\gamma}(Y(t) - R(t)) \frac{\partial u}{\partial x}(t, v(t, Y(t)))\left(|\lambda(t)|^2 \md t +   \lambda(t)^\top \md W(t)\right),\\
&Y(0)=P_0(0)\xi.
\end{aligned}
\right.
\end{align}
To obtain the desired form of the FBSDE via the generalized Itô formula, we need to ensure that the process spends zero time in $\mathcal{N} \subset Q_T$, the singular set of Lebesgue measure zero where the PDE \eqref{eq:pde} fails to hold.
Fubini's theorem, together with the invariance of null sets under Lipschitz mappings, implies that the transformed set
$\tilde{\mathcal{N}} := \{(t, y) \mid (t, v(t, y)) \in \mathcal{N}\}$ also has Lebesgue measure zero because $v(t,\cdot)$ is globally bi-Lipschitz for each $t \in [0,T]$. $\mathcal{F}_0$ and $\sigma(Y(0))$ coincide up to null sets because $Y(0)=P_0(0)\xi$. 
There exists a measurable map $\Phi : \mathbb{R} \times C([0,T], \mathbb{R}^d) \to C([0,T], \mathbb{R})$ such that $Y_{\cdot} = \Phi(Y(0), W_{\cdot})$ by the well-posedness of the SDE and the Yamada–Watanabe theorem (see, e.g.,  \citet*{kallenberg1996existence}[Theorem 1]). Moreover, as $W_{\cdot}$ is independent of $\mathcal{F}_0$, for any bounded measurable function $f$, it holds that $\mathbb{E}[f(Y(0), W_{\cdot}) \mid \mathcal{F}_0] = \mathbb{E}[f(y_0, W_{\cdot})]\big|_{y_0=Y(0)}$.

For any $y_0\in \mR$, note that the diffusion coefficient of $Y$ is uniformly non-degenerate.  Applying Krylov’s estimate (see \citet*{krylov1980controlled}[Chapter 2, Section 3, Theorem 4]) to $\Phi_{\cdot}(y_0, W_{\cdot})$, we know that there exists a constant $C$ that depends only on $T$, $\kappa$, $\Lambda$, $M$, $\delta_1$, and $M_1$ such that
\begin{align*}
\mathbb{E} \left[ \int_0^T \mathbbm{1}_{\tilde{\mathcal{N}}}(t, \Phi_t(y_0, W_{\cdot})) \md t \right] \leq C \|\mathbbm{1}_{\tilde{\mathcal{N}}}\|_{L^2([0,T] \times \mathbb{R})} = 0.
\end{align*}
As a result, we have
\begin{align}\label{eq:nutset}
    &\mathbb{E}\left[\int_0^T \mathbbm{1}_{{\mathcal{N}}}(t, v(t, Y(t))) \md t\right] = \mathbb{E}\left[\int_0^T \mathbbm{1}_{\tilde{\mathcal{N}}}(t, \Phi_t(Y(0), W_{\cdot})) \md t\right] \nonumber \\
=&\mathbb{E}\left[\mathbb{E}\left[\int_0^T \mathbbm{1}_{\tilde{\mathcal{N}}}(t, \Phi_t(Y(0), W_{\cdot})) \md t\Big|\mathcal{F}_0\right]\right]\nonumber\\=&\mathbb{E}\left[\left(\mathbb{E}\left[\int_0^T \mathbbm{1}_{\tilde{\mathcal{N}}}(t, \Phi_t(y_0, W_{\cdot})) \md t\right]\right)_{y_0=Y(0)}\right]=0.
\end{align}
Subsequently, we apply the inverse spatial mapping to define $x(t):=v(t,Y(t))$, then $Z(t)=Y(t)-R(t)=u(t,x(t))-R(t)$ and $\tilde{\gamma}(Z(t))=\gamma\left(t,\frac{u(t,x(t))}{P_0(t)}\right)$. Using \eqref{eq:nutset} and applying the generalized It\^o formula (see, e.g., \citet*{krylov1980controlled}[Chapter 2, Section 10, Theorem 1])
leads to
\begin{align}
\md x(t) &= \left[ -\frac{\partial u}{\partial t}\left(\frac{\partial u}{\partial x}\right)^{-1} \md t - \frac{1}{2}\frac{\partial^2 u}{\partial x^2}\left(\frac{\partial u}{\partial x}\right)^{-3} \left(|\lambda(t)|^2\gamma^2(t, P_0^{-1}(t)u(t,x(t)))\left(\frac{\partial u}{\partial x}\right)^2\right) \md t \right]\nonumber \\
&\quad + \left(\frac{\partial u}{\partial x}\right)^{-1}\gamma(t, P_0^{-1}(t)u(t,x(t))) \left[ |\lambda(t)|^2\frac{\partial u}{\partial x} \md t + \frac{\partial u}{\partial x} \lambda(t)^\top \md W(t) \right]\nonumber\\
&= \gamma(t, P_0^{-1}(t)u(t,x(t))) \lambda(t)^\top \md W(t).\label{eq:sde:x}
\end{align}
Moreover, using  $u(t, x(t)) = u(t, v(t, Y(t))) = Y(t)$ and \eqref{eq:sde:Y}, we obtain
\begin{align*}
\md u(t,x(t)) = \gamma(t, P_0^{-1}(t)u(t,x(t)))\frac{\partial u}{\partial x}(t,x(t))\left(|\lambda(t)|^2 \md t + \lambda(t)^\top \md W(t)\right).
\end{align*}
Then, let us define the wealth process $X(t) := P_0^{-1}(t)u(t,x(t))$ with $X(0)=P_0^{-1}(0)u(0,x(0))=P_0^{-1}(0)Y(0)=\xi$ and
\begin{align}\label{sde:X}
\md X(t)&= r(t)X(t)\md t + P_0^{-1}(t)\gamma(t, X(t))\frac{\partial u}{\partial x}(t,x(t))\left( |\lambda(t)|^2\md t +  \lambda(t)^\top \md W(t)\right)\nonumber\\
&=r(t)X(t)\md t-P_0^{-1}(t)\left[Q(t)-\gamma(t,X(t))\lambda(t)\right]^\top\left(\lambda(t)\md t+\md W(t)\right),
\end{align}
where $Q(t):= \gamma(t, X(t))\left(1 - \frac{\partial u}{\partial x}(t,x(t))\right)\lambda(t)$. Let $P(t):=x(t)-u(t,x(t))$, we have $P(T)=x(T)-u(T,x(T))=0$ and 
\begin{align*}
\md P(t) &=  -|\lambda(t)|^2\gamma(t, X(t))\frac{\partial u}{\partial x}(t,x(t)) \md t + Q(t)^\top \md W(t)\\
&= -\left(\gamma(t,X(t))|\lambda(t)|^2-\lambda^\top(t)Q(t)\right)\md t+Q^{\top}(t)\md W(t).
\end{align*}
The conclusion then holds that $(X,P,Q)$ solves the FBSDE \eqref{eq:fbsde_sufficiency}.
\end{proof}

To fully justify that the constructed strategy constitutes an intra-personal equilibrium via Theorem \ref{sufficiency}, it remains to verify that the associated state process is atomless.

\begin{proposition}\label{atomless}
For every $t \in (0,T]$, the random variable $X(t)$ is atomless.
\end{proposition}
\begin{proof}
Fix $t\in(0,T]$ and $x \in \mathbb{R}$. The equation $u(t,y)=P_0(t)x$ for $y$ admits a unique solution $y=z^* \in \mathbb{R}$. The events $\{X(t)=x\}$ and $\{x(t)=z^*\}$ coincide by the relation $X(t)=P_0(t)^{-1}u(t,x(t))$.

Recall from \eqref{eq:hat:sigma} that $\hat{\sigma}(s, x) := \gamma(s, P_0^{-1}(s)u(s, x)) \lambda(s)$. Then $c_3^2\leq |\hat{\sigma}|^2\leq C_2^2$, where $c_3 := \gamma_1 \lambda_{\min} > 0$ and $C_2 := \gamma_2 \lambda_{\max} < \infty$. Using  \eqref{eq:sde:x}, we have that   the process $x$ satisfies 
\begin{align*}
\md x(s) &= \hat{\sigma}^{\top}(s, x(s)) \md W(s).
\end{align*}
Suppose otherwise that $\mathbb{P}(x(t)=z^*)=\delta>0$, and set $A := \{x(t) = z^*\}$. Because $x(\cdot)$ has continuous sample paths a.s., $x(t)=\lim\limits_{s\uparrow t}x(s)$ is $\mathcal{F}_{t-}$-measurable, hence $A\in\mathcal{F}_{t-}$.
Let $\{t_n\}_{n\ge1}\subset(0,t)$ be an increasing sequence with $t_n\uparrow t$, and define $P_n := \mathbb{P}(A \mid \mathcal{F}_{t_n})$. We have $P_n \to \mathbb{P}(A \mid \mathcal{F}_{t-}) = \mathbbm{1}_A$ a.s. when $n\to\infty$ by Paul Lévy's zero–one law (see Theorem 9.4.8 and the subsequent corollary in \citet*{chung2000course}). In particular, on $A$, $P_n\to 1$ almost surely.

Define $V_n := \mE_{t_n}[(x(t) - z^*)^2]$. In light of $\mE_{t_n}[x(t) - x(t_n)] = 0$,  we obtain the decomposition
\begin{align}\label{eq:variance decomposition}
V_n &= \mE_{t_n}[(x(t) - x(t_n))^2] + (x(t_n) - z^*)^2.
\end{align}
By It\^o's isometry and the uniform lower bound on $\hat{\sigma}$, it holds that
\begin{align*}
V_n\geq \mE_{t_n}[(x(t) - x(t_n))^2] &= \mE_{t_n}\left[\int_{t_n}^t |\hat{\sigma}(s, x(s))|^2 \md s\right] \ge c_3^2 (t - t_n).
\end{align*}
Similarly, we have $\mE_{t_n}[(x(t) - x(t_n))^2] \le C_2^2 (t - t_n)$. 
Moreover, because $\mE_{t_n}[x(t) - z^*] = x(t_n) - z^*$ and $(x(t) - z^*)\mathbbm{1}_A = 0$, applying the Cauchy-Schwarz inequality, we get
\begin{align*}
(x(t_n) - z^*)^2 &=\left( \mE_{t_n}[(x(t) - z^*)\mathbbm{1}_{A^c}]\right)^2\leq \mE_{t_n}[(x(t) - z^*)^2 \mathbbm{1}_{A^c}] \mathbb{P}(A^c \mid \mathcal{F}_{t_n}) \le V_n(1 - P_n).
\end{align*}
 Substituting these two upper bounds  into \eqref{eq:variance decomposition} yields 
\begin{align*}
V_n 
\le C_2^2 (t - t_n) + V_n(1 - P_n) 
\;&\Longleftrightarrow\;
V_n P_n \le C_2^2 (t - t_n) \\
&\Longrightarrow\;
\mathbbm{1}_{B_n} V_n P_n \le \mathbbm{1}_{B_n} C_2^2 (t - t_n) \\
&\Longrightarrow\;
\mathbbm{1}_{B_n} V_n \le 2 C_2^2 (t - t_n)\, \mathbbm{1}_{B_n} \\
&\Longrightarrow\;
\mathbbm{1}_{B_n} (x(t_n) - z^*)^2 
\le 2 C_2^2 (t - t_n)\, \mathbbm{1}_{B_n},
\end{align*}
where  $B_n := \{P_n > \frac{1}{2}\}$ is  $\mathcal{F}_{t_n}$-measurable. 
Applying It\^o's formula and using the upper bound $C_2^2$ of $\hat{\sigma}$, we obtain 
\begin{align}\label{4-moment}
\mE_{t_n}[(x(t) - x(t_n))^4] &= 6\mE_{t_n}\left[\int_{t_n}^t (x(s) - x(t_n))^2 |\hat{\sigma}(s, x(s))|^2 \md s\right]\le 6 C_2^2 \int_{t_n}^t \mE_{t_n}[(x(s) - x(t_n))^2] \md s \nonumber\\
&\le 6 C_2^4 \int_{t_n}^t (s - t_n) \md s = 3 C_2^4 (t - t_n)^2.
\end{align}
It then follows that
\begin{align}\label{variance inequality}
V_n = \mE_{t_n}[(x(t) - z^*)^2 \mathbbm{1}_{A^c}] &\le \left(\mE_{t_n}[(x(t) - z^*)^4]\right)^{\frac{1}{2}} (1 - P_n)^{\frac{1}{2}}.
\end{align}
Using $(a+b)^4 \le 8a^4 + 8b^4$, we have
\begin{align*}
&\mE_{t_n}[(x(t) - z^*)^4] \le 8\mE_{t_n}[(x(t) - x(t_n))^4] + 8(x(t_n) - z^*)^4\\
\Longrightarrow &
\mathbbm{1}_{B_n} \mE_{t_n}[(x(t) - z^*)^4] \le 24 C_2^4 (t - t_n)^2 \mathbbm{1}_{B_n} + 32 C_2^4 (t - t_n)^2 \mathbbm{1}_{B_n} = 56 C_2^4 (t - t_n)^2 \mathbbm{1}_{B_n}.
\end{align*}
Substituting this back into \eqref{variance inequality} yields
\begin{align*}
\mathbbm{1}_{B_n} c_3^2 (t - t_n)\leq \mathbbm{1}_{B_n} V_n &\le \mathbbm{1}_{B_n} \sqrt{56} C_2^2 (t - t_n) \sqrt{1 - P_n}.
\end{align*}
Dividing both sides by $(t - t_n)$ and rearranging the terms lead to
\begin{align*}
\mathbbm{1}_{B_n} (1 - P_n) &\ge \mathbbm{1}_{B_n} \frac{c_3^4}{56 C_2^4} \ge 0.
\end{align*}
Note that $P_n \to 1$ a.s. on $A$, for almost every $\omega \in A$, there exists $N(\omega)$ such that $\omega \in B_n$ for all $n > N(\omega)$. Hence,
$(1-P_n(\omega))\geq \frac{c_3^4}{56 C_2^4} > 0$ for $n>N(\omega)$, which contradicts $P_n \to 1$ a.s.
Therefore, $\mathbb{P}(x(t) = z^*) = 0$,  thus $\mathbb{P}(X(t) = x) = 0$.
\end{proof}

\section{Existence of Mean-field Equilibrium}\label{Mean-field Consistency and Equilibrium}

Building upon the existence and characterization of intra-personal equilibrium in Theorems \ref{necessary} and \ref{sufficiency} of the last section, we turn to addressing the existence of mean-field equilibrium for the time-inconsistent MFG, i.e., we verify that the mean-field consistency condition holds. To emphasize the dependence on $m(\cdot)$, throughout this section we incorporate $m$ into the notations of the relevant processes and functions, such as $\gamma(t,x;m)$, thereby distinguishing them from those in the previous section where $m$ was treated as fixed.
\vskip 2pt \noindent
We propose to first investigate the existence of the mean-field equilibrium for the auxiliary system under smoothing regularization, and then recover the mean-field equilibrium in the original system by passing to the limit as $\varepsilon \to 0$. 

\subsection{Mean-field equilibrium in the regularized system}
\vskip 2pt
For $\epsilon>0$, let $u^{\varepsilon}(\cdot,\cdot;m)\in C^{1,2}(Q_T)$ be the unique solution of the regularized PDE \eqref{eq:regularized} obtained in Theorem \ref{epsilon-existence}. Consider the following auxiliary SDE:
\begin{align}\label{eq:x:epsilon}
\begin{cases}
\md x^{\varepsilon}(t;m) = \gamma^{\varepsilon}\big(t, P_{0}(t)^{-1}u^{\varepsilon}(t, x^{\varepsilon}(t;m); m); m\big) \lambda(t)^\top \md W(t), \\
u^{\varepsilon}(0, x^{\varepsilon}(0;m); m) = P_{0}(0)\xi.
\end{cases}
\end{align}
By Theorem \ref{epsilon-existence}, $u^\epsilon(\cdot,\cdot;m)$ has
bounded first-order spatial derivative. Consequently, the diffusion coefficient in \eqref{eq:x:epsilon} is globally Lipschitz continuous in the spatial argument, and hence the SDE \eqref{eq:x:epsilon} admits a unique strong solution, denoted by the auxiliary process $\{x^{\varepsilon}(t;m)\}$.
Define the wealth process $\{X^{\varepsilon}(t;m)\}$ by 
\begin{align*}
X^{\varepsilon}(t;m) := P_{0}(t)^{-1}u^{\varepsilon}(t, x^{\varepsilon}(t;m); m),
\end{align*}
the two processes $P^{\varepsilon}(t;m)$ and $Q^{\varepsilon}(t;m)$ respectively by
\begin{align}
P^{\varepsilon}(t;m) &:= x^{\varepsilon}(t;m) - u^{\varepsilon}\big(t, x^{\varepsilon}(t;m); m\big),\nonumber \\
Q^{\varepsilon}(t;m) &:= \gamma^{\varepsilon}\big(t, X^{\varepsilon}(t;m); m\big) \left( 1 - \frac{\partial u^{\varepsilon}}{\partial x}\big(t, x^{\varepsilon}(t;m); m\big) \right)\lambda(t).\label{eq:epsilon:Q}
\end{align}
Then, using the same arguments as in Proposition \ref{FBSDE-existence}, we have the following Proposition \ref{eq:fbsde:epsilon}.

\begin{proposition}\label{eq:fbsde:epsilon}
$(X^{\varepsilon}, P^{\varepsilon}, Q^{\varepsilon})$ satisfies the following FBSDE system:
\begin{align}\label{eq:fbsde:epsilon:xpq}
\left\{
\begin{aligned}
&\md X^{\varepsilon}(t;m) = \left( r(t)X^{\varepsilon}(t;m) - \frac{1}{P_0(t)} \left[ Q^{\varepsilon}(t;m) - \gamma^{\varepsilon}\big(t, X^{\varepsilon}(t;m); m\big)\lambda(t) \right]^{\top} \lambda(t) \right) \md t \\
&\qquad \quad\quad\quad - \frac{1}{P_0(t)} \left[ Q^{\varepsilon}(t;m) - \gamma^{\varepsilon}\big(t, X^{\varepsilon}(t;m); m\big)\lambda(t) \right]^{\top} \md W(t), \\
&\md P^{\varepsilon}(t;m) = -\left( \gamma^{\varepsilon}\big(t, X^{\varepsilon}(t;m); m\big) |\lambda(t)|^2 - \lambda^{\top}(t)Q^{\varepsilon}(t;m) \right) \md t + Q^{\varepsilon}(t;m)^{\top} \md W(t), \\
&X^{\varepsilon}(0;m) = \xi, \quad P^{\varepsilon}(T;m) = 0.
\end{aligned}
\right.
\end{align}
\end{proposition}
\noindent
To facilitate the fixed-point argument for mean-field consistency, we first construct a compact, convex subset $\mathcal{K}\subset C([0,T])$ as follows.

\begin{proposition}\label{prop:K-construct}
Define the regularized mapping $\Phi^{\varepsilon}: C([0,T]) \to C([0,T])$ by $\Phi^{\varepsilon}(m)(t) := \mathbb{E}[X^{\varepsilon}(t; m)]$. Then there exists a nonempty, closed, bounded, convex, and equicontinuous subset $\mathcal{K}\subset C([0,T])$, independent of $\varepsilon$, such that $\Phi^{\varepsilon}(\mathcal{K}) \subseteq \mathcal{K}$.
\end{proposition}
\begin{proof}
Fix $m\in C([0,T])$. For notational simplicity, let $ \tilde{m}^{\varepsilon}(\cdot)$ denote $\Phi^{\varepsilon}(m)(\cdot)$.
 Taking expectations on both sides of the first SDE in \eqref{eq:fbsde:epsilon:xpq}, we obtain the following ODE:
\begin{align*}
\frac{\md \tilde{m}^{\varepsilon}}{\md t}(t) = r(t)\tilde{m}^{\varepsilon}(t) - \frac{1}{P_{0}(t)}\lambda(t)^{\top}\mathbb{E}[Q^{\varepsilon}(t;m)] + \frac{1}{P_{0}(t)}|\lambda(t)|^{2}\mathbb{E}[\gamma^{\varepsilon}\big(t, X^{\varepsilon}(t;m); m \big)].
\end{align*}
The last two terms can be estimated by
\begin{align*}
&\left| - \frac{1}{P_{0}(t)}\lambda(t)^{\top}\mathbb{E}[Q^{\varepsilon}(t;m)] + \frac{1}{P_{0}(t)}|\lambda(t)|^{2}\mathbb{E}[\gamma^{\varepsilon}\big(t, X^{\varepsilon}(t;m); m \big)] \right| \\
\leq &\, \frac{\lambda_{\max}}{P_{\min}}|\lambda(t)| \mE\left[\gamma^{\varepsilon}\big(t, X^{\varepsilon}(t;m); m \big) \left( 1 + \left| \frac{\partial u^{\varepsilon}}{\partial x} \right| \right)\right]+ \frac{\lambda_{\max}^{2}}{P_{\min}} \gamma_2 \\
\leq &\, \frac{\lambda^2_{\max}}{P_{\min}}  \gamma_2 (2 + M_1)=:C_3.
\end{align*}
As a result, $\left| \frac{\md \tilde{m}^{\varepsilon}}{\md t}(t) \right| \le r_{\max} |\tilde{m}^{\varepsilon}(t)| + C_3$.
Thanks to Gronwall inequality on $[0,t]$, it holds that
\begin{align*}
&|\tilde{m}^{\varepsilon}(t)| \le \left( |\mathbb{E}[\xi]| + C_3 T \right) e^{r_{\max} T} =: C_4\quad \forall\ t\in[0,T]\\
\Longrightarrow & \left| \frac{\md \tilde{m}^{\varepsilon}}{\md t}(t) \right| \le r_{\max} C_4 + C_3 =: C_5\quad \forall\  t\in[0,T].
\end{align*}
Define the subset $\mathcal{K} \subset C([0,T])$ by
\begin{align*}
\mathcal{K} := \Big\{ m \in C([0,T]) \ \Big| &\ m(0) = \mathbb{E}[\xi], \ \|m\|_\infty \le C_4,\\
&\  |m(t_{1}) - m(t_{2})| \le C_5 |t_{1} - t_{2}| \quad \ \forall \   t_{1},\  t_{2} \in [0,T]\Big\}.
\end{align*}
It is straightforward to verify that $\mathcal{K}$ is nonempty, closed and convex. Moreover, we have  $\Phi^{\varepsilon}(\mathcal{K}) \subseteq \mathcal{K}$ by the bounds $C_4$ and $C_5$ derived for $\tilde{m}^{\epsilon}(t)$ and its derivative. Finally, $\mathcal{K}$ is uniformly bounded and equicontinuous; thus, by invoking the Arzelà–Ascoli theorem, it is relatively compact in $C([0,T])$. In conclusion, $\mathcal{K}$ is compact.
\end{proof}

To eventually prove the continuity of the mapping $\Phi^{\epsilon}$, we first show the continuous dependence of the state process on the population aggregation.

\begin{lemma}\label{lem:p-convergence}
Let $\{m_{k}\}_{k=1}^{\infty} \subset \mathcal{K}$ satisfy $\lim\limits_{k \to \infty} \|m_{k} - m\|_{\infty} = 0$.
Then $X^{\varepsilon}(t;m_k)$ converges uniformly in probability to $X^{\varepsilon}(t;m)$ on $[0,T]$, i.e., for any $\eta > 0$,
\begin{align}\label{eq:convergence:X:mk}
\lim_{k \to \infty} \mathbb{P}\left( \sup_{t \in [0,T]} |X^{\varepsilon}(t;m_k) - X^{\varepsilon}(t;m)| > \eta \right) = 0.
\end{align}
\end{lemma}
\begin{proof}
\textbf{Step 1: Uniform convergence of $u^{\varepsilon}(\cdot,\cdot; m_k)$ to $u^{\varepsilon}(\cdot,\cdot; m)$ on compact sets. } 

For notational convenience, define $u^{\varepsilon}_{k}(t,x) := u^{\varepsilon}(t,x; m_k)$. Applying the same argument as in the proof of Theorem \ref{existence-u}, we can find a subsequence $\{u^{\varepsilon}_{k_j}\}$ and a function $$\tilde{u}^\epsilon\in W^{1,2}_{2,\mathrm{loc}}(Q_T)\cap C^{\alpha/2, \alpha}_{\mathrm{loc}}(Q_T)\quad \text{for all }\alpha\in\left(0,\frac{1}{2}\right)$$
such that, on every compact subset $K\subset Q_T$, $u^{\varepsilon}_{k_j}$ converges to $\tilde{u}^\epsilon$ weakly in $W^{1,2}_2(K)$ and strongly in $C^{\alpha/2, \alpha}(K)$  for any $\alpha < \frac{1}{2}$. Moreover, $|\tilde{u}^\epsilon(t,x)-x|\leq M(T-t)$ for all $(t,x)\in Q_T$. Note that
\begin{align}\label{eq:gamma:epsilon:m}
|\gamma^\varepsilon(t, w; m_k) - \gamma^\varepsilon(t, w; m)| \le &\int_{\mathbb{R}} |\gamma(t, y; m_k) - \gamma(t, y; m)| \omega_\varepsilon(w - y) \md y \nonumber \\
\leq & \frac{\|\omega\|_{\infty}}{\varepsilon}|\gamma_1-\gamma_2| \|m_k - m\|_\infty \to 0,\quad \text{as } k\to \infty.
\end{align}
Combining the above uniform convergence of $u^{\varepsilon}_{k_j}$ on compact sets, we see that the coefficients  $D^{\varepsilon}(t, u^{\varepsilon}_{k_j}; m_{k_j})$ and $V^{\varepsilon}(t, u^{\varepsilon}_{k_j}; m_{k_j})$ converge uniformly (on compact sets) to $D^{\varepsilon}(t, \tilde{u}^\epsilon; m)$ and $V^{\varepsilon}(t, \tilde{u}^\epsilon; m)$, respectively. Similarly, arguing as in the proof of Proposition \ref{lem:dividing_curve}, we conclude that $\tilde{u}^\epsilon$ is a strong solution of the regularized PDE \eqref{eq:regularized}.

On the other hand, by Theorem \ref{epsilon-existence}, the regularized PDE \eqref{eq:regularized} admits a unique classical solution $u^{\varepsilon}:=u^{\varepsilon}(\cdot, \cdot; m)$. It follows from Lemma \ref{lem:uniqueness_strong} that $\tilde{u}^\epsilon \equiv u^{\varepsilon}$.  Therefore, $u^\epsilon_{k_j}$ converges to $u^\epsilon$ in $C(K)$ for any compact set $K \subset Q_T$. Moreover, by a similar argument, any convergent subsequence must have the same limit $u^\epsilon$. This implies that the entire sequence $\{u^{\varepsilon}_{k}\}$ converges uniformly on $K$ to $u^{\varepsilon}$:
\begin{align*}
\lim_{k \to \infty} \sup_{(t,x) \in K} |u^{\varepsilon}_{k}(t,x) - u^{\varepsilon}(t,x)| = 0.
\end{align*}

Recalling $|u^{\varepsilon}_{k}(t,x) - x| \le M(T-t)$, we deduce that for the curves $b^{\varepsilon}(t)$ and $b^{\varepsilon}_{k}(t)$ determined by $u^{\varepsilon}(t, b^{\varepsilon}(t)) = P_{0}(t)m(t)$ and $u^{\varepsilon}_{k}(t, b^{\varepsilon}_{k}(t)) = P_{0}(t)m_{k}(t)$, it holds that
\begin{align*}
    |b^{\varepsilon}_{k}(t) - P_{0}(t)m_{k}(t)| &= |b^{\varepsilon}_{k}(t) - u^{\varepsilon}_{k}(t, b^{\varepsilon}_{k}(t))| \le M(T-t) \le MT.
\end{align*}
This implies $|b^{\varepsilon}_{k}(t)| \le |P_{0}(t)m_{k}(t)| + MT$. There exists a constant $R^b> 0$ independent of $\epsilon$ such that $|b^{\varepsilon}_{k}(t)| \le R^b$ and $|b^{\varepsilon}(t)| \le R^b$ for all $k$ and $t \in [0,T]$ due to the continuity of $P_0(t)$ on $[0,T]$ and the uniform boundedness of $\{m_k\}$.
Furthermore, by the bi-Lipschitz property of $u^{\varepsilon}$ established in Proposition \ref{lma:bound:'},
\begin{align*}
|u^{\varepsilon}(t, b^{\varepsilon}_{k}(t)) - u^{\varepsilon}(t, b^{\varepsilon}(t))| \ge \delta_1 |b^{\varepsilon}_{k}(t) - b^{\varepsilon}(t)|.
\end{align*}
In view that $u^{\varepsilon}_{k}(t, b^{\varepsilon}_{k}(t)) - u^{\varepsilon}(t, b^{\varepsilon}(t)) = P_{0}(t)m_{k}(t) - P_{0}(t)m(t)$, the triangle inequality gives
\begin{align*}
\delta_1 |b^{\varepsilon}_{k}(t) - b^{\varepsilon}(t)|
&\le |u^{\varepsilon}(t, b^{\varepsilon}_{k}(t)) - u^{\varepsilon}_{k}(t, b^{\varepsilon}_{k}(t))| + P_{0}(t)|m_{k}(t) - m(t)|\\
& \le \sup_{(t,x) \in[0,T]\times[-R^b,R^b]} |u^{\varepsilon}_{k}(t,x) - u^{\varepsilon}(t,x)| + P_{\max}\|m_{k} - m\|_{\infty}\to 0 \quad \text{as }k\to \infty.
\end{align*}
\noindent
\textbf{Step 2: Uniform convergence of  $X^{\varepsilon}(\cdot;m_k)$  to $X^{\varepsilon}(\cdot;m)$.} 

\noindent
For notational simplicity, we define $x^{\varepsilon}_{k}(t) := x^{\varepsilon}(t;m_k)$ and $x^{\varepsilon}(t) := x^{\varepsilon}(t;m)$. Moreover, set $\hat{\sigma}^{\varepsilon}_{k}(t, x) := \gamma^{\varepsilon}\big(t, P_{0}(t)^{-1}u^{\varepsilon}_{k}(t, x) ; m_{k}\big)\lambda(t)$ and $\hat{\sigma}^{\varepsilon}(t, x) := \gamma^{\varepsilon}\big(t, P_{0}(t)^{-1}u^{\varepsilon}(t, x) ; m\big)\lambda(t)$. Then, recalling that $c_3 := \gamma_1 \lambda_{\min} > 0$ and $C_2 := \gamma_2 \lambda_{\max} < \infty$, we have
$c_3^2\leq |\hat{\sigma}^{\varepsilon}_{k}(t,x)|^2\leq C_2^2$. 
Observe that
\begin{align}\label{eq:hat:sigma:dif}
|\hat{\sigma}^{\varepsilon}_{k}(s, x) - \hat{\sigma}^{\varepsilon}(s, x)| &= |\lambda(s)| \cdot \left| \gamma^{\varepsilon}\big(s, P_{0}(s)^{-1}u^{\varepsilon}_{k}(s, x) ; m_{k}\big) - \gamma^{\varepsilon}\big(s, P_{0}(s)^{-1}u^{\varepsilon}(s, x) ; m\big) \right| \nonumber\\
&\le \lambda_{\max} \left(L_{\gamma, \varepsilon} P_{\min}^{-1} |u^{\varepsilon}_{k}(s, x) - u^{\varepsilon}(s, x)| + \|\gamma^{\varepsilon}(s, \cdot ; m_{k}) - \gamma^{\varepsilon}(s, \cdot ; m)\|_\infty\right).
\end{align}
As shown in \eqref{eq:gamma:epsilon:m}, we have $\|\gamma^{\varepsilon}(s, \cdot ; m_{k}) - \gamma^{\varepsilon}(s, \cdot ; m)\|_{\infty} \to 0$ as $\|m_k - m\|_{\infty} \to 0$. Combining with the uniform convergence $u^{\varepsilon}_{k} \to u^{\varepsilon}$ on $ [0,T] \times [-N, N]$ and using \eqref{eq:hat:sigma:dif} yields, for any $N > 0$,
\begin{align}\label{eq:convergence:delta:N}
\delta^{\varepsilon}_{k}(N) := \sup_{s \in [0,T], |x| \le N} |\hat{\sigma}^{\varepsilon}_{k}(s, x) - \hat{\sigma}^{\varepsilon}(s, x)| \,\,\to 0 \quad \text{as }k\to\infty.
\end{align}
Before analyzing the dynamics, we must control the difference in initial states. Using $ P_0(0)\xi = u^{\varepsilon}_{k}(0, x^{\varepsilon}_{k}(0)) = u^{\varepsilon}(0, x^{\varepsilon}(0))$, 
and arguing as in \eqref{delta_0}, we obtain
\begin{align*}
\delta_1 |x^{\varepsilon}_{k}(0) - x^{\varepsilon}(0)|\leq |u^{\varepsilon}(0, x^{\varepsilon}_{k}(0)) - u^{\varepsilon}_{k}(0, x^{\varepsilon}_{k}(0))|.
\end{align*}
Moreover, the bound $|u^{\varepsilon}_{k}(0,x) - x| \le MT$ implies $|P_0(0)\xi - x^{\varepsilon}_{k}(0)| \le MT$. $|x^{\varepsilon}_{k}(0)| \le P_0(0) \|\xi\|_{\infty} + MT =: C_6$ almost surely holds for all $k$ because $\xi \in L^\infty$. Thus $ 
\lim\limits_{k \to \infty} |x^{\varepsilon}_{k}(0) - x^{\varepsilon}(0)| = 0$ almost surely by the uniform convergence of $u^{\varepsilon}_{k} \to u^{\varepsilon}$ uniformly on $\{0\} \times [-C_6, C_6]$.
Using the dominated convergence theorem  yields
\begin{align}\label{eq:convergence:x:0}
    \lim\limits_{k \to \infty} \mathbb{E}[|x^{\varepsilon}_{k}(0) - x^{\varepsilon}(0)|^2] = 0.
\end{align}

We now estimate the difference between the processes up to a localization stopping time. Define $\tau_{k, N}^\epsilon := \inf\{t \ge 0 : |x_k^\epsilon(t)| > N\}$. Applying Doob's maximal inequality and It\^{o}'s isometry, we obtain that for any $t \in [0,T]$,
\begin{align*}
&\mathbb{E}\left[ \sup_{s \le t} |x^{\varepsilon}_{k}(s \wedge \tau^{\varepsilon}_{k, N}) - x^{\varepsilon}(s \wedge \tau^{\varepsilon}_{k, N})|^{2} \right] -2\mathbb{E}\left[|x^{\varepsilon}_{k}(0) - x^{\varepsilon}(0)|^{2}\right]\\
\le &\  2\mathbb{E}\left[ \sup_{s \le t} \left| \int_{0}^{s \wedge \tau^{\varepsilon}_{k, N}} \left(\hat{\sigma}^{\varepsilon}_{k}(u, x^{\varepsilon}_{k}(u)) - \hat{\sigma}^{\varepsilon}(u, x^{\varepsilon}(u))\right)^\top \md W(u) \right|^{2} \right] \\
\le &\  8 \mathbb{E} \int_{0}^{t \wedge \tau^{\varepsilon}_{k, N}} |\hat{\sigma}^{\varepsilon}_{k}(u, x^{\varepsilon}_{k}(u)) - \hat{\sigma}^{\varepsilon}(u, x^{\varepsilon}(u))|^{2} \md u \\
\le &\  16 \mathbb{E} \int_{0}^{t \wedge \tau^{\varepsilon}_{k, N}} \Big( |\hat{\sigma}^{\varepsilon}_{k}(u, x^{\varepsilon}_{k}(u)) - \hat{\sigma}^{\varepsilon}(u, x^{\varepsilon}_{k}(u))|^{2} + |\hat{\sigma}^{\varepsilon}(u, x^{\varepsilon}_{k}(u)) - \hat{\sigma}^{\varepsilon}(u, x^{\varepsilon}(u))|^{2} \Big) \md u\\
\le &\ 16T [\delta^{\varepsilon}_{k}(N)]^{2} + 16 L_{\varepsilon,T}^{2} \int_{0}^{t} \mathbb{E}\left[ \sup_{s \le u} |x^{\varepsilon}_{k}(s \wedge \tau^{\varepsilon}_{k, N}) - x^{\varepsilon}(s \wedge \tau^{\varepsilon}_{k, N})|^{2} \right] \md u,
\end{align*}
where the last inequality follows from the fact that 
 $\hat{\sigma}^{\varepsilon}(t,x)$ is globally Lipschitz in $x$ with constant $L_{\varepsilon,T} := \lambda_{\max} L_{\gamma, \varepsilon} P_{\min}^{-1} M_1$. Using Gronwall’s inequality, \eqref{eq:convergence:delta:N} and \eqref{eq:convergence:x:0}, we have that for any $N>0$,
\begin{align*}
&\mathbb{E}\left[ \sup_{t \le T} |x^{\varepsilon}_{k}(t \wedge \tau^{\varepsilon}_{k, N}) - x^{\varepsilon}(t \wedge \tau^{\varepsilon}_{k, N})|^{2} \right] \\
\le &\  \left( 2\mathbb{E}[|x^{\varepsilon}_{k}(0) - x^{\varepsilon}(0)|^{2}] + 16T [\delta^{\varepsilon}_{k}(N)]^{2} \right) \mathrm{e}^{16 L_{\varepsilon,T}^{2} T}\,\,\to 0, \quad \text{as }k\to \infty.
\end{align*}
Hence, for any $\eta>0$, 
\begin{align}\label{eq:convergence:x:mk}
\lim_{k \to \infty} \mathbb{P}\left( \sup_{t \le T} |x^{\varepsilon}_{k}(t \wedge \tau^{\varepsilon}_{k, N}) - x^{\varepsilon}(t \wedge \tau^{\varepsilon}_{k, N})| > \eta \right) = 0.
\end{align}
To lift the localization, we need to bound the tail probability of the stopping time. To achieve this, we first establish a uniform fourth-moment estimate. 
Based on the Burkholder-Davis-Gundy  (BDG) inequality, there exists a universal constant $C$ such that
\begin{align*}
&\mathbb{E}\left[ \sup_{t \in [0,T]} \left| \int_{0}^{t} \left(\hat{\sigma}^{\varepsilon}_{k}(s, x^{\varepsilon}_{k}(s))\right)^\top \md W(s) \right|^{4} \right]\leq C \mathbb{E}\left[\left(\int_0^T\left|\hat{\sigma}^{\varepsilon}_{k}(s, x^{\varepsilon}_{k}(s))\right|^2\md s\right)^2\right]
\leq C C_2^4T^2.
\end{align*}
Then,  using $(a+b)^{4} \le 8(a^{4}+b^{4})$, we obtain\begin{align*}
\mathbb{E}\left[ \sup_{t \in [0,T]} |x^{\varepsilon}_{k}(t)|^{4} \right] \le 8 \mathbb{E}\left[|x^{\varepsilon}_{k}(0)|^{4}\right] + 8CC_2^{4} T^{2}\leq 8C_6^4+8CC_2^{4} T^{2}=:C_7.
\end{align*}
By Markov’s inequality, for any $N > 0$, it holds that $\mathbb{P}\left( \sup_{t \in [0,T]} |x^{\varepsilon}_{k}(t)| > N \right) \le \frac{C_7}{N^{4}}$ and hence $\mathbb{P}(\tau^{\varepsilon}_{k, N} \le T) \le \frac{C_7}{N^{4}}$. Combining the localized estimate with the tail probability, we can now conclude the uniform convergence in probability. Observe that
\begin{align*}
|X^{\varepsilon}(t;m_k) - X^{\varepsilon}(t;m)|&=| P_{0}(t)^{-1}u^{\varepsilon}_{k}(t, x^{\varepsilon}_{k}(t))-P_{0}(t)^{-1}u^{\varepsilon}(t, x^{\varepsilon}(t))|\\
&\le P_{\min}^{-1} |u^{\varepsilon}_{k}(t, x^{\varepsilon}_{k}(t)) - u^{\varepsilon}(t, x^{\varepsilon}_{k}(t))| + P_{\min}^{-1} M_1 |x^{\varepsilon}_{k}(t) - x^{\varepsilon}(t)|.
\end{align*}
There exists some $K=K(N) > 0$ such that for all $k > K$,
\begin{align*}
\sup_{(t,x) \in [0,T]\times[-N,N]} |u^{\varepsilon}_{k}(t,x) - u^{\varepsilon}(t,x)| \le \frac{\eta P_{\min}}{2}
\end{align*}
as $u^{\varepsilon}_{k} \to u^{\varepsilon}$ uniformly on $[0,T]\times[-N,N]$.
Thus, for any $k>K$,
\begin{align*}
&\mathbb{P}\left( \sup_{t \in [0,T]} |X^{\varepsilon}(t;m_k) - X^{\varepsilon}(t;m)| > \eta \right) \\
\le &\  \mathbb{P}(\tau^{\varepsilon}_{k, N} \le T) + \mathbb{P}\left( \sup_{t \in [0,T]} |X^{\varepsilon}(t;m_k) - X^{\varepsilon}(t;m)| > \eta, \, \tau^{\varepsilon}_{k, N} > T \right)\\ \leq &\  \frac{C_7}{N^{4}}+ \mathbb{P}\left( \sup_{t \in [0,T]} |x^{\varepsilon}_{k}(t \wedge \tau^{\varepsilon}_{k, N}) - x^{\varepsilon}(t \wedge \tau^{\varepsilon}_{k, N})| > \frac{\eta P_{\min}}{2 M_1} \right).
\end{align*}By taking $\limsup\limits_{k \to \infty}$, the second term vanishes in view of \eqref{eq:convergence:x:mk}. Hence,
\begin{align*}
\limsup_{k \to \infty} \mathbb{P}\left( \sup_{t \in [0,T]} |X^{\varepsilon}(t;m_k) - X^{\varepsilon}(t;m)| > \eta \right) \le \frac{C_7}{N^{4}}.
\end{align*}
Letting $ N  \to \infty $, we complete the proof.
\end{proof}

Building upon the uniform convergence in probability of the state processes established in Lemma \ref{lem:p-convergence}, we next verify the continuity of the mean-field mapping  $\Phi^{\varepsilon}$.
\begin{proposition}\label{prop:Phi-continuous}
The mapping $\Phi^{\varepsilon}: \mathcal{K} \to \mathcal{K}$ is continuous on $(C([0,T]), \|\cdot\|_{\infty})$. 
\end{proposition}
\begin{proof}
Define
\begin{align*}
    \Sigma^{\varepsilon}_k(s):=&-\frac{1}{P_0(s)} \left[ Q^{\varepsilon}(s;m_k) - \gamma^{\varepsilon}\big(s, X^{\varepsilon}(s;m_k); m_k\big)\lambda(s) \right]\\
    =&P_0^{-1}(s)\gamma^{\varepsilon}(s, X^{\varepsilon}(s;m_k);m_k)\frac{\partial u^{\varepsilon}_{k}}{\partial x}(s,x^{\varepsilon}_{k}(s))\lambda(s).
\end{align*}
Then there exist two constants $c_4>0$ and $C_8>0$, independent of $k$ and $\varepsilon$, such that  $0 < c_4 \le \left| \Sigma^{\varepsilon}_k(t)^\top \right| \le C_8 < \infty$. 
For $m_{k} \in \mathcal{K}$, the process $\{X^{\varepsilon}(t;m_k)\}$ satisfies 
\begin{align*}
    X^{\varepsilon}(t;m_k) &= \xi + \int_{0}^{t} \left( r(s)X^{\varepsilon}(s;m_k) + \Sigma^{\varepsilon}_k(s)^\top \lambda(s) \right) \md s + \int_{0}^{t} \Sigma^{\varepsilon}_k(s)^\top \md W(s).
\end{align*}
Applying $(a+b+c)^{4} \le 27(a^{4} + b^{4} + c^{4})$ yields
\begin{align*}
|X^{\varepsilon}(t;m_k)|^{4} \le 27| \xi|^{4}  + 27\left| \int_{0}^{t} \left( r(s)X^{\varepsilon}(s;m_k) + \Sigma^{\varepsilon}_k(s)^\top \lambda(s) \right) \md s \right|^{4} + 27\left| \int_{0}^{t} \Sigma^{\varepsilon}_k(s)^\top \md W(s) \right|^{4}.
\end{align*}
Thanks to H\"{o}lder's inequality and $(a+b)^{4} \le 8(a^{4}+b^{4})$, we have 
\begin{align*}
\left| \int_{0}^{t} \left( r(s)X^{\varepsilon}(s;m_k) + \Sigma^{\varepsilon}_k(s)^\top \lambda(s) \right) \md s \right|^{4}
&\le t^{3} \int_{0}^{t} \left| r(s)X^{\varepsilon}(s;m_k) + \Sigma^{\varepsilon}_k(s)^\top \lambda(s) \right|^{4} \md s \\
&\le 8T^{3} \int_{0}^{t} \left( r_{\max}^{4} |X^{\varepsilon}(s;m_k)|^{4} + \lambda_{\max}^{4} C_8^{4} \right) \md s \\
&\le 8T^{4} \lambda_{\max}^{4} C_8^{4} + 8T^{3} r_{\max}^{4} \int_{0}^{t} |X^{\varepsilon}(s;m_k)|^{4} \md s.
\end{align*}
Moreover, by BDG inequality, there exists a universal constant $C > 0$ such that
\begin{align*}
&\mathbb{E}\left[ \sup_{\tau \le t} \left| \int_{0}^{\tau} \Sigma^{\varepsilon}_k(s)^\top \md W(s) \right|^{4} \right] \le C \, \mathbb{E}\left[ \left( \int_{0}^{t} \left| \Sigma^{\varepsilon}_k(s)^\top \right|^{2} \md s \right)^{2} \right] \le C C_8^{4} T^{2}.
\end{align*}
Let $M^{\varepsilon}_{k}(t) := \mathbb{E}\left[ \sup\limits_{s \le t} |X^{\varepsilon}(s;m_k)|^{4} \right]$. The last three inequalities yield
\begin{align*}
M^{\varepsilon}_{k}(t) &\le C_9 + C_{10}\int_{0}^{t} M^{\varepsilon}_{k}(s) \md s,
\end{align*}
where $C_9 := 27 \mathbb{E}[|\xi|^{4}] + 216 T^{4} \lambda_{\max}^{4} C_8^{4} + 27 C C_8^{4} T^{2}$ and $C_{10} := 216 T^{3} r_{\max}^{4}$  are independent of $k$ and $\varepsilon$. 
It follows  from  Gronwall's inequality    that
\begin{align}\label{4-uniform}
\mathbb{E}\left[ \sup_{t \in [0,T]} |X^{\varepsilon}(t;m_k)|^{4} \right] = M^{\varepsilon}_{k}(T) \le C_9 e^{C_{10} T} =: C_{11}.
\end{align}
Given the uniform fourth-moment estimate, we are ready to verify the continuity of the mapping $\Phi^\epsilon$. To this end, define $Y^{\varepsilon}_{k} := \sup\limits_{t \in [0,T]} |X^{\varepsilon}(t;m_k) - X^{\varepsilon}(t;m)|$. We have  $Y^{\varepsilon}_{k} \xrightarrow{\mathbb{P}} 0$ by \eqref{eq:convergence:X:mk}. Applying $(a-b)^{4} \le 8(a^{4}+b^{4})$ yields
\begin{align*}
\mathbb{E}[(Y^{\varepsilon}_{k})^{4}] &
\le 8 \mathbb{E}\left[ \sup_{t \in [0,T]} |X^{\varepsilon}(t;m_k)|^{4} \right] + 8 \mathbb{E}\left[ \sup_{t \in [0,T]} |X^{\varepsilon}(t;m)|^{4} \right] \le 16 C_{11}.
\end{align*}
Then the sequence $\{Y^{\varepsilon}_{k}\}$ is uniformly integrable, leading to $\lim\limits_{k \to \infty} \mathbb{E}[Y^{\varepsilon}_{k}]=0$. By its definition, 
\begin{align*}
\|\Phi^{\varepsilon}(m_{k}) - \Phi^{\varepsilon}(m)\|_{\infty} = \sup_{t \in [0,T]} |\mathbb{E}[X^{\varepsilon}(t;m_k)] - \mathbb{E}[X^{\varepsilon}(t;m)]|\leq \mathbb{E}[Y^{\varepsilon}_{k}]\,\,\to 0 \quad \text{as }\,k\to\infty.
\end{align*}
Thus $\Phi^{\varepsilon}$ is continuous.
\end{proof}

Having both the compactness and continuity properties, we can resort to Schauder's fixed-point theorem to deduce the existence of a mean-field equilibrium for the regularized system.

\begin{theorem}
For any $\varepsilon > 0$, there exists some $m^{\varepsilon} \in C([0,T])$ satisfying
\begin{align}
m^{\varepsilon}(t) = \mathbb{E}[X^{\varepsilon}(t; m^{\varepsilon})]\quad \text{for all }t\in[0,T].\label{eq:fixedpoint:epsilon}
\end{align}
\end{theorem}

\begin{proof}
There exists a non-empty, compact, and convex subset $\mathcal{K} \subset C([0,T])$, independent of $\varepsilon$, such that $\Phi^{\varepsilon}(\mathcal{K}) \subseteq \mathcal{K}$ by Proposition \ref{prop:K-construct}. Moreover, the mapping $\Phi^{\varepsilon}: \mathcal{K} \to \mathcal{K}$ is continuous by Proposition \ref{prop:Phi-continuous}. Hence, there exists some $m^{\varepsilon} \in \mathcal{K}$ such that  $\Phi^{\varepsilon}(m^{\varepsilon}) = m^{\varepsilon}$ by the Schauder fixed-point theorem. This is equivalent to \eqref{eq:fixedpoint:epsilon} by the definition of $\Phi^{\varepsilon}$.
\end{proof}

\subsection{Existence of Mean-field Equilibrium via Convergence Analysis}

For any $\varepsilon > 0$, let $m^{\varepsilon} \in \mathcal{K}$ be a fixed point of $\Phi^{\varepsilon}$, i.e., $m^{\varepsilon} = \Phi^{\varepsilon}(m^{\varepsilon})$. Then there exists a subsequence $\{\varepsilon_{k}\}_{k=1}^{\infty}$ with $\lim\limits_{k \to \infty} \varepsilon_{k} = 0$ and a limit $m_{*} \in \mathcal{K}$ such that $m_k:=m^{\varepsilon_{k}}$  converges uniformly to $m_{*}$ in $C([0,T])$. 
We first state the main result of this paper.

\begin{theorem}\label{thm:mean-field}
There exists a solution $(X_*,P_*,Q_*)$ to the FBSDE system \eqref{eq:fbsde_sufficiency} associated with $m_*$. Moreover, $X_*(t)$ is atomless for all $t\in[0,T]$, {and the pair $(\bar{\pi}_*, m_*)$ is an open-loop mean-field equilibrium, where $\bar{\pi}_*$ is given by
\begin{align*}
    \bar{\pi}_*(t) &:= -\frac{1}{P_0(t)}(\sigma^\top(t))^{-1}[Q_*(t)-\gamma(t,X_*(t); m_*)\lambda(t)], \quad t\in[0,T].
\end{align*}}
\end{theorem}

\begin{remark}\label{rem:X-difference}
The process $X_*$ may differ from  $X(\cdot;m_*)$ obtained  in Proposition \ref{FBSDE-existence}; see Remark \ref{rmk:dis:m_*} for details.
\end{remark}

The rigorous proof of Theorem \ref{thm:mean-field} relies on some delicate convergence analysis of the regularized systems. To enhance readability, let us first outline three main steps of the proof:

\begin{itemize}
   \item \textbf{Step 1: Convergence of the PDE solutions  for the  Regularized PDE \eqref{eq:regularized} } 
   We show that a subsequence of solutions to the regularized PDEs \eqref{eq:regularized}, corresponding to $m^{\epsilon_k}$, converges to a limit function $u_*$. We then verify that this limit is a strong solution to the discontinuous PDE \eqref{eq:pde} associated with $m_{*}$, which further yields the existence of a solution to FBSDE \eqref{eq:fbsde_sufficiency}; see Proposition \ref{prop:m_*}.
    
   \item \textbf{Step 2: Tightness and Gy\"ongy-Krylov Conditions.} 
We then prove that the sequence of auxiliary state processes corresponding to $m^{\epsilon_k}$ is tight and satisfies the conditions required by the Gyöngy–Krylov lemma; see Lemma \ref{lem:tight} and Proposition \ref{lem:limit-sde}.

    \item \textbf{Step 3: Convergence and Verification of the Fixed-Point.} Using Gyöngy–Krylov lemma, we obtain convergence in probability of the state processes to a limit process. Together with the uniform integrability of the wealth processes, this allows us to pass to the limit and finally conclude the mean-field consistency condition.
\end{itemize}

Let $u_k(t,x):=u^{\epsilon_k}(t,x;m_k)$ be the solution of the regularized PDE \eqref{eq:regularized} with $\epsilon=\epsilon_k$ and $m=m_k$.
We begin our convergence analysis by extracting a convergent subsequence from the solutions of the regularized PDE problems.
\begin{proposition}\label{prop:m_*}
    There exists a subsequence of  $\{u_k\}$ and a function $u_*\in W^{1,2}_{2,\mathrm{loc}}(Q_T)\cap C^{\alpha/2, \alpha}_{\mathrm{loc}}(Q_T)$ for all $\alpha\in\left(0,1/2\right)$  such that, on every compact subset $K\subset Q_T$, this subsequence converges to $u_*$ weakly in $W^{1,2}_2(K)$ and strongly in $C^{\alpha/2, \alpha}(K)$  for any $\alpha < \frac{1}{2}$. It holds that $u_*$ satisfies $|u_*(t, x) - x| \le M(T - t)$ on $Q_T$, and $\delta_1 \leq \frac{\partial u_*}{\partial x} \leq M_1$ a.e. in $Q_T$. 
In addition, $u_*$ is a strong solution to the PDE \eqref{eq:pde} corresponding to $m = m_*$.
\end{proposition}

\begin{proof}
The existence of the subsequence, still denoted by $\{u_k\}$, the regularity of $u_*$, and the corresponding estimates follow from arguments analogous to those  of Theorem \ref{existence-u}. We will only prove that $u_*$ is a strong solution of the PDE \eqref{eq:pde}.

Let $\mathcal{N}_u := \{(t, x) \in Q_T : u_*(t,x) = P_0(t)m_*(t)\}$, where $\mathcal{N}_u$ has zero Lebesgue measure.
For any $(t,x) \in Q_T \setminus \mathcal{N}_u$, define $d_{t,x} := |u_*(t,x) - P_0(t)m_*(t)| > 0$. For sufficiently large $k$, we have $|u_k(t,x) - u_*(t,x)| < \frac{d_{t,x}}{3}$ and $P_0(t)|m_k(t) - m_*(t)| < \frac{d_{t,x}}{3}$, which yields
\begin{align*}
(u_k(t,x) - P_0(t)m_k(t))(u_*(t,x) - P_0(t)m_*(t)) \ge d_{t,x}^2 - d_{t,x}\left(\frac{2d_{t,x}}{3}\right) = \frac{d_{t,x}^2}{3} > 0.
\end{align*}
That is, for $k$ sufficiently large, $u_k(t,x) - P_0(t)m_k(t)$ and $u_*(t,x) - P_0(t)m_*(t)$ have the same sign. For sufficiently small $\epsilon_k$, the support of the mollifier lies entirely on one side of the discontinuity point of $\gamma$: on the right-hand side if  $u_*(t,x) - P_0(t)m_*(t)>0$, and on the left-hand side if $u_*(t,x) - P_0(t)m_*(t)<0$.
Consequently, we obtain pointwise convergence a.e. on $Q_T$:
\begin{align*}
\lim_{k\to\infty} D^{\epsilon_k}(t, u_k; m_k) = D(t, u_*; m_*), \quad \lim_{k\to\infty} V^{\epsilon_k}(t, u_k; m_k) = V(t, u_*; m_*).
\end{align*}
Thus, arguing similarly as in Proposition \ref{lem:dividing_curve}, we get that $u_*$ is a strong solution of \eqref{eq:pde}.
\end{proof}

\begin{remark}\label{rmk:dis:m_*}
In Proposition \ref{prop:m_*}, 
$u_*$ is obtained as the limit of a subsequence of the regularized sequence $\{u_k\}$, where for each $\epsilon_k$, the corresponding $m_k$ is different, unlike $u(\cdot,\cdot;m_*)$ constructed in Theorem \ref{existence-u}, where for each $\epsilon_k$, $m$ is fixed as $m_*$. We emphasize that the $u_*$ constructed in Proposition \ref{prop:m_*}, rather than the one constructed in Theorem \ref{existence-u}, is the one we ultimately seek.
\end{remark}

With $u_*$ at hand, we can invoke Proposition \ref{FBSDE-existence} to guarantee the existence of a solution $(X_*, P_*, Q_*)$ to the FBSDE system \eqref{eq:fbsde_sufficiency} corresponding to $m_*$. Furthermore, by Proposition \ref{atomless}, the wealth process $X_*(t)$ is atomless for all $t\in(0,T]$. Consequently, according to the sufficiency result in Theorem \ref{sufficiency}, the strategy $\bar{\pi}_*$ defined in Theorem \ref{thm:mean-field} is an equilibrium strategy given the population aggregation $m_*$, satisfying Condition (a) of Definition \ref{def:MFE}. 
Having established condition (a), the remainder of the proof focuses on condition (b). To conclude that $(\overline{\pi}_{*}, m_{*})$ is an open-loop mean-field equilibrium, it remains only to verify the mean-field consistency condition (b) of Definition \ref{def:MFE}, namely $m_*(t) = \mathbb{E}[X_*(t)]$.
We begin by introducing some notation and establishing auxiliary results.

First, let $b_k$ and $b_*$ be defined by $u_k(t,b_k(t))=R(t;m_k)$ and by $u_*(t,b_*(t))=R(t;m_*)$ respectively. The local uniform convergence $u_k \to u_*$ together with the uniform convergence $m_k \to m_*$ implies
 \begin{align*}
\lim_{k\to\infty} \|b_k - b_*\|_\infty = 0
\end{align*}
by using the same argument as in the proof of Lemma \ref{lem:p-convergence}.

Moreover, let $X_k$ denote the solution to the forward SDE in the FBSDE system \eqref{eq:fbsde:epsilon:xpq} in Proposition \ref{eq:fbsde:epsilon} with $\epsilon=\epsilon_k$ and $m=m_k$.  Similarly, let $x_k(t) := x^{\epsilon_k}(t; m_k)$ and define $\hat{\sigma}_k(t, x) := \hat{\sigma}^{\epsilon_k}(t, x; m_k)$.
There exists a unique $\mathcal{F}_0$-measurable bounded random variable $y_k$ such that $u_k(0, y_k) = P_0(0)\xi$ because $|u_k(0,x)-u_k(0,y)|\geq \delta_1|x-y|$.  Moreover, the uniform convergence $u_k(0, \cdot) \to u_*(0, \cdot)$ on compact sets implies that $y_k$ converges a.s. to a random variable $y_*$ satisfying $u_*(0, y_*) = P_0(0)\xi$ by an argument analogous to Lemma B.1 in \cite{cheng2025equilibrium}. The auxiliary process $x_k$ then satisfies
\begin{align*}
x_k(t) = y_k + \int_0^t \hat{\sigma}_k(s, x_k(s))^\top \md W(s).
\end{align*}
 For the limiting system, define
\[
\hat{\sigma}_*(t,x)
:=\gamma\!\left(t,P_0^{-1}(t)u_*(t,x);m_*\right)\lambda(t),
\]
and let $x_*$ be the strong solution of
\[
x_*(t)
=y_*+\int_0^t\hat{\sigma}_*(s,x_*(s))^\top\md W(s)
\]
obtained by the same argument as in
Proposition \ref{FBSDE-existence}.

To pass to the limit in the fixed-point equation $m_k(t)=\mathbb{E}[X_k(t)]$, we first establish tightness of the  sequence of probability measures induced by $\{x_k\}_{k\ge1}$, which ensures the existence of weakly convergent subsequences.

\begin{lemma}\label{lem:tight}
  The sequence of probability measures induced by $\{x_k\}_{k\ge1}$ on $C([0,T])$ is tight.
\end{lemma}
\begin{proof}
As established in Proposition \ref{atomless}, the diffusion coefficients $\hat{\sigma}_k$ are uniformly bounded independently of $k$ such that $0 < c_3 \le |\hat{\sigma}_k(s, y)| \le C_2 < \infty$. For $0 \le s < t \le T$, arguing as in \eqref{4-moment}, we have
\begin{align}\label{ana-4-moment}
\mathbb{E}\left[|x_k(t) - x_k(s)|^4\right]\leq 3 C_2^4 (t - s)^2.
\end{align}
For any $\gamma \in (0, 1/4)$, the expected H\"{o}lder seminorm $[x_k]_{\gamma}$ is uniformly bounded by a constant $C_{12}$ independent of $k$:
\begin{align*}
\sup_{k} \mathbb{E} \left[ [x_k]_{\gamma} \right] = \sup_{k} \mathbb{E} \left[ \sup_{0 \le s < t \le T} \frac{|x_k(t) - x_k(s)|}{|t-s|^{\gamma}} \right] \le C_{12} < \infty
\end{align*}
by the Kolmogorov-Chentsov theorem (see e.g., \citet*{revuz2013continuous}[Chapter 1, Theorem 2.1]).
The sequence $\{y_k\}$ is almost surely bounded by a constant $C_{13} > 0$ because $\xi \in L^\infty$. Define  $A_n := \{x \in C([0,T]) \mid |x(0)| \le C_{13}, \ [x]_{\gamma} \le n\}$. $A_n$ is relatively compact in $C([0,T])$ by the Arzel\`a--Ascoli theorem. 
\begin{align*}
\sup_{k} \mathbb{P}(x_k \notin A_n) = \sup_{k} \mathbb{P}([x_k]_{\gamma} > n) \le \sup_{k} \frac{\mathbb{E} \left[ [x_k]_{\gamma} \right]}{n} \le \frac{C_{12}}{n}
\end{align*}
by Markov's inequality.
For any $\eta > 0$, choosing $n \ge C_{12}/\eta$ yields $\sup\limits_{k} \mathbb{P}(x_k \notin A_n) \le \eta$, which proves tightness.
\end{proof}

Next, we consider two arbitrary subsequences of $\{x_k\}$ and show that they admit a common subsequence whose joint law converges weakly to a random element supported on the diagonal, as stated in the following proposition. This allows us to apply the Gy\"ongy--Krylov lemma (see e.g. \citet*{gyongy1996existence}[Lemma 1.1]) and conclude that $\{x_k\}$ is a Cauchy sequence in probability.
\begin{proposition}\label{lem:limit-sde}
   Let $\{x_{1,k}\}_{k\ge1}$ and $\{x_{2,k}\}_{k\ge1}$ be two arbitrary subsequences of $\{x_k\}$. Then $\{(x_{1,k},x_{2,k})\}_{k\ge1}$ admits a subsequence converging weakly to a random element supported on the diagonal $\{(x,y): (x,y)\in C([0,T])^2, x=y\}$.
\end{proposition}
\begin{proof}
For $a\in\{1,2\}$, let $\{y_{a,k}\}_{k\ge1}$ and $\{\hat{\sigma}_{a,k}\}_{k\ge1}$ denote the subsequences of $\{y_k\}_{k\ge1}$ and $\{\hat{\sigma}_k\}_{k\ge1}$ corresponding to $\{x_{a,k}\}_{k\ge1}$. Combining the tightness of $\{x_k\}$ with the tightness of $W$ and the almost sure convergence of $\{y_k\}$, we conclude that the joint sequence $\{(x_{1,k},x_{2,k},W,y_{1,k},y_{2,k})\}$ is tight in the product space; see \citet*{parthasarathy2005probability}[Theorem 3.2] for the tightness of $W$. By Prokhorov's theorem, it admits a weakly convergent subsequence, which we continue to index by $k$.
By Skorokhod's representation theorem, there exist a probability space $(\tilde{\Omega},\tilde{\mathcal{F}},\tilde{\mathbb{P}})$, random elements $(\tilde{x}_{1,k},\tilde{x}_{2,k},\tilde{W}_k,\tilde{y}_{1,k},\tilde{y}_{2,k})$ having the same distribution as $(x_{1,k},x_{2,k},W,y_{1,k},y_{2,k})$, and a limiting random element $(\tilde{x}_{1,*},\tilde{x}_{2,*},\tilde{W},\tilde{y}_{1,*},\tilde{y}_{2,*})$ having the weak limit law, such that
\begin{align}\label{eq:converge:weak}
(\tilde{x}_{1,k}, \tilde{x}_{2,k}, \tilde{W}_k, \tilde{y}_{1,k}, \tilde{y}_{2,k})\to (\tilde{x}_{1,*}, \tilde{x}_{2,*}, \tilde{W}, \tilde{y}_{1,*}, \tilde{y}_{2,*}) \quad \text{almost surely}.
\end{align}

Note that $(y_{1,k},y_{2,k})$ converges almost surely, and hence weakly, to $(y_*,y_*)$. Moreover, $(\tilde{y}_{1,k}, \tilde{y}_{2,k})$ converges almost surely, and thus weakly, to $(\tilde{y}_{1,*}, \tilde{y}_{2,*})$. Because $(\tilde{y}_{1,k}, \tilde{y}_{2,k})$ has the same distribution as $(y_{1,k},y_{2,k})$ for each $k$, it follows that the limits $(\tilde{y}_{1,*}, \tilde{y}_{2,*})$ and $(y_*,y_*)$ also have the same distribution. Consequently, $\tilde{y}_{1,*} = \tilde{y}_{2,*}$ almost surely, which we denote by $\tilde{y}_*$. Because $\tilde x_{a,k}(0)=\tilde y_{a,k}$ for $a\in\{1,2\}$, the almost sure convergence in \eqref{eq:converge:weak} also gives
\begin{align*}
\tilde x_{1,*}(0)=\tilde y_{1,*}=\tilde y_*=\tilde y_{2,*}=\tilde x_{2,*}(0),
\qquad \tilde{\mathbb P}\text{-a.s.}
\end{align*}

By the well-posedness of \eqref{eq:x:epsilon} and the Yamada-Watanabe theorem, there exists a Borel measurable map $\Phi_{a,k}$ such that $x_{a,k} = \Phi_{a,k}(y_{a,k}, W)$ $\mathbb{P}$-a.s. for $a\in\{1,2\}$. Moreover, both $y_{1,k}$ and $y_{2,k}$ are measurable functions of $\xi$, thus $W$ is independent of the pair $(y_{1,k},y_{2,k})$; equality of the joint laws therefore shows that $\tilde W_k$ is independent of $(\tilde y_{1,k},\tilde y_{2,k})$. Let
\begin{align*}
\tilde{\mathcal G}^{k}_t
:=\sigma\big(\tilde y_{1,k},\tilde y_{2,k},\tilde W_k(u):0\le u\le t\big),
\end{align*}
and let $\tilde{\mathbb F}^{k}=\{\tilde{\mathcal F}^{k}_t\}_{0\le t\le T}$ be its usual augmentation. Then $\tilde W_k$ is a standard $d$-dimensional Brownian motion with respect to $\tilde{\mathbb F}^{k}$, both $\tilde x_{1,k}$ and $\tilde x_{2,k}$ are $\tilde{\mathbb F}^{k}$-adapted, and
\begin{align*}
\tilde x_{a,k}(t)=\tilde y_{a,k}
+\int_0^t\hat\sigma_{a,k}(u,\tilde x_{a,k}(u))^\top\md\tilde W_k(u),
\qquad a\in\{1,2\}.
\end{align*}
Next, let
\begin{align*}
\tilde{\mathcal G}_t
:=\sigma\big(\tilde x_{1,*}(u),\tilde x_{2,*}(u),\tilde W(u):0\le u\le t\big),
\end{align*}
and let $\tilde{\mathbb F}=\{\tilde{\mathcal F}_t\}_{0\le t\le T}$ be its usual augmentation. 
For any $0 \le s < t \le T$ and any bounded continuous functional $\tilde\Phi:C([0,s])^2\times C([0,s];\mathbb R^d)\to\mathbb R$, define 
$H_k := \tilde \Phi ( \tilde{x}_{1,k}|_{[0,s]}, \tilde{x}_{2,k}|_{[0,s]}, \tilde{W}_k|_{[0,s]} )$ and 
$H_* := \tilde \Phi ( \tilde{x}_{1,*}|_{[0,s]}, \tilde{x}_{2,*}|_{[0,s]}, \tilde{W}|_{[0,s]} )$. 
Because $H_k$ is $\tilde{\mathcal{F}}^k_s$-measurable and $\tilde{W}_k$ is a Brownian motion with respect to $\tilde{\mathcal{F}}^k_t$, we have for $i,j\in\{1,2,\cdots,d\}$
\begin{align*}
&\tilde{\mathbb{E}} [ ( \tilde{W}_k^i(t) - \tilde{W}_k^i(s) ) H_k ] = 0, \\
&\tilde{\mathbb{E}} [ ( \tilde{W}_k^i(t)\tilde{W}_k^j(t) - \tilde{W}_k^i(s)\tilde{W}_k^j(s) - \delta_{ij}(t-s) ) H_k ] = 0.
\end{align*}
By the Skorokhod representation theorem, $H_k \xrightarrow{\tilde{\mathbb{P}}\text{-a.s.}} H_*$. 
Because $\tilde{W}_k$ has the same distribution as a standard Brownian motion, we have
\begin{align*}
\sup_{k \ge 1} \tilde{\mathbb{E}} [ | \tilde{W}_k^i(t) - \tilde{W}_k^i(s) |^4 ] < \infty\quad \text{and}\quad \sup_k\tilde{\mathbb E}\Big[\big|\tilde W_k^i(t)\tilde W_k^j(t)-\tilde W_k^i(s)\tilde W_k^j(s)
-\delta_{ij}(t-s)\big|^2\Big]<\infty.
\end{align*}
The uniform fourth-moment bound and the boundedness of $H_k$ ensure the uniform integrability of 
$\{(\tilde{W}_k^i(t) - \tilde{W}_k^i(s)) H_k\}$ and 
$\{(\tilde{W}_k^i(t)\tilde{W}_k^j(t) - \tilde{W}_k^i(s)\tilde{W}_k^j(s) - \delta_{ij}(t-s)) H_k\}$. 
Therefore, sending $k \to \infty$, we obtain
\begin{align}
\left\{
\begin{aligned}
&\tilde{\mathbb{E}} [ ( \tilde{W}^i(t) - \tilde{W}^i(s) ) H_* ] = 0, \label{eq:lim1} \\
&\tilde{\mathbb{E}} [ ( \tilde{W}^i(t)\tilde{W}^j(t) - \tilde{W}^i(s)\tilde{W}^j(s) - \delta_{ij}(t-s) ) H_* ] = 0. 
\end{aligned}
\right.
\end{align}
Because $\tilde \Phi$ is arbitrary, by the monotone class theorem, \eqref{eq:lim1} are equivalent to
\begin{align*}
&\tilde{\mathbb{E}} [ \tilde{W}^i(t) - \tilde{W}^i(s) \mid \tilde{\mathcal{F}}_s ] = 0, \\
&\tilde{\mathbb{E}} [ \tilde{W}^i(t)\tilde{W}^j(t) - \tilde{W}^i(s)\tilde{W}^j(s) - \delta_{ij}(t-s) \mid \tilde{\mathcal{F}}_s ] = 0.
\end{align*}
By L\'evy's characterization theorem, $\tilde{W}$ is a standard $d$-dimensional Brownian motion with respect to $\tilde{\mathbb{F}}$. 
Thus, the stochastic integrals 
$\int_0^t \hat{\sigma}_*(u, \tilde{x}_{1,*}(u))^\top \,\mathrm{d}\tilde{W}(u)$ and 
$\int_0^t \hat{\sigma}_*(u, \tilde{x}_{2,*}(u)) ^\top\,\mathrm{d}\tilde{W}(u)$ are well-defined.

Let $\tilde {\mathcal{N}}_* =  \{(t, y) : y = b_*(t)\}$.  By Lemma \ref{zero-measure} and Fubini's theorem, there exists a Lebesgue null set $\mathcal{T} \subset [0,T]$ such that for any $t \in [0,T] \setminus \mathcal{T}$, the set $\{\tilde{\omega} \in \tilde{\Omega} : \tilde{x}_{1,*}(t, \tilde{\omega}) = b_*(t)\}$ has zero $\tilde{\mathbb{P}}$ measure. As in the proof of Proposition \ref{prop:m_*}, for each $s \in [0,T] \setminus \mathcal{T}$, we have
\begin{align}\label{eq:converge:tilde:x_k}
\lim_{k\to\infty} \hat{\sigma}_{1,k}(s, \tilde{x}_{1,k}(s)) = \hat{\sigma}_*(s, \tilde{x}_{1,*}(s))\quad \text{a.s.}.
\end{align}
Hence, by the dominated convergence theorem, for every $0\le s<t\le T$, $\tilde{\mathbb P}$-a.s.,
\begin{align*}
\int_s^t|\hat\sigma_{1,k}(u,\tilde x_{1,k}(u))|^2\md u
&\longrightarrow\int_s^t|\hat\sigma_*(u,\tilde x_{1,*}(u))|^2\md u,\\
 \int_s^t\hat\sigma_{1,k}^p(u,\tilde x_{1,k}(u))\md u
&\longrightarrow\int_s^t\hat\sigma_*^p(u,\tilde x_{1,*}(u))\md u,
\qquad p=1,\ldots,d.
\end{align*}

Define $\tilde M_{1,k}(t):=\tilde x_{1,k}(t)-\tilde y_{1,k}$ and $\tilde M_{1,*}(t):=\tilde x_{1,*}(t)-\tilde y_*$.  Retain the joint test variables $H_k,H_*$ introduced above. Because $H_k$ is $\tilde{\mathcal F}^k_s$-measurable and $\tilde M_{1,k}=\int\hat\sigma_{1,k}^\top\md\tilde W_k$, for $p\in\{1,\ldots,d\}$,
\begin{align}\label{eq:tilde:k:M}
\left\{
\begin{aligned}
&\tilde{\mathbb{E}}\left[\left(\tilde{M}_{1,k}(t) - \tilde{M}_{1,k}(s)\right)H_k\right] = 0, \\
&\tilde{\mathbb{E}}\left[\left(\tilde{M}_{1,k}(t)^2 - \tilde{M}_{1,k}(s)^2 - \int_s^t |\hat{\sigma}_{1,k}(u, \tilde{x}_{1,k}(u))|^2 \md u\right) H_k\right] = 0, \\
&\tilde{\mathbb{E}}\left[\left(\tilde{M}_{1,k}(t)\tilde{W}_k^p(t) - \tilde{M}_{1,k}(s)\tilde{W}_k^p(s) - \int_s^t \hat{\sigma}_{1,k}^p(u, \tilde{x}_{1,k}(u)) \md u\right) H_k\right] = 0.
\end{aligned}
\right.
\end{align}
 We have $\tilde{M}_{1,k} \to \tilde{M}_{1,*}$ uniformly $\tilde{\mathbb{P}}$-a.s.  and $H_k \to H_*$ $\tilde{\mathbb{P}}$-a.s. given that $\tilde{x}_{1,k} \to \tilde{x}_{1,*}$, $\tilde{W}_k \to \tilde{W}$ and $\tilde{y}_{1,k} \to \tilde{y}_*$ uniformly $\tilde{\mathbb{P}}$-a.s. 
In addition, the uniform fourth-moment bound on $\{\tilde{x}_{1,k}\}$ established in \eqref{ana-4-moment} ensures the uniform integrability of $\{\tilde{M}_{1,k}(t)\}$, $\{\tilde{M}_{1,k}(t)^2\}$ and $\{\tilde{M}_{1,k}(t)\tilde{W}_k^p(t)\}$. Therefore, sending $k \to \infty$ in \eqref{eq:tilde:k:M}, we obtain
\begin{align*}
&\tilde{\mathbb{E}}\left[\left(\tilde{M}_{1,*}(t) - \tilde{M}_{1,*}(s)\right) H_*\right] = 0, \\
&\tilde{\mathbb{E}}\left[\left(\tilde{M}_{1,*}(t)^2 - \tilde{M}_{1,*}(s)^2 - \int_s^t |\hat{\sigma}_*(u, \tilde{x}_{1,*}(u))|^2 \md u\right) H_*\right] = 0, \\
&\tilde{\mathbb{E}}\left[\left(\tilde{M}_{1,*}(t)\tilde{W}^p(t) - \tilde{M}_{1,*}(s)\tilde{W}^p(s) - \int_s^t \hat{\sigma}_*^p(u, \tilde{x}_{1,*}(u)) \md u\right) H_*\right] = 0.
\end{align*}
Thus, $\{\tilde{M}_{1,*}(t)\}$, $\{\tilde{M}_{1,*}(t)^2 - \int_0^t |\hat{\sigma}_*(u, \tilde{x}_{1,*}(u))|^2 \md u\}$ and $\{\tilde{M}_{1,*}(t)\tilde{W}^p(t) - \int_0^t \hat{\sigma}_*^p(u, \tilde{x}_{1,*}(u)) \md u\}$ are martingales with respect to the filtration $\tilde{\mathbb{F}}$. Hence,
\begin{align*}
\langle \tilde{M}_{1,*} \rangle_t = \int_0^t |\hat{\sigma}_*(u, \tilde{x}_{1,*}(u))|^2 \md u, \quad \langle \tilde{M}_{1,*}, \tilde{W}^p \rangle_t = \int_0^t \hat{\sigma}_*^p(u, \tilde{x}_{1,*}(u)) \md u.
\end{align*}
Define $V_1(t) = \tilde{M}_{1,*}(t) - \int_0^t \hat{\sigma}_*(u, \tilde{x}_{1,*}(u))^\top \md \tilde{W}(u)$. Its quadratic variation is
\begin{align*}
\langle V_1 \rangle_t &= \langle \tilde{M}_{1,*} \rangle_t - 2 \sum_{p=1}^d\int_0^t \hat{\sigma}_*^p(u, \tilde{x}_{1,*}(u)) \md \langle \tilde{M}_{1,*}, \tilde{W}^p \rangle_u + \int_0^t |\hat{\sigma}_*(u, \tilde{x}_{1,*}(u))|^2 \md u = 0.
\end{align*}
Because $\langle V_1 \rangle_t \equiv 0$ and $V_1(0) = 0$, it follows that $V_1(t) = 0$ $\tilde{\mathbb{P}}$-a.s., which implies
\begin{align*}
\tilde{x}_{1,*}(t) = \tilde{y}_* + \int_{0}^{t} \hat{\sigma}_*(u, \tilde{x}_{1,*}(u)) ^\top\md \tilde{W}(u).
\end{align*}
Similarly,
\begin{align*}
\tilde{x}_{2,*}(t) = \tilde{y}_* + \int_{0}^{t} \hat{\sigma}_*(u, \tilde{x}_{2,*}(u))^\top \md \tilde{W}(u).
\end{align*}

Let $\tilde {{\mathcal{N}}}_1 \subset [0,T] \times \mathbb{R}$ be the set where the strong solution $u_*$ fails to satisfy the PDE \eqref{eq:pde} corresponding to $m_*$. The set $\{(s,\omega) \in [0,T]\times \tilde{\Omega} : (s, \tilde{x}_{i,*}(s, \tilde{\omega})) \in \tilde {\mathcal{N}}_1\}$ has zero $\md t \otimes \md \tilde {\mathbb{P}}$ measure for $i=1,2$ by Lemma \ref{zero-measure}.
Applying the generalized It\^{o} formula to the process $\{\tilde{Y}_i(t) := u_*(t, \tilde{x}_{i,*}(t))\}$ implies both $\tilde{Y}_i$ solve \eqref{eq:sde:Y}. The  pathwise uniqueness of $Y$ established in Proposition \ref{FBSDE-existence}, together with the bijectivity of $u_*(t, \cdot)$, yields $\tilde{x}_{1,*} = \tilde{x}_{2,*}$ a.s. The conclusion then follows.
\end{proof}

\noindent
\textbf{Proof of Theorem \ref{thm:mean-field}:}
For every pair of subsequences of $\{x_k\}$, Proposition \ref{lem:limit-sde} provides a further common subsequence whose joint law converges weakly to a probability measure supported on the diagonal.  $\{x_k\}$ is a Cauchy sequence in probability according to the Gy\"ongy-Krylov lemma (see, e.g., \citet*{gyongy1996existence}[Lemma 1.1]), and thus converges uniformly in probability to some limit $x$ on $(\Omega, \mathcal{F}, \mathbb{P})$. It follows that there exists a subsequence of $\{x_k\}$ converging almost surely to $x$. An argument analogous to the proof of Proposition \ref{lem:limit-sde} shows that $x$ satisfies the same  SDE \eqref{eq:sde:x} as $x_*$, and that its strong solution is unique, yielding $x = x_*$. As established via Proposition \ref{FBSDE-existence} and Proposition \ref{atomless}, this ensures that the corresponding limit process $(X_*, P_*, Q_*)$ solves the associated FBSDE \eqref{eq:fbsde_sufficiency} and $X_*$ is atomless. Consequently, Theorem \ref{sufficiency} guarantees that the intra-personal equilibrium condition (a) is satisfied.

Next, to verify condition (b), we note that similar to the proof of Lemma \ref{lem:p-convergence}, we have
\begin{align*}
\lim_{k\to\infty} \mathbb{P}\left(\sup_{t\in[0,T]} |X_k(t) - X_*(t)| > \eta\right) = 0.
\end{align*}
As established in \eqref{4-uniform}, $\{X_k\}$ has a uniform fourth-moment bound independent of $k$. Uniform integrability and Vitali's convergence theorem upgrade the convergence to $L^1$ such that
\begin{align*}
\lim_{k\to\infty} \mathbb{E}\left[\sup_{t\in[0,T]} |X_k(t) - X_*(t)|\right] = 0.
\end{align*}
Taking the limit in the fixed-point equation $m_k(t) = \mathbb{E}[X_k(t)]$ yields
\begin{align*}
m_*(t) = \mathbb{E}[X_*(t)],
\end{align*}
which confirms the mean-field consistency condition (b) in Definition \ref{def:MFE}. Combining with the fact that $\bar{\pi}_*$ in Theorem \ref{thm:mean-field} satisfies the intra-personal equilibrium condition (a) by Theorem \ref{sufficiency}, we conclude that the pair $(\bar{\pi}_*, m_*)$ constitutes an open-loop mean-field equilibrium.

\begin{appendices}
\section{Proofs of Auxiliary Results}\label{sec:App}
\subsection{Proof of Theorem \ref{necessary}}\label{Proof of Necessary}

First, classical results (see, e.g., \citet*{zhang2017backward}[Theorem 4.3.1]) guarantee that BSDE \eqref{eq:bsde} admits a unique solution $(P,Q)\in L_{\mathbb{F}}^{2}(\Omega; C([0, T];\mathbb{R}))\times L_{\mathbb{F}}^{2}(0, T; \mathbb{R}^d)$.
\vskip 1pt \noindent
For any $t \in [0, T)$, $\eta \in L_{\mathcal{F}_{t}}^{2}(\Omega; \mathbb{R}^d)$, and $\epsilon \in (0, T - t)$, we denote by $X^{\epsilon,t}$ the state process $X^{\pi^{t,\epsilon,\eta}}$ under the control $\pi^{t,\epsilon,\eta}$, and denote
\begin{align*}
X_1^{\epsilon,t} := X^{\epsilon,t}(T) - X(T) =\int_t^{t+\epsilon} P_0(s)\eta^{\top}\theta(s) \md s + \int_t^{t+\epsilon} P_0(s)\eta^{\top}\sigma(s) \md W(s).
\end{align*}
Before proceeding with the proof, we first give a more explicit representation in the next result for the difference between the objective functionals corresponding to the perturbed strategy and the equilibrium strategy.
\begin{lemma}\label{lem:functional_limit}
Define the martingale $H:=\left\{H(s) = \mE_s[X(T)]:s\in[0,T]\right\}$. Let $Z \in L_{\mathbb{F}}^{2}(0, T; \mathbb{R}^d)$ be the process given by the martingale representation theorem such that $X(T) = H(\cdot) + \int_{\cdot}^T Z^{\top}(s) \md W(s)$. It holds that
\begin{align}\label{eq:J:dif}
J(t, \pi^{t,\epsilon,\eta}) - J(t, \bar{\pi}) =&-\gamma(t, X(t)) \eta^{\top} \int_t^{t+\epsilon}  P_0(s)\theta(s) \md s +\frac{1}{2} \eta^{\top} \left( \int_t^{t+\epsilon} P_0(s)^2 \sigma(s)\sigma^{\top}(s) \md s \right) \eta\nonumber\\
&+ \eta^{\top} \mE_t\left[ \int_t^{t+\epsilon} P_0(s)\sigma(s)Z(s) \md s \right].
\end{align}
\end{lemma}
\begin{proof}
A direct computation yields
\begin{align*}
\Var_t(X^{\epsilon,t}(T)) - \Var_t(X(T)) &= \left( \Var_t(\mE_{t+\epsilon}[X^{\epsilon,t}(T)]) + \mE_t[\Var_{t+\epsilon}(X^{\epsilon,t}(T))] \right) \\
&\quad - \left( \Var_t(\mE_{t+\epsilon}[X(T)]) + \mE_t[\Var_{t+\epsilon}(X(T))] \right)\\
&= \Var_t(H(t+\epsilon) + X_1^{\epsilon,t}) - \Var_t(H(t+\epsilon)) \\
&= \Var_t(X_1^{\epsilon,t}) + 2 \Cov_t(H(t+\epsilon), X_1^{\epsilon,t}),
\end{align*}
where the second equality results from the fact that $X_1^{\epsilon,t}$ is $\mathcal{F}_{t+\epsilon}$-measurable and hence $\Var_{t+\epsilon}(X^{\epsilon,t}(T)) = \Var_{t+\epsilon}(X(T) + X_1^{\epsilon,t}) = \Var_{t+\epsilon}(X(T))$. Consequently,
\begin{align*}
J(t, \pi^{t,\epsilon,\eta}) - J(t, \bar{\pi}) = -\gamma(t, X(t)) \mE_t[X_1^{\epsilon,t}] + \Cov_t(H(t+\epsilon), X_1^{\epsilon,t}) + \frac{1}{2} \Var_t(X_1^{\epsilon,t}).
\end{align*}
Let
\begin{align*}
\alpha_\epsilon := \int_t^{t+\epsilon} P_0(s)\eta^{\top}\theta(s) \md s, \quad \beta_\epsilon := \int_t^{t+\epsilon} P_0(s)\eta^{\top}\sigma(s) \md W(s).
\end{align*}
Then  $\mE_t[X_1^{\epsilon,t}] = \mE_t[\alpha_\epsilon] = \alpha_\epsilon$ and $\Var_t(\alpha_\epsilon) = 0$ because  $\eta \in L_{\mathcal{F}_t}^2(\Omega; \mathbb{R}^d)$ and the deterministic coefficients $P_0(s)$ and $\theta(s)$ are continuous. It thus holds that
\begin{align*}
-\gamma(t, X(t)) \mE_t[X_1^{\epsilon,t}] = -\gamma(t, X(t)) \eta^{\top} \int_t^{t+\epsilon}  P_0(s)\theta(s) \md s,
\end{align*}
and 
\begin{align*}
    \Var_t(X_1^{\epsilon,t}) = \Var_t(\beta_\epsilon) = \mE_t[\beta_\epsilon^2]=\eta^{\top} \left( \int_t^{t+\epsilon} P_0(s)^2 \sigma(s)\sigma^{\top}(s) \md s \right) \eta.
\end{align*}
Similarly,
\begin{align*}
    \Cov_t(H(t+\epsilon), X_1^{\epsilon,t}) &= \mE_t[(H(t+\epsilon) - \mE_t[H(t+\epsilon)])(X_1^{\epsilon,t} - \mE_t[X_1^{\epsilon,t}])]\\
    &=\mE_t[(H(t+\epsilon) - H(t))\beta_\epsilon]= \mE_t\left[ \int_t^{t+\epsilon} Z^{\top}(s) \md W(s) \cdot \beta_\epsilon \right] \\&= \mE_t\left[ \int_t^{t+\epsilon} P_0(s)\eta^{\top}\sigma(s)Z(s) \md s \right].
\end{align*}
As a result, \eqref{eq:J:dif} holds.
\end{proof}
We now divide both sides of \eqref{eq:J:dif} by $\epsilon$ and pass to the limit as $\epsilon \to 0^+$. We have the estimate
\begin{align*}
\left| \frac{1}{\epsilon} \int_t^{t+\epsilon} P_0(s)\sigma(s)Z(s) \md s \right| \le K \left( \sup_{h \in \mathbb{Q} \cap (0, T-t]} \frac{1}{h} \int_t^{t+h} |Z(s)| \md s \right) := K M_t(Z)
\end{align*}
because the coefficients $P_0$ and $\sigma$ are bounded by a constant $K$.
Note that the Lebesgue integral is continuous with respect to $h$, the supremum over the countable set of rationals in $(0, T-t]$ is equivalent to the supremum over all real $h \in (0, T-t]$. Hence, the maximal function $M_t(Z)$ is a well-defined $\mathcal{F}_T$-measurable random variable.
Fubini's theorem implies that almost all sample paths of $Z$ belong to $L^2(0, T; \mathbb{R}^d)$ in light of $Z \in L_{\mathbb{F}}^{2}(0, T; \mathbb{R}^d)$. Applying the Hardy–Littlewood maximal inequality (see, e.g., \citet*[Theorem 9]{hardy1930maximal}) pathwise and then taking expectations, we get that $\mE [ \int_0^T |M_t(Z)|^2 \md t ] \le C \mE [ \int_0^T |Z(t)|^2 \md t ] < \infty$. Specifically, $M_t(Z) \in L_{\mathcal{F}_T}^2(\Omega; \mathbb{R}) \subset L_{\mathcal{F}_T}^1(\Omega; \mathbb{R})$ for a.e. $t\in[0,T]$, and thus the dominating random variable $K M_t(Z)$ is integrable. Consequently, using the dominated convergence theorem yields
\begin{align*}
\lim_{\epsilon \to 0^+} \frac{1}{\epsilon} \eta^{\top} \mE_t\left[ \int_t^{t+\epsilon} P_0(s)\sigma(s)Z(s) \md s \right] &= \eta^{\top} \mE_t\left[ \lim_{\epsilon \to 0^+} \frac{1}{\epsilon} \int_t^{t+\epsilon} P_0(s)\sigma(s)Z(s) \md s \right] \\
&= \eta^{\top} \mE_t[ P_0(t)\sigma(t)Z(t) ] = \eta^{\top} P_0(t)\sigma(t)Z(t) \quad \text{a.s.}.
\end{align*}
Moreover, noting $P_0$, $\theta$, and $\sigma$ are deterministic and continuous, we have that, for any $t \in [0,T)$,
\begin{align*}
    &\lim_{\epsilon\to 0+}\frac{-\gamma(t, X(t)) \eta^{\top} \int_t^{t+\epsilon}  P_0(s)\theta(s) \md s +\frac{1}{2} \eta^{\top} \left( \int_t^{t+\epsilon} P_0(s)^2 \sigma(s)\sigma^{\top}(s) \md s \right) \eta}{\epsilon}\\
    =& -\gamma(t, X(t)) \eta^{\top}P_0(t)\theta(t)+\frac12\eta^\top P_0(t)^2\sigma(t)\sigma^{\top}(t)\eta\quad \text{a.s.}.
\end{align*}
As a result, we obtain that for a.e. $t\in[0,T)$
\begin{align*}
\lim_{\epsilon \to 0^+} \frac{J(t, \pi^{t,\epsilon,\eta}) - J(t, \bar{\pi})}{\epsilon} &= \eta^{\top} \left( -\gamma(t, X(t)) P_0(t)\theta(t) + P_0(t)\sigma(t)Z(t) \right) \\
&\quad\quad\quad+ \frac{1}{2} \eta^{\top} \left( P_0(t)^2 \sigma(t)\sigma^{\top}(t) \right) \eta\quad \text{a.s.}.
\end{align*}
In view that $\bar{\pi}$ is an equilibrium and $\theta(\cdot) = \sigma(\cdot)\lambda(\cdot)$, it holds that for a.e. $t\in[0,T)$,
\begin{align}\label{eq:foc}
Z(t)= \gamma(t, X(t)) \lambda(t)
\quad \text{a.s.}.
\end{align}
Define $\tilde{P}(s) := H(s) - P_0(s)X(s)$ and $\tilde{Q}(s) := Z(s) - P_0(s)\sigma^{\top}(s)\bar{\pi}(s)$.  
Applying It\^o's formula gives
\begin{align*}
\md \tilde{P}(s) 
&= Z^{\top}(s) \md W(s) - P_0(s) [ (r(s)X(s) + \bar{\pi}^{\top}(s)\theta(s)) \md s + \bar{\pi}^{\top}(s)\sigma(s) \md W(s) ] + r(s)P_0(s)X(s) \md s \\
&= -P_0(s)\bar{\pi}^{\top}(s)\theta(s) \md s + \tilde{Q}^{\top}(s) \md W(s)\\
&= -\left( \gamma(s, X(s)) |\lambda(s)|^2 - \lambda^{\top}(s)\tilde{Q}(s) \right) \md s + \tilde{Q}^{\top}(s) \md W(s).
\end{align*}
Moreover, we have  the terminal condition $\tilde{P}(T) = H(T) - P_0(T)X(T) = 0$. By the uniqueness of solution to the BSDE \eqref{eq:bsde}, it holds that $(\tilde{P}, \tilde{Q})=(P,Q)$. Substituting $\tilde{Q}(t)=Q (t) = Z(t) - P_0(t)\sigma^{\top}(t)\bar{\pi}(t)$ into \eqref{eq:foc}, we obtain
\begin{align*}
\sigma(t)\left[Q(t) + P_0(t)\sigma^\top(t)\bar{\pi}(t)\right] - \gamma(t, X(t))\sigma(t)\lambda(t) = 0.
\end{align*}
Solving for $\bar{\pi}(t)$, we  obtain \eqref{eq:barU}, which completes the proof. 

\subsection{Proof of Theorem \ref{sufficiency}}\label{Proof of Sufficiency}
Because $Q \in L_{\mathbb{F}}^{2}(0, T; \mathbb{R}^d)$ and the deterministic coefficients are bounded, it follows that $\bar{\pi} \in L_{\mathbb{F}}^{2}(0, T; \mathbb{R}^d)$, and the corresponding state process satisfies $X^{\bar{\pi}}=X$. Define $H(\cdot) = P(\cdot) + P_0(\cdot)X(\cdot) \in L_{\mathbb{F}}^{2}(\Omega; C([0, T];\mathbb{R}))$. Applying It\^o's formula yields
\begin{align*}
\md H(t) &= \left[ -\gamma(t, X(t)) |\lambda(t)|^2 + \lambda^{\top}(t)Q(t) + P_0(t)\bar{\pi}^{\top}(t)\theta(t) \right] \md t + \left[ Q(t) + P_0(t)\sigma^{\top}(t)\bar{\pi}(t) \right]^{\top} \md W(t)\\
&= \left[ Q(t) + P_0(t)\sigma^{\top}(t)\bar{\pi}(t) \right]^{\top} \md W(t),
\end{align*}
where the second equality is from \eqref{eq:barU} and $\theta(\cdot) = \sigma(\cdot)\lambda(\cdot)$.
Then $H(\cdot)$ is a square-integrable martingale. Using the terminal condition $H(T) = P(T) + P_0(T)X(T) = X(T)$, we deduce that $H(t) = \mathbb{E}_t[X(T)]$ and it admits the representation $X(T) = H(t) + \int_t^T Z^{\top}(s) \md W(s)$, $ t\in[0,T]$, where $Z(\cdot) := Q(\cdot) + P_0(\cdot)\sigma^{\top}(\cdot)\bar{\pi}(\cdot) \in L_{\mathbb{F}}^{2}(0, T; \mathbb{R}^d)$.
We have $\sigma(s)Z(s) = \gamma(s, X(s))\theta(s)$ by \eqref{eq:barU}. 
Substituting it into \eqref{eq:J:dif}, we obtain that, for any $t\in[0,T)$, $\epsilon\in(0,T-t)$, and $\eta \in L_{\mathcal{F}_t}^2(\Omega; \mathbb{R}^d)$,
\begin{align*}
J(t, \pi^{t,\epsilon,\eta}) - J(t, \bar{\pi}) &= \eta^{\top} \mathbb{E}_t\left[ \int_t^{t+\epsilon} P_0(s)\theta(s) \left( \gamma(s, X(s)) - \gamma(t, X(t)) \right) \md s \right] \\
&\quad + \frac{1}{2} \eta^{\top} \left( \int_t^{t+\epsilon} P_0(s)^2 \sigma(s)\sigma^{\top}(s) \md s \right) \eta.
\end{align*}
For any $t\in[0,T)$, define $d_t := |X(t) - m(t)|$ and the $\mathcal{F}_t$-measurable random variable
\begin{align*}
\epsilon_{1,t} &:= \inf \left\{ \delta\in[0,T-t] : \sup_{s \in [t, t+\delta]} |m(s) - m(t)| \ge \frac{1}{4} d_t \right\}.
\end{align*}
Because $X(t)$ is atomless, it holds that $d_t > 0$ a.s. 
Combined with the continuity of $m$, it implies that $\epsilon_{1,t} > 0$ a.s.
For any constant $\epsilon \in( 0,T-t)$, define
$\Omega_{t,\epsilon} := \{ \omega \in \Omega : \epsilon_{1,t}(\omega) \ge \epsilon \}$.
It holds that
$\lim\limits_{\epsilon \downarrow 0} \mathbbm{1}_{\Omega_{t,\epsilon}} = 1$ a.s. in view of $\epsilon_{1,t} > 0$ a.s. Let
\begin{align*}
\Delta\gamma(s,t) &:= \gamma(s, X(s)) - \gamma(t, X(t)),\quad s\in[t,T).
\end{align*}
 On $\Omega_{t,\epsilon}$, by the definition of $\epsilon_{1,t}$, for any $s \in (t, t+\epsilon]$, it holds that $|m(s) - m(t)| \le \frac{1}{4} d_t$. 
Moreover, if $\Delta\gamma(s,t) \neq 0$, the terms $(X(t) - m(t))$ and $(m(s) - X(s))$ are either both non-negative or both non-positive. Consequently, on $\Omega_{t,\epsilon}\cap \left\{\Delta\gamma(s,t) \neq 0\right\}$, it holds that
\begin{align*}
|X(s) - X(t)|
&\ge |(X(t) - m(t)) + (m(s) - X(s))| - |m(s) - m(t)| \\
&= |X(t) - m(t)| + |m(s) - X(s)| - |m(s) - m(t)| \\
&\ge d_t + 0 - \frac{1}{4}d_t = \frac{3}{4}d_t.
\end{align*}
Define $K_\gamma = |\gamma_1 - \gamma_2|$. Using $|\Delta\gamma(s,t)| \le K_\gamma \mathbbm{1}_{\{\Delta\gamma(s,t) \neq 0\}}$, we further get
\begin{align}
&\mathbbm{1}_{\Omega_{t,\epsilon}} \mathbb{E}_t \left[ |\Delta\gamma(s,t)| \right]
\le K_\gamma \mathbb{E}_t \left[ \mathbbm{1}_{\Omega_{t,\epsilon}} \mathbbm{1}_{\{\Delta\gamma(s,t) \neq 0\}} \right] \le K_\gamma \mathbb{E}_t \left[ \mathbbm{1}_{\Omega_{t,\epsilon}} \mathbbm{1}_{\left\{|X(s) - X(t)| \ge \frac{3}{4}d_t\right\}} \right]\nonumber \\
&= K_\gamma \mathbbm{1}_{\Omega_{t,\epsilon}} \mathbb{P}_t \left( |X(s) - X(t)| \ge \frac{3}{4}d_t \right)\leq K_\gamma \mathbbm{1}_{\Omega_{t,\epsilon}} \frac{16}{9d_t^2} \mathbb{E}_t \left[ |X(s) - X(t)|^2 \right],\label{eq:delta:gamma}
\end{align}
where the last inequality used the Markov inequality and the fact that $d_t>0$ a.s. Using the Cauchy-Schwarz inequality and It\^o's isometry, we obtain that, for any $s\in[t,T)$,
\begin{align*}
\mathbb{E}_t \left[ |X(s) - X(t)|^2 \right] &= \mathbb{E}_t \left[ \left| \int_t^s \left( r(v)X(v) + \bar{\pi}^{\top}(v)\theta(v) \right) \md v + \int_t^s \bar{\pi}^{\top}(v)\sigma(v) \md W(v) \right|^2 \right] \\
&\le \mathbb{E}_t \left[ \int_t^s \left( 2T \left| r(v)X(v) + \bar{\pi}^{\top}(v)\theta(v) \right|^2 + 2 \left| \bar{\pi}^{\top}(v)\sigma(v) \right|^2 \right) \md v \right] =: \Lambda_t(s).
\end{align*}
Let $Y := \int_t^T \left( 2T \left| r(v)X(v) + \bar{\pi}^{\top}(v)\theta(v) \right|^2 + 2 \left| \bar{\pi}^{\top}(v)\sigma(v) \right|^2 \right) \md v$. Then $\mathbb{E}[Y] < \infty$. Using the dominated convergence theorem yields $\lim\limits_{s \downarrow t} \Lambda_t(s) = 0$ a.s.
Substituting  $\mathbb{E}_t \left[ |X(s) - X(t)|^2 \right] \le \Lambda_t(s)$ into \eqref{eq:delta:gamma} yields
\begin{align*}
\mathbbm{1}_{\Omega_{t,\epsilon}} \mathbb{E}_t \left[ |\Delta\gamma(s,t)| \right] &\le \mathbbm{1}_{\Omega_{t,\epsilon}} \frac{16 K_\gamma}{9d_t^2} \Lambda_t(s).
\end{align*}
Integrating this inequality over the interval $[t, t+\epsilon]$ and dividing by $\epsilon$, we have
\begin{align*}
\mathbbm{1}_{\Omega_{t,\epsilon}} \frac{1}{\epsilon} \int_t^{t+\epsilon} \mathbb{E}_t \left[ |\Delta\gamma(s,t)| \right] \md s &\le \mathbbm{1}_{\Omega_{t,\epsilon}} \frac{16 K_\gamma}{9d_t^2} \frac{1}{\epsilon} \int_t^{t+\epsilon} \Lambda_t(s) \md s\leq \mathbbm{1}_{\Omega_{t,\epsilon}} \frac{16 K_\gamma}{9d_t^2} \Lambda_t(t+\epsilon).
\end{align*}
Let $K_{\theta} := \max\limits_{s \in [0,T]} |P_0(s)\theta(s)|$, 
we further have
\begin{align*}
\left| \frac{1}{\epsilon} \mathbb{E}_t \left[ \int_t^{t+\epsilon} P_0(s)\theta(s) \Delta\gamma(s,t) \md s \right] \right| &= \mathbbm{1}_{\Omega_{t,\epsilon}} \left| \frac{1}{\epsilon} \mathbb{E}_t \left[ \int_t^{t+\epsilon} P_0(s)\theta(s) \Delta\gamma(s,t) \md s \right] \right| \\
&\quad + \mathbbm{1}_{\Omega_{t,\epsilon}^c} \left| \frac{1}{\epsilon} \mathbb{E}_t \left[ \int_t^{t+\epsilon} P_0(s)\theta(s) \Delta\gamma(s,t) \md s \right] \right|\\
&\leq \mathbbm{1}_{\Omega_{t,\epsilon}} \frac{16 K_\gamma K_{\theta}}{9d_t^2} \Lambda_t(t+\epsilon)+K_{\theta} K_\gamma \mathbbm{1}_{\Omega_{t,\epsilon}^c}.
\end{align*}
Therefore
\begin{align*}
\limsup_{\epsilon \downarrow 0} \left| \frac{1}{\epsilon} \mathbb{E}_t \left[ \int_t^{t+\epsilon} P_0(s)\theta(s) \Delta\gamma(s,t) \md s \right] \right| &\le \limsup_{\epsilon \downarrow 0} \left( \frac{16 K_{\theta} K_\gamma}{9d_t^2} \Lambda_t(t+\epsilon) + K_{\theta} K_\gamma \mathbbm{1}_{\Omega_{t,\epsilon}^c} \right) = 0.
\end{align*}
Moreover, we have
\begin{align*}
\lim_{\epsilon \downarrow 0} \frac{1}{2\epsilon} \eta^{\top} \left( \int_t^{t+\epsilon} P_0(s)^2 \sigma(s)\sigma^{\top}(s) \md s \right) \eta &= \frac{1}{2} P_0(t)^2 \eta^{\top} \sigma(t)\sigma^{\top}(t) \eta
\end{align*}
in view that $P_0$, $\theta$, and $\sigma$ are deterministic and continuous.
Consequently, for any $t \in [0, T)$ and $\eta \in L_{\mathcal{F}_t}^2(\Omega; \mathbb{R}^d)$, it holds a.s. that
\begin{align*}
\liminf_{\epsilon \downarrow 0} \frac{J(t, \pi^{t,\epsilon,\eta}) - J(t, \bar{\pi})}{\epsilon} &= 0 + \frac{1}{2} P_0(t)^2 \eta^{\top}\sigma(t)\sigma^{\top}(t)\eta \ge 0.
\end{align*}
Thus, $\bar{\pi}(\cdot)$ is an equilibrium strategy.

\subsection{Proof of Lemma \ref{lem:acl}} \label{appendix:lma:smooth:method}
We only prove (i), because (ii) can be shown by an  analogous argument.  There exists a sequence $\{u^k\} \subset C^\infty(Q_T)$ such that $u^k \to u$ in $W_{2,\mathrm{loc}}^{1,2}(Q_{T})$ because $u\in W_{2,\mathrm{loc}}^{1,2}(Q_{T})$. 
For any $N \in \mathbb{N}$, let $K_N = [-N,N]$. Then
$$ \iint_{[0,T] \times K_N} \left( |u^k - u|^2 + \left| \frac{\partial u^k}{\partial x} - \frac{\partial u}{\partial x} \right|^2 + \left| \frac{\partial^2 u^k}{\partial x^2} - \frac{\partial^2 u}{\partial x^2} \right|^2 \right) \md x \md t \rightarrow 0 \quad \text{as } k \rightarrow \infty. $$
 Let $E_k(t, N) := \int_{K_N} \left( |u^k - u|^2 + \left| \frac{\partial u^k}{\partial x} - \frac{\partial u}{\partial x} \right|^2 + \left| \frac{\partial^2 u^k}{\partial x^2} - \frac{\partial^2 u}{\partial x^2} \right|^2 \right) \md x$, $t\in[0,T]$. 
Then $E_k(\cdot, N)$ converges to $0$ in $L^1([0,T])$ and there exists a subsequence $\{k_j\}_{j=1}^\infty$ such that $\lim\limits_{j \to \infty} E_{k_j}(t, N) = 0$ for a.e. $t \in [0,T]$. Employing a diagonal argument, we  extract a subsequence (still denoted as $\{u^k\}$) and find a set of full measure $\mathcal{T}_1 \subset [0,T]$ such that for every $t \in \mathcal{T}_1$, $u^k(t,\cdot) \rightarrow u(t,\cdot)$, $\frac{\partial u^k}{\partial x}(t,\cdot) \rightarrow \frac{\partial u}{\partial x}(t,\cdot)$, and $\frac{\partial^2 u^k}{\partial x^2}(t,\cdot) \rightarrow \frac{\partial^2 u}{\partial x^2}(t,\cdot)$ in $L^2(K_N)$ for all $N \ge 1$. 
By the completeness of Sobolev spaces, it follows that for each $t \in \mathcal{T}_1$, $u(t,\cdot) \in W_{2,\mathrm{loc}}^2(\mathbb{R})$ and its weak derivatives coincide with $\frac{\partial u}{\partial x}(t,\cdot)$ and $\frac{\partial^2 u}{\partial x^2}(t,\cdot)$. Consequently, $\frac{\partial u}{\partial x}(t,\cdot) \in W_{2,\mathrm{loc}}^1(\mathbb{R})$, which implies that $\frac{\partial u}{\partial x}(t,\cdot)$ is absolutely continuous on any finite interval (see e.g. \citet*{evans2022partial}[Section 5.10, Problem 4]).

\subsection{Proof of Lemma \ref{lma:v:derivative}}\label{appendix:lma:v:derivative}
For a standard mollifier $\tilde \omega_k(\cdot) := k^2 \tilde\omega(k\cdot)$, where $\tilde\omega$ is smooth, non-negative, compactly supported, and satisfies $\iint_{\mathbb{R}^2} \tilde\omega(t,z) \md z \md t = 1$, let us define $u^k := u * \tilde\omega_k \in C^\infty(Q_T)$. This sequence converges to $u$ in $W^{1,2}_2(K)$ for any compact $K \subset Q_T$, and uniformly in $C(K)$. By extracting a subsequence (still denoted by $\{u^k\}$), we may assume that $\partial_x u^k$, $\partial_{xx} u^k$, and $\partial_t u^k$ converge almost everywhere to the corresponding weak derivatives of $u$. The convolution preserves bounds because $\tilde\omega_k \ge 0$ and has unit mass, thus $\delta_1 \le \frac{\partial u^k}{\partial x} \le M_1$.

There exists a unique $v^k \in C^\infty(Q_T)$ such that $u^k(t, v^k(t,y)) = y$ by the implicit function theorem. Differentiating this equality yields
\begin{align*}
\frac{\partial v^k}{\partial y} &= \left(\frac{\partial u^k}{\partial x}(t, v^k)\right)^{-1}, \quad 
\frac{\partial^2 v^k}{\partial y^2} = -\frac{\partial^2 u^k}{\partial x^2}(t, v^k)\left(\frac{\partial u^k}{\partial x}(t, v^k)\right)^{-3}, \\
\frac{\partial v^k}{\partial t} &= -\frac{\partial u^k}{\partial t}(t, v^k)\left(\frac{\partial u^k}{\partial x}(t, v^k)\right)^{-1}.
\end{align*}
For any compact set $K \subset Q_T$, let $K' \subset Q_T$ be a compact neighborhood containing the points $(t, v(t,y))$ and $(t, v^k(t,y))$ for all $(t,y) \in K$ and all sufficiently large $k$. Then
\begin{align}\label{delta_0}
&\left|u(t, v(t,y)) - u^k(t, v(t,y))\right|=\left|u^k(t, v^k(t,y)) - u^k(t, v(t,y))\right|\nonumber \\
=& \left| \int_{v(t,y)}^{v^k(t,y)} \frac{\partial u^k}{\partial x}(t, \xi) \md \xi \right|\ge \delta_1 \left|v^k(t,y) - v(t,y)\right|\quad \text{for all }(t,y)\in K. 
\end{align}
As a result, $v^k$ converges uniformly to $v$ on $K$ because $u^k$ converges uniformly to $u$ on $K'$.

To prove that $v \in W_{2,\mathrm{loc}}^{1,2}(Q_T)$ with weak derivatives defined in \eqref{eq:derivative:1}–\eqref{eq:derivative:2}, it suffices to show that the derivatives of the smooth approximating sequence $v^k$ converge in $L^2(K)$ to the respective expressions. As the arguments for all three derivatives are similar, we only present the proof for $\frac{\partial^2 v^k}{\partial y^2}$. Define $g(t,y) := -\frac{\partial^2 u}{\partial x^2}(t, v)\left(\frac{\partial u}{\partial x}(t, v)\right)^{-3}$. Then
\begin{align}\label{eq:estimate:split}
\left\| \frac{\partial^2 v^k}{\partial y^2} - g \right\|_{L^2(K)} &\le \left\| \frac{\frac{\partial^2 u^k}{\partial x^2}(t, v^k)}{\left(\frac{\partial u^k}{\partial x}(t, v^k)\right)^3} - \frac{\frac{\partial^2 u}{\partial x^2}(t, v^k)}{\left(\frac{\partial u}{\partial x}(t, v^k)\right)^3} \right\|_{L^2(K)} \nonumber \\ & + \left\| \frac{\frac{\partial^2 u}{\partial x^2}(t, v^k)}{\left(\frac{\partial u}{\partial x}(t, v^k)\right)^3} - \frac{\frac{\partial^2 u}{\partial x^2}(t, v)}{\left(\frac{\partial u}{\partial x}(t, v)\right)^3} \right\|_{L^2(K)}.
\end{align}
Using $y = u^k(t,x)$, $\md y = \frac{\partial u^k}{\partial x} \md x \le M_1 \md x$ and the change of variables $x = v^k(t,y)$, we have 
\begin{align*}
\left\| \frac{\frac{\partial^2 u^k}{\partial x^2}(t, v^k)}{\left(\frac{\partial u^k}{\partial x}(t, v^k)\right)^3} - \frac{\frac{\partial^2 u}{\partial x^2}(t, v^k)}{\left(\frac{\partial u}{\partial x}(t, v^k)\right)^3} \right\|_{L^2(K)}^2 &\le M_1 \iint_{K'} \left| \frac{\frac{\partial^2 u^k}{\partial x^2}}{\left(\frac{\partial u^k}{\partial x}\right)^3} - \frac{\frac{\partial^2 u}{\partial x^2}}{\left(\frac{\partial u}{\partial x}\right)^3} \right|^2 \md x \md t.
\end{align*}
We decompose the integrand and derive that
\begin{align*}
\left| \frac{\frac{\partial^2 u^k}{\partial x^2}}{\left(\frac{\partial u^k}{\partial x}\right)^3} - \frac{\frac{\partial^2 u}{\partial x^2}}{\left(\frac{\partial u}{\partial x}\right)^3} \right|^2 &\le 2 \left| \frac{\frac{\partial^2 u^k}{\partial x^2} - \frac{\partial^2 u}{\partial x^2}}{\left(\frac{\partial u^k}{\partial x}\right)^3} \right|^2 + 2 \left| \frac{\partial^2 u}{\partial x^2} \left( \frac{1}{\left(\frac{\partial u^k}{\partial x}\right)^3} - \frac{1}{\left(\frac{\partial u}{\partial x}\right)^3} \right) \right|^2 \\
&\le \frac{2}{\delta_1^6} \left| \frac{\partial^2 u^k}{\partial x^2} - \frac{\partial^2 u}{\partial x^2} \right|^2 + 2 \left| \frac{\partial^2 u}{\partial x^2} \right|^2 \left| \frac{1}{\left(\frac{\partial u^k}{\partial x}\right)^3} - \frac{1}{\left(\frac{\partial u}{\partial x}\right)^3} \right|^2.
\end{align*}
The integral of the first term converges to zero due to the $L^2$ convergence of $\frac{\partial^2 u^k}{\partial x^2}$. For the second term, the integrand converges to zero almost everywhere and is dominated by $\frac{8M_1^6}{\delta_1^{12}} \left| \frac{\partial^2 u}{\partial x^2} \right|^2 \in L^1(K')$. Hence, its integral also tends to zero by the dominated convergence theorem. This verifies that the first term on the right hand side of \eqref{eq:estimate:split} converges to zero.

Let $F(t,x) := \frac{\partial^2 u}{\partial x^2}(t,x) \left(\frac{\partial u}{\partial x}(t,x)\right)^{-3} \in L^2(K')$. For any $\rho>0$, there exists $F_{\rho} \in C_c(Q_T)$ such that $\|F - F_{\rho}\|_{L^2(K')} < \rho$. We decompose the second term on the right-hand side of \eqref{eq:estimate:split} as
\begin{align*}
\left\| F(t, v^k) - F(t, v) \right\|_{L^2(K)} \le \left\| F(t, v^k) - F_{\rho}(t, v^k) \right\|_{L^2(K)} &+ \left\| F_{\rho}(t, v^k) - F_{\rho}(t, v) \right\|_{L^2(K)} \\&+ \left\| F_{\rho}(t, v) - F(t, v) \right\|_{L^2(K)}.
\end{align*}
The first and third norms are bounded by $\sqrt{M_1}\rho$  by applying the change of variables $x = v^k(t,y)$ and $x = v(t,y)$, respectively. The second term converges to zero due to the global uniform continuity of $F_{\rho}$ and the uniform convergence of $v^k$. Thus, $\frac{\partial^2 v^k}{\partial y^2}$ converges to $g$ in $L^2(K)$.
\section{Auxiliary Result of Equation \eqref{eq:regularized}}\label{sec:AppB}
\begin{lemma}\label{lem:uniqueness_strong}
There exists a unique strong solution $u^\varepsilon \in W^{1,2}_{2,\mathrm{loc}}(Q_T)$ to the regularized PDE \eqref{eq:regularized} satisfying $|u^\varepsilon(t,x)-x|\leq M(T-t)$ and $\partial_x u^\varepsilon \in L^\infty(Q_T)$.
\end{lemma}
\begin{proof}
Based on Theorem \ref{epsilon-existence}, the regularized PDE \eqref{eq:regularized} admits a classical solution $u^{\varepsilon}$ satisfying $\partial_x u^{\varepsilon}, \partial^2_{xx} u^{\varepsilon} \in L^\infty(Q_T)$ and $|u^{\varepsilon}(t,x) - x| \le M(T-t)$. The existence of the strong solution is immediately guaranteed because any classical solution is inherently a strong solution in $W^{1,2}_{2,\mathrm{loc}}(Q_T)$.

Suppose that $\tilde{u}^\epsilon \in W^{1,2}_{2,\mathrm{loc}}(Q_T)$ is another strong solution to \eqref{eq:regularized} satisfying $|\tilde{u}^\epsilon(t,x)-x|\leq M(T-t)$ and $\partial_x \tilde{u}^\epsilon \in L^\infty(Q_T)$. Define $\varphi^\epsilon := \tilde{u}^\epsilon - u^{\varepsilon} \in W_{2,\mathrm{loc}}^{1,2}(Q_T)$. Then $\varphi^\epsilon$ is a strong solution of the  equation:
\begin{align}\label{eq:varphi:epsilon}
\frac{\partial \varphi^\epsilon}{\partial t} + D^{\varepsilon}(t, \tilde{u}^\epsilon; m)\frac{\partial^2 \varphi^\epsilon}{\partial x^2} - V^{\varepsilon}(t, \tilde{u}^\epsilon; m)\frac{\partial \varphi^\epsilon}{\partial x} = F_\varphi^\epsilon,
\end{align}
where
\begin{align*}
F_\varphi^\epsilon := \big[D^{\varepsilon}(t, u^{\varepsilon}; m) - D^{\varepsilon}(t, \tilde{u}^\epsilon; m)\big]\frac{\partial^2 u^{\varepsilon}}{\partial x^2} - \big[V^{\varepsilon}(t, u^{\varepsilon}; m) - V^{\varepsilon}(t, \tilde{u}^\epsilon; m)\big]\frac{\partial u^{\varepsilon}}{\partial x}.
\end{align*}
Recall that $D^{\varepsilon}(t, \cdot; m)$ and $V^{\varepsilon}(t, \cdot; m)$ are Lipschitz continuous with some constants $L_{D,\epsilon}, L_{V,\epsilon} > 0$, and $\left\|\frac{\partial u^{\varepsilon}}{\partial x}\right\|_\infty=M_1$ and $\left\|\frac{\partial^2 u^{\varepsilon}}{\partial x^2}\right\|_\infty= M_2^{\varepsilon} $. Then 
\begin{align*}
|F_\varphi^\epsilon| \le (L_{D,\epsilon} M_2^{\varepsilon} + L_{V,\epsilon} M_1)|\varphi^\epsilon| =: C_1(\epsilon)|\varphi^\epsilon|.
\end{align*}
Moreover, by virtue of $|u^\epsilon(t,x) - x| \le M(T-t)$ and $|\tilde{u}^\epsilon(t,x) - x| \le M(T-t)$, together with  $\partial_x u^\epsilon \in L^\infty(Q_T)$ and $\partial_x \tilde{u}^\epsilon \in L^\infty(Q_T)$, it holds that $\varphi^\epsilon\in L^\infty(Q_T)$ and $\partial_x \varphi^\epsilon \in L^\infty(Q_T)$. 

Recalling $\phi(x) = e^{-\sqrt{1+x^2}}$, we introduce a smooth cut-off function $\zeta_n(x)$ for positive integers $n$, such that $\zeta_n(x) = 1$ for $|x| \le n$, $\zeta_n(x) = 0$ for $|x| \ge 2n$, and $|\zeta_n'(x)| \le 2/n$. Multiplying both sides of  \eqref{eq:varphi:epsilon} by $\zeta_n \phi^2 \varphi^\epsilon$ and integrating over $[t, T] \times \mathbb{R}$, we obtain
\begin{align}
\int_t^T \int_{\mathbb{R}} \zeta_n \phi^2 \varphi^\epsilon \frac{\partial \varphi^\epsilon}{\partial s} \md x \md s + &\int_t^T \int_{\mathbb{R}} \zeta_n \phi^2 D^{\varepsilon} \varphi^\epsilon \frac{\partial^2 \varphi^\epsilon}{\partial x^2} \md x \md s - \int_t^T \int_{\mathbb{R}} \zeta_n \phi^2 V^{\varepsilon} \varphi^\epsilon \frac{\partial \varphi^\epsilon}{\partial x} \md x \md s\nonumber \\
&= \int_t^T \int_{\mathbb{R}} \zeta_n \phi^2 F_\varphi^\epsilon \varphi^\epsilon \md x \md s.\label{eq:zeta:n}
\end{align}
Employing a smooth approximation argument similar to that of Lemma \ref{lma:v:derivative}, we may apply the chain rule to obtain $\varphi^\epsilon \frac{\partial \varphi^\epsilon}{\partial s} = \frac{1}{2}\frac{\partial}{\partial s}(\varphi^\epsilon)^2$  a.e. Moreover, by Lemma \ref{lem:acl} (ii), for a.e. $x \in \mathbb{R}$, $\left(\varphi^\epsilon(\cdot, x)\right)^2$ is absolutely continuous on $[0,T]$. Therefore, using the terminal condition  $\varphi^\epsilon(T, \cdot) = 0$ and the dominated convergence theorem, we obtain
\begin{align*}
\int_t^T \int_{\mathbb{R}} \zeta_n \phi^2 \varphi^\epsilon \frac{\partial \varphi^\epsilon}{\partial s} \md x \md s &= \frac{1}{2} \int_{\mathbb{R}} \zeta_n \phi^2 \left( \int_t^T \frac{\partial}{\partial s}\big(\varphi^\epsilon(s,x)\big)^2 \md s \right) \md x \\
&= \frac{1}{2} \int_{\mathbb{R}} \zeta_n \phi^2 \left[ \varphi^\epsilon(T,x)^2 - \varphi^\epsilon(t,x)^2 \right] \md x \\
&= - \frac{1}{2} \int_{\mathbb{R}} \zeta_n \phi^2 \varphi^\epsilon(t,x)^2 \md x\to - \frac{1}{2} \int_{\mathbb{R}}  \phi^2 \varphi^\epsilon(t,x)^2 \md x\quad \text{as }n\to \infty.
\end{align*}
Furthermore, by Lemma \ref{lem:acl}(i),  both $\varphi^\epsilon(t, \cdot)$ and $\frac{\partial \varphi^\epsilon}{\partial x}(t, \cdot)$ are absolutely continuous on  $[-2n, 2n]$  for a.e. $t\in[0,T]$. Combining with the fact that $\zeta_n \in C_c^\infty(\mathbb{R})$ ensures the boundary terms vanish, i.e., $\left[ \zeta_n \phi^2 D^{\varepsilon} \varphi^\epsilon \frac{\partial \varphi^\epsilon}{\partial x} \right]_{-2n}^{2n} = 0$. Using integration by parts yields
\begin{align*}
\int_t^T \int_{\mathbb{R}} \zeta_n \phi^2 D^{\varepsilon} \varphi^\epsilon \frac{\partial^2 \varphi^\epsilon}{\partial x^2} \md x \md s &= - \int_t^T \int_{\mathbb{R}} \frac{\partial}{\partial x}(\zeta_n \phi^2 D^{\varepsilon} \varphi^\epsilon) \frac{\partial \varphi^\epsilon}{\partial x} \md x \md s \\
&= - \int_t^T \int_{\mathbb{R}} \zeta_n \frac{\partial}{\partial x}(\phi^2 D^{\varepsilon} \varphi^\epsilon) \frac{\partial \varphi^\epsilon}{\partial x} \md x \md s - \int_t^T \int_{\mathbb{R}} \zeta_n' \phi^2 D^{\varepsilon} \varphi^\epsilon \frac{\partial \varphi^\epsilon}{\partial x} \md x \md s,
\end{align*}
where the second equality follows from the Leibniz rule for weak derivatives (see, e.g., \citet*[Section 5.2, Theorem 1]{evans2022partial}).

Moreover, by the Leibniz rule, $\frac{\partial}{\partial x}(\phi^2 D^{\varepsilon} \varphi^\epsilon) = \phi^2 D^{\varepsilon} \frac{\partial \varphi^\epsilon}{\partial x} + \left(2\phi\phi' D^{\varepsilon} + \phi^2 \frac{\partial D^{\varepsilon}}{\partial x}\right) \varphi^\epsilon$. This term is bounded and absolutely integrable on $Q_T$ because $\phi$ and $\phi'$ decay exponentially, and $D^{\varepsilon}, \frac{\partial D^{\varepsilon}}{\partial x}, \varphi^\varepsilon, \frac{\partial \varphi^\varepsilon}{\partial x} \in L^\infty(Q_T)$. Consequently,
 \begin{align*}
     \int_t^T \int_{\mathbb{R}} \zeta_n \frac{\partial}{\partial x}(\phi^2 D^{\varepsilon} \varphi^\epsilon) \frac{\partial \varphi^\epsilon}{\partial x} \md x \md s\to \int_t^T \int_{\mathbb{R}}  \frac{\partial}{\partial x}(\phi^2 D^{\varepsilon} \varphi^\epsilon) \frac{\partial \varphi^\epsilon}{\partial x} \md x \md s \quad \text{as }n\to \infty.
 \end{align*}
Moreover,
\begin{align*}
    \left|\int_t^T \int_{\mathbb{R}} \zeta_n' \phi^2 D^{\varepsilon} \varphi^\epsilon \frac{\partial \varphi^\epsilon}{\partial x} \md x \md s\right|\leq \frac{2}{n} \int_t^T \int_{\mathbb{R}} \phi^2 \left|D^{\varepsilon} \varphi^\epsilon \frac{\partial \varphi^\epsilon}{\partial x}\right| \md x \md s\to 0\quad \text{as }n\to \infty.
\end{align*}
Leveraging the dominated convergence theorem and sending  $n \to \infty$ in \eqref{eq:zeta:n} yield
\begin{align}\label{eq:zeto:to:infty}
- \frac{1}{2} \int_{\mathbb{R}} \phi^2 \varphi^\epsilon(t,x)^2 \md x &- \int_t^T \int_{\mathbb{R}} \frac{\partial}{\partial x}(\phi^2 D^{\varepsilon} \varphi^\epsilon) \frac{\partial \varphi^\epsilon}{\partial x} \md x \md s - \int_t^T \int_{\mathbb{R}} \phi^2 V^{\varepsilon} \varphi^\epsilon \frac{\partial \varphi^\epsilon}{\partial x} \md x \md s = \int_t^T \int_{\mathbb{R}} \phi^2 F_\varphi^\epsilon \varphi^\epsilon \md x \md s.
\end{align}
Substituting $\frac{\partial}{\partial x}(\phi^2 D^{\varepsilon} \varphi^\epsilon) = \phi^2 D^{\varepsilon} \frac{\partial \varphi^\epsilon}{\partial x} + \left(2\phi\phi' D^{\varepsilon} + \phi^2 \frac{\partial D^{\varepsilon}}{\partial x}\right) \varphi^\epsilon$   into \eqref{eq:zeto:to:infty} and using $D^\epsilon\geq\kappa$, we get that
\begin{align*}
&\frac{1}{2} \int_{\mathbb{R}} \phi^2 \varphi^\epsilon(t,x)^2 \md x +\kappa \int_t^T \int_{\mathbb{R}} \phi^2 \left|\frac{\partial \varphi^\epsilon}{\partial x}\right|^2 \md x \md s \\
\le & \ - \int_t^T \int_{\mathbb{R}} \left(2\frac{\phi'}{\phi} D^{\varepsilon} + \frac{\partial D^{\varepsilon}}{\partial x} + V^{\varepsilon} \right) \phi^2 \varphi^\epsilon \frac{\partial \varphi^\epsilon}{\partial x} \md x \md s  - \int_t^T \int_{\mathbb{R}} \phi^2 F_\varphi^\epsilon \varphi^\epsilon \md x \md s\\
 \le & \ C_2(\epsilon) \int_t^T \int_{\mathbb{R}} \phi^2 |\varphi^{\epsilon}| \left|\frac{\partial \varphi^{\epsilon}}{\partial x}\right| \, \md x \, \md s + C_1(\epsilon) \int_t^T \int_{\mathbb{R}} \phi^2 (\varphi^\epsilon)^2 \, \md x \, \md s  \\
\le & \ \frac{\rho}{2} \int_t^T \int_{\mathbb{R}} \phi^2 \left|\frac{\partial \varphi^\epsilon}{\partial x}\right|^2 \, \md x \, \md s + \left( \frac{(C_2(\epsilon))^2}{2\rho} + C_1(\epsilon) \right) \int_t^T \int_{\mathbb{R}} \phi^2 (\varphi^{\epsilon})^2\, \md x \, \md s,
\end{align*}
where the second inequality follows from the bounds
 $\left|\frac{\phi'}{\phi}\right| \le 1$ and $D^{\varepsilon}, \partial_x D^{\varepsilon}, V^{\varepsilon} \in L^\infty(Q_T)$, which imply that $2\frac{\phi'}{\phi} D^{\varepsilon} + \frac{\partial D^{\varepsilon}}{\partial x} + V^{\varepsilon}$ is bounded by some constant $C_2(\epsilon) > 0$, and  the estimate $|F_\varphi^\epsilon| \le C_1(\epsilon)|\varphi^\epsilon|$. The last inequality follows from Young's inequality.

Let $C_3(\epsilon) := \frac{(C_2(\epsilon))^2}{\rho} + 2C_1(\epsilon)$. Then
\begin{align*}
\int_{\mathbb{R}}\phi^2 \varphi^\epsilon(t, x)^2 \md x \le C_3(\epsilon) \int_t^T \left( \int_{\mathbb{R}}\phi^2 \varphi^\epsilon(s, x)^2 \md x \right) \md s.
\end{align*}
Using the backward Gronwall inequality, we have
\begin{align*}
\int_t^T \int_{\mathbb{R}}\phi^2 \varphi^\varepsilon(s, x)^2 \md x \md s = 0 \quad \text{for all } t \in [0, T].
\end{align*}
As $\phi > 0$, it follows that $\varphi^\epsilon \equiv 0$, and hence $\tilde{u}^\epsilon = u^{\varepsilon}$.
\end{proof}

\section{A Technical Lemma for Proposition \ref{lem:limit-sde}}\label{sec:AppC}
\begin{lemma}\label{zero-measure}
Recall the processes $\tilde{x}_{a,k}$ and $\tilde{x}_{a,*}$, $a\in\{1,2\}$, constructed in the proof of Proposition \ref{lem:limit-sde}. Let $\tilde {\mathcal{N}}_* \subset [0,T] \times \mathbb{R}$ be a set of Lebesgue measure zero. Then, for each $a\in\{1,2\}$, the set $\{(s,\tilde{\omega}) \in [0,T]\times \tilde{\Omega} : (s, \tilde{x}_{a,*}(s, \tilde{\omega})) \in \tilde {\mathcal{N}}_*\}$ has zero $\md t \otimes \md \tilde {\mathbb{P}}$ measure.
\end{lemma}

\begin{proof}
Without loss of generality, let us fix $a=1$.
For any integer $n \ge 1$, let $U_n \subset [0,T] \times \mathbb{R}$ be an open set such that $\tilde {\mathcal{N}}_* \subset U_n$ and $\text{Leb}(U_n) < 1/n$.  
Similar to the proof of Proposition \ref{FBSDE-existence}, we  have
\begin{align*}
    \tilde{\mathbb{E}}\left[\int_0^T \mathbbm{1}_{U_n}(t, \tilde{x}_{1,k}(t)) \md t\right] \leq   C\|\mathbbm{1}_{U_n}\|_{L^2([0,T] \times \mathbb{R})} < C \sqrt{1/n},
\end{align*}
where the constant $C$ depends only on $c_3$, $C_2$ and $T$.

The indicator function $\mathbbm{1}_{U_n}(t, x)$ is lower semi-continuous because $U_n$ is open. Moreover, we have $\liminf\limits_{k \to \infty} \mathbbm{1}_{U_n}(t, \tilde{x}_{1,k}(t)) \ge \mathbbm{1}_{U_n}(t, \tilde{x}_{1,*}(t))$ in view of $\tilde{x}_{1,k} \to \tilde{x}_{1,*}$ a.s. Applying Fatou's lemma yields
\begin{align*}
\tilde{\mathbb{E}}\left[\int_0^T \mathbbm{1}_{U_n}(t, \tilde{x}_{1,*}(t)) \md t\right] \le \liminf_{k\to\infty} \tilde{\mathbb{E}}\left[\int_0^T \mathbbm{1}_{U_n}(t, \tilde{x}_{1,k}(t)) \md t\right] \le C \sqrt{1/n}.
\end{align*}
As $\tilde {\mathcal{N}}_* \subset U_n$, it holds that $\mathbbm{1}_{\tilde {\mathcal{N}}_*} \le \mathbbm{1}_{U_n}$. Letting $n \to \infty$, we get
$\tilde{\mathbb{E}}\left[\int_0^T \mathbbm{1}_{\tilde {\mathcal{N}}_*}(t, \tilde{x}_{1,*}(t)) \md t\right] = 0$.
The conclusion then follows.
\end{proof}
\end{appendices}

\vspace{0.3in}\noindent
\textbf{Acknowledgements}: Zongxia  Liang is supported by the National Natural Science Foundation of China under grant no. 12271290. Xiang Yu is supported by the Hong Kong RGC General Research Fund (GRF) under grant no. 15214125.

{\small
\bibliographystyle{abbrvnat}
\bibliography{sample}
}

\end{document}